\documentclass[11pt, reqno, psamsfonts]{amsart}
\pdfoutput=1

\usepackage{amssymb}
\usepackage{amsthm}
\usepackage{amsmath}
\usepackage{latexsym}
\usepackage[T1]{fontenc}
\usepackage[utf8]{inputenc}
\usepackage[russian, french, english]{babel}

\usepackage{graphicx}
\usepackage{wrapfig}
\usepackage{mathtools}
\usepackage{tikz}
\usepackage{amsbsy}
\usepackage[inline]{enumitem}
\usepackage{mathrsfs}
\usetikzlibrary{shapes,snakes}
\usetikzlibrary{arrows.meta}
\usetikzlibrary{decorations.pathmorphing}
\usetikzlibrary{patterns}
\usetikzlibrary{calc}
\usepackage{float}
\usepackage{array}
\usepackage{subcaption}
\usepackage{multicol}
\usepackage{stmaryrd}
\usepackage{cancel} 
\usepackage{algorithm}
\usepackage{algpseudocode}
\makeatletter
\algnewcommand{\LineComment}[1]{\Statex \hskip\ALG@thistlm {\color{gray}\texttt{// #1}}}
\makeatother

\usepackage{pdflscape}

\usepackage{color, colortbl}
\definecolor{almond}{rgb}{0.94, 0.87, 0.8}
\definecolor{bubbles}{rgb}{0.91, 1.0, 1.0}

\counterwithin{algorithm}{section}
\usepackage[justification=centering, labelfont=bf]{caption}

\usepackage{lmodern}
\usepackage{mathabx}

\usepackage{upgreek}
\usepackage{titlesec}
\usepackage{bbm}
\usepackage[spacing=true,kerning=true,babel=true,tracking=true]{microtype}
\usepackage[shortcuts]{extdash}
\usepackage[foot]{amsaddr}
\usepackage[left=1in,right=1in,top=0.91in,bottom=0.91in,bindingoffset=0cm]{geometry}
\usepackage{bm}
\usepackage{centernot}
\usepackage{mdframed}
\usepackage{hyperref}
\usepackage{mathdots}
\usepackage{xspace}
\hypersetup{
    colorlinks=true,
    linkcolor=blue,
    filecolor=magenta,      
    urlcolor=red,
    citecolor=cyan,
}
\allowdisplaybreaks

\usepackage{makecell}

\usepackage[skins]{tcolorbox}

\usepackage[backend=biber, style=alphabetic, sorting=nyt, maxnames=100,backref=true, sortcites = true]{biblatex}
\title{A linear-time algorithm for $(1+\eps)\Delta$-edge-coloring}
\date{}
\author{\lsstyle Anton~Bernshteyn}
\email{bernshteyn@math.ucla.edu}
\author{\lsstyle Abhishek~Dhawan}
\email{adhawan2@illinois.edu}
\address{\normalfont{}(AB) \textls{Department of Mathematics, University of California, Los Angeles, CA, USA}}
\address{\normalfont{}(AD) \textls{Department of Mathematics, University of Illinois Urbana-Champaign, Urbana, IL, USA}}
\thanks{AB's research was partially supported by the NSF grant DMS-2045412 and the NSF CAREER grant DMS-2239187.
AD's research was partially supported by the Georgia Tech ARC-ACO Fellowship, NSF grant DMS-2053333 (PI: Cheng Mao), the NSF CAREER grant DMS-2239187 (PI: Anton Bernshteyn), and the NSF RTG grant DMS-1937241.}

\newtheoremstyle{bfnote}%
{}{}%
{\slshape}{}%
{\bfseries}{\bfseries.}%
{ }%
{\thmname{#1}\thmnumber{ #2}\thmnote{ \ep{\normalfont{}#3}}}

\newtheoremstyle{claim}%
{}{}%
{\slshape}{}%
{\itshape}{.}%
{ }%
{\thmname{#1}\thmnumber{ #2}\thmnote{ \ep{\normalfont{}#3}}}

\theoremstyle{bfnote}
\newtheorem{theo}{Theorem}[section]
\newtheorem*{theo*}{Theorem}
\newtheorem{prop}[theo]{Proposition}
\newtheorem{Lemma}[theo]{Lemma}

\newtheorem{conj}[theo]{Conjecture}

\newtheorem*{corl*}{Corollary}

\theoremstyle{definition}
\newtheorem{defn}[theo]{Definition}
\newtheorem*{defn*}{Definition}

\newtheorem*{exmp*}{Example}

\theoremstyle{remark}
\newtheorem*{ques*}{Question}
\newtheorem*{remk*}{Remark}

\theoremstyle{claim}

\newcounter{ForClaims}[section]
\newtheorem{claim}{Claim}[ForClaims]

\newtheorem*{claim*}{Claim}

\makeatletter
\newcommand{\neutralize}[1]{\expandafter\let\csname c@#1\endcsname\count@}
\makeatother

\newenvironment{claimproof}{\noindent$\rhd$\hspace{1em}}{\hfill$\lhd$}

\newcommand{\0}{\varnothing}
\newcommand{\set}[1]{\{#1\}}
\newcommand{\N}{{\mathbb{N}}}
\newcommand{\Z}{\mathbb{Z}}

\renewcommand{\P}{\mathbb{P}}
\newcommand{\E}{\mathbb{E}}
\renewcommand{\epsilon}{\varepsilon}
\newcommand{\eps}{\epsilon}
\renewcommand{\phi}{\varphi}
\renewcommand{\theta}{\vartheta}
\renewcommand{\leq}{\leqslant}
\renewcommand{\geq}{\geqslant}
\newcommand{\defeq}{\coloneqq}

\newcommand{\bemph}[1]{{\normalfont#1}} 
\newcommand{\ep}[1]{\bemph{(}#1\bemph{)}} 

\newcommand{\emphdef}[1]{\textbf{\textit{{#1}}}}
\newcommand{\blank}{\mathsf{blank}}
\newcommand{\col}{\mathsf{col}}
\newcommand{\dom}{\mathsf{dom}}
\newcommand{\Shift}{\mathsf{Shift}}
\newcommand{\vend}{\mathsf{vEnd}}
\newcommand{\vstart}{\mathsf{vStart}}

\numberwithin{equation}{section}
\newcommand{\emphd}[1]{\emphdef{#1}}

\newcommand{\poly}{\mathsf{poly}}
\newcommand{\Start}{\mathsf{Start}}
\newcommand{\End}{\mathsf{End}}
\newcommand{\Pivot}{\mathsf{Pivot}}
\newcommand{\length}{\mathsf{length}}
\newcommand{\visited}{\mathsf{visited}}

\newcommand{\aug}{\mathsf{Aug}}

\newcommand{\bbone}{\mathbbm{1}}

\newcommand{\IE}{E_{\mathsf{int}}}
\newcommand{\IV}{V_{\mathsf{int}}}
\newcommand{\val}{\mathsf{val}}
\newcommand{\wt}{\mathsf{wt}}
\newcommand{\kmax}{k_{\max}}

\newcommand{\boldcol}{\mathbf{col}}
\newcommand{\same}{s}
\newcommand{\y}{\tilde y}

\newcommand{\algsize}{\small}

\makeatletter
\newenvironment{breakablealgorithm}
  {
   \begin{center}
     \refstepcounter{algorithm}
     \hrule height.8pt depth0pt \kern2pt
     \renewcommand{\caption}[2][\relax]{
       {\raggedright\textbf{\ALG@name~\thealgorithm} ##2\par}%
       \ifx\relax##1\relax 
         \addcontentsline{loa}{algorithm}{\protect\numberline{\thealgorithm}##2}%
       \else 
         \addcontentsline{loa}{algorithm}{\protect\numberline{\thealgorithm}##1}%
       \fi
       \kern2pt\hrule\kern2pt
     }
  }{
     \kern2pt\hrule\relax
   \end{center}
  }
\makeatother

\titleformat{\section}[block]{\scshape}{\thesection.}{1ex}{}
\titleformat{\subsection}[block]{\bfseries}{\thesubsection.}{1ex}{}
\titleformat{\subsection}[block]{\bfseries}{\thesubsection.}{1ex}{}
\titleformat{\subsubsection}[runin]{\bfseries}{\bfseries\upshape\thesubsubsection.}{1ex}{}[.---]

\titlespacing*{\section}{0pt}{*3}{*1}
\titlespacing*{\subsection}{0pt}{*3}{*1}
\titlespacing*{\subsubsection}{0pt}{*1.5}{*0}

\renewbibmacro{in:}{}

\renewbibmacro*{volume+number+eid}{%
	\printfield{volume}%
	\setunit*{\addnbspace}
	\printfield{number}%
	\setunit{\addcomma\space}%
	\printfield{eid}}

\DeclareFieldFormat[article]{volume}{\textbf{#1}\space}
\DeclareFieldFormat[article]{number}{\mkbibparens{#1}}

\DeclareFieldFormat{journaltitle}{#1,}
\DeclareFieldFormat[thesis]{title}{\mkbibemph{#1}\addperiod}
\DeclareFieldFormat[article, unpublished, thesis]{title}{\mkbibemph{#1},}
\DeclareFieldFormat[book]{title}{\mkbibemph{#1}\addperiod}
\DeclareFieldFormat[unpublished]{howpublished}{#1, }

\DeclareFieldFormat{pages}{#1}

\DeclareFieldFormat[article]{series}{Ser.~#1\addcomma}

\setlist{topsep=3pt,itemsep=3pt}

\makeatletter

\usepackage{environ}

\NewEnviron{problem}[1]{%
\begin{center}\fbox{\parbox{4in}{%
    {\centering\scshape #1\par}%
    \parskip=1ex
    \everypar{\hangindent=1em}%
    \BODY
}}\end{center}}

\usepackage{framed}
\renewenvironment{shaded}{%
  \MakeFramed{\advance\hsize-\width \FrameRestore\FrameRestore}}%
 {\endMakeFramed}
\definecolor{shadecolor}{gray}{0.9}

\begin{document}

\maketitle

\begin{abstract}
    We present a randomized algorithm that, given a constant $\epsilon > 0$, outputs a proper $(1+\epsilon)\Delta$-edge-coloring of an $m$-edge simple graph $G$ of maximum degree $\Delta \geq 1/\epsilon$ in $O(m)$ time with high probability. This is the first linear-time algorithm for this problem covering the full range of possible values of $\Delta$. 
    Indeed, even for edge-coloring with $2\Delta - 1$ colors (i.e., meeting the ``greedy'' bound), no such linear-time algorithm has been previously known.
\end{abstract}

\section{Introduction}

    \subsection{The main result}

    All graphs in this paper are finite, undirected, and simple. Edge-coloring is one of the central and most well-studied problems in graph theory \cites[\S17]{BondyMurty}[\S5.3]{Diestel}{EdgeColoringMonograph}. In what follows, we write $\N \defeq \set{0,1,\ldots}$, $\N^+ \defeq \set{1,2,\ldots}$, and, given $q \in \N$, we let $[q] \defeq \set{1,2,\ldots, q}$.

    \begin{defn}[Proper edge-coloring]
        Two edges $e$, $f \in E(G)$ of a graph $G$ are \emphd{adjacent}, or \emphd{neighbors}, if $e \neq f$ and $e \cap f \neq \0$. An \emphd{edge-coloring} of $G$ is a mapping $\phi \colon E(G) \to \N^+$, and a \emphd{$q$-edge-coloring} for $q \geq 0$ is an edge-coloring $\phi$ with $\phi(e) \leq q$ for all $e \in E(G)$.\footnote{Note that, for convenience, we allow $q$ to be non-integer.} We refer to the value $\phi(e)$ as the \emphd{color} of the edge $e$. An edge-coloring $\phi$ is \emphd{proper} if $\phi(e) \neq \phi(f)$ whenever $e$ and $f$ are neighbors. The \emphd{chromatic index} of $G$, denoted by $\chi'(G)$, is the minimum $q$ such that $G$ admits a proper $q$-edge-coloring.
    \end{defn}

    In a graph $G$ of maximum degree $\Delta$ every edge is adjacent to at most $2\Delta - 2$ other edges, so a simple greedy algorithm yields the bound $\chi'(G) \leq (2\Delta - 2) + 1 = 2\Delta - 1$. For $\Delta \geq 3$, this can be improved by a factor of roughly $2$:

    \begin{theo}[{Vizing \cite{Vizing}}]\label{theo:Vizing}
        If $G$ is a graph of maximum degree $\Delta$, then $\chi'(G) \leq \Delta + 1$.
    \end{theo}

    See \cite[\S{}A.1]{EdgeColoringMonograph} for an English translation of \cite{Vizing} and \cites[\S17.2]{BondyMurty}[\S5.3]{Diestel} for textbook presentations. Since the edges incident to a vertex of $G$ must receive distinct colors in a proper edge-coloring, we have an almost matching lower bound $\chi'(G) \geq \Delta$. By a result of Holyer \cite{Holyer}, distinguishing between the cases $\chi'(G)=\Delta$ and $
    \chi'(G) = \Delta + 1$ is \textsf{NP}-hard, even when $\Delta = 3$.

    In this paper, we investigate the natural algorithmic problem of efficiently generating a proper $q$-edge-coloring of a given graph $G$ for $q \geq \Delta + 1$. In the following discussion, we use $n$ and $m$ for the number of vertices and edges of $G$ respectively. The benchmark results in this line of inquiry are summarized in Table \ref{table:history}. When $q = \Delta + 1$ (i.e., if we are trying to construct a coloring that matches Vizing's bound), the best known running time of $O(m \log \Delta)$ was achieved very recently by Assadi, Behnezhad, Bhattacharya, Costa, Solomon, and Zhang \cite{assadi2024vizing}.\footnote{We remark that an earlier version of this paper appeared several months prior to \cite{assadi2024vizing} and is cited there as giving the fastest algorithm that was known at the time for the regime when $\Delta$ is much smaller than $n$.} 
    On the other hand, it is clear that any edge-coloring algorithm requires at least $\Omega(m)$ time. In other words, the result of \cite{assadi2024vizing} is at most a $\log \Delta$ factor away from being optimal. A natural question arises: Can the $\log\Delta$ factor be removed, i.e., is there a \emph{truly linear-time algorithm} for $(\Delta+1)$-edge-coloring? More generally: 

    \begin{shaded}
        \begin{quote}
            \textsl{For which values of $q$ does there exist an edge-coloring algorithm using $q$ colors with \textbf{linear}---i.e., $O(m)$---running time?}
        \end{quote}
    \end{shaded}

     {
		\renewcommand{\arraystretch}{1.3}
		\begin{table}[t]\small
			\begin{tabular}{| c | c | c | >{\centering\arraybackslash\scriptsize}m{0.25\textwidth} |}
				\hline
				\thead{\#{} Colors} & \thead{Runtime} & \thead{Deterministic?} & \thead{References} \\\hline\hline
                 $\Delta + 1$ & $O(mn)$ & Deterministic & {Vizing \cite{Vizing}, Bollob\'as \cite[94]{Bollobas}, Rao--Dijkstra \cite{RD},  Misra--Gries \cite{MG}} \\\hline
                $\Delta + 1$ & $O(m\sqrt{n \log n})$ & Deterministic & {Arjomandi \cite{Arjomandi}\footnotemark, Gabow--Nishizeki--Kariv--Leven--Terada~\cite{GNKLT}} \\\hline
                $\Delta + 1$ & $O(m\sqrt{n})$ & Deterministic & {Sinnamon \cite{Sinnamon}} \\\hline
                $\Delta + 1$ & $\tilde{O}(mn^{1/3})$ & Randomized & {Bhattacharya--Carmon--Costa--
                Solomon--Zhang \cite{BCCSZ_cube_root}} \\\hline
                $\Delta + 1$ & expected $O(n^2 \log n)$ & Randomized & {Assadi \cite{Assadi}} \\\hline
                 $\Delta + 1$ & $O(m\,\Delta \log n)$ & Deterministic & {Gabow--Nishizeki--Kariv--Leven--Terada~\cite{GNKLT}} \\\hline
                $\Delta + 1$ & $O(m\,\Delta^{17})$ & Randomized & {AB--AD \cite{bernshteyn2023fast}\footnotemark} \\\hline
                  $\Delta + 1$ & $O(m\,\log\Delta)$ & Randomized & Assadi--Behnezhad--Bhattacharya--Costa--Solomon--Zhang \cite{assadi2024vizing} \\\hline\hline
                $2\Delta - 1$ & $O(m\,\Delta)$ & Deterministic & {Folklore (greedy)} \\\hline
                 $2\Delta - 1$ & $O(m\,\log\Delta)$ & Deterministic & {Sinnamon \cite{Sinnamon}, Bhattacharya--Chakrabarty--Henzinger--Nanongkai \cite{BCHN_dynamic}} 
                 \\\hline\hline
             $(1+\epsilon)\Delta$ & $O(m\,\log^6 n/\epsilon^2)$ for $\Delta = \Omega(\log n/\epsilon)$ & Randomized & {Duan--He--Zhang \cite{DHZ_dynamic}} \\\hline
                 $(1+\epsilon)\Delta$ & $O(m\,\log n/\epsilon)$ & Deterministic & {Elkin--Khuzman \cite{EK_splitting}} \\\hline
                 $(1+\epsilon)\Delta$ & $O(\max \{m/\epsilon^{18}, \ m \,\log\Delta\})$ & Randomized & {Elkin--Khuzman \cite{EK_splitting}} \\\hline
                 $(1+\epsilon)\Delta$ & \makecell{$O(m \,\log(1/\epsilon)/\epsilon^2)$ for $\Delta \geq (\log n / \epsilon)^{\poly(1/\epsilon)}$} & Randomized & {Bhattacharya--Costa--
                 Panski--Solomon \cite{BCPS_nibble}} \\\hline
                 $(1+\epsilon)\Delta$ & \makecell{expected $O(m \,\log(1/\epsilon))$ for $\Delta \gg \log n / \epsilon$} & Randomized & {Assadi \cite{Assadi}} \\\hline\hline
                 \rowcolor{almond} $(1+\epsilon)\Delta$ & $O(m \,\log(1/\eps)/\epsilon^4)$ & Randomized & {\textbf{This paper}} \\\hline
			\end{tabular}
            \caption{A brief survey of edge-coloring algorithms. The stated runtime of randomized algorithms is attained with high probability, unless indicated otherwise. 
            }\label{table:history}
		\end{table}	
	}

    \footnotetext[3]{See the comments in \cite[41]{GNKLT} regarding a possible gap in Arjomandi's analysis.}
    \footnotetext{The bound explicitly stated in \cite{bernshteyn2023fast} is $O(n\,\Delta^{18})$, but in fact the argument there gives $O(m \, \Delta^{17})$.}

    Perhaps surprisingly, heretofore no such linear-time algorithm has been known even for edge-coloring with $2\Delta-1$ colors! As mentioned earlier, a proper $(2\Delta - 1)$-edge-coloring of $G$ can be found by a straightforward greedy algorithm, where the edges are processed one by one, with each edge receiving an arbitrary color not yet used on any of its neighbors. At first glance, this may seem like a linear-time algorithm, but in fact its running time is only $O(m \,\Delta)$, because finding an available color for an edge requires surveying its neighbors and takes $O(\Delta)$ time. This bound has been improved to $O(m\,\log\Delta)$ by Sinnamon \cite{Sinnamon} (it is also implicit in \cite{BCHN_dynamic} by Bhattacharya, Chakrabarty, Henzinger, and Nanongkai), which has remained the state of the art prior to our work.

    That being said, there do exist known linear-time edge-coloring algorithms that operate under additional assumptions on $\Delta$. For example, when $\Delta$ is a \emph{constant}, a proper $(\Delta+1)$-edge-coloring can be found in linear time (this was first shown by the authors in \cite{bernshteyn2023fast}). On the other hand, for constant $\epsilon > 0$ and \emph{sufficiently large} $\Delta$, there are linear-time algorithms for $(1+\epsilon)\Delta$-edge-coloring. More precisely, for
    \begin{equation}\label{eq:Delta_lower_bound}
        \Delta \,\geq\, \left(\frac{\log n}{\epsilon}\right)^{\Omega(\log(1/\epsilon)/\epsilon)},
    \end{equation}
    Bhattacharya, Costa, Panski, and Solomon \cite{BCPS_nibble} gave a $(1+\epsilon)\Delta$-edge-coloring algorithm with running time $O(m \,\log(1/\epsilon)/\epsilon^2)$, which is linear in $m$ when $\epsilon = \Theta(1)$.

    Bhattacharya et al.~\cite[\S1.4]{BCPS_nibble} described it as a ``challenging open question'' to improve the lower bound \eqref{eq:Delta_lower_bound} on $\Delta$. A first step toward answering this question was recently taken by Assadi \cite{Assadi}, who designed a randomized $(1+\epsilon)\Delta$-edge-coloring algorithm with expected runtime $O(m \, \log(1/\epsilon))$ under the assumption $\Delta \gg \log n/ \epsilon$. (The algorithm from \cite{BCPS_nibble} attains its stated runtime with high probability, not just in expectation, while to reach a high-probability result using the approach of \cite{Assadi} requires a super-linear runtime.)

    In our main result we fully answer Bhattacharya et al.'s question by completely eliminating the lower bound on $\Delta$. That is, we give a linear-time algorithm for $(1+\epsilon)\Delta$-edge-coloring for any constant $\epsilon > 0$ and with no restrictions on $\Delta$ (except the necessary inequality $\Delta \geq 1/\epsilon$):

    \begin{tcolorbox}
    \begin{theo}[$(1+\epsilon)\Delta$-edge-coloring in linear time]\label{theo:main_theo}
            There is a randomized algorithm that, given $0 < \epsilon < 1$ and a graph $G$ with $m$ edges and maximum degree $\Delta \geq 1/\eps$, outputs a proper $(1 + \eps)\Delta$-edge-coloring of $G$ in time $O(m\,\log (1/\eps) /\epsilon^4)$ with probability at least $1 - \eps^{\Omega(m)}$.
    \end{theo}
    \end{tcolorbox}

    Again, we emphasize that in the range $1 \ll \Delta = O(\log n)$, no linear-time algorithm has previously been available \emph{even for $(2\Delta-1)$-edge-coloring}, let alone for $(1+\epsilon)\Delta$-edge-coloring. Indeed, as we explain below, due to certain fundamental limitations, earlier approaches seem incapable of achieving a linear running time without requiring $\Delta$ to be at least $\Omega(\log n)$.  Reducing the number of colors from $(1+\epsilon)\Delta$ to $\Delta + 1$ remains an interesting (and likely quite difficult) open problem.

    Another respect in which our algorithm outperforms the contributions of Bhattacharya et al. \cite{BCPS_nibble}, Assadi \cite{Assadi}, and Assadi et al.~\cite{assadi2024vizing} is in its failure probability, which is \emph{exponentially small} in $m$. For comparison, Assadi's result \cite{Assadi} only bounds the expected runtime of the algorithm, while the algorithms in \cite{BCPS_nibble,assadi2024vizing} fail with polynomial probability. (Indeed, the only randomized algorithms in Table~\ref{table:history} with exponentially small failure probabilities are the one from \cite{bernshteyn2023fast} and the result of this paper.) 
    Developing {deterministic} linear-time edge-coloring algorithms remains a major open problem. Even with $2\Delta - 1$ colors, the best currently known deterministic algorithm has running time $O(m \,\log \Delta)$ \cite{Sinnamon,BCHN_dynamic}.

    \subsection{Overview of the algorithm and its analysis}\label{section: informal overview}

    Our proof of Theorem~\ref{theo:main_theo} builds on the groundwork laid in our earlier paper \cite{bernshteyn2023fast}. As mentioned previously, that paper provides (among other things) a $(\Delta + 1)$-edge-coloring algorithm with running time $O(m\,\Delta^{17})$. In \cite{EK_splitting}, Elkin and Khuzman used it as a subroutine in order to build proper $(1+\epsilon)\Delta$-edge-colorings for $\epsilon = \Theta(1)$ in time $O(m\,\log\Delta)$. Roughly, their idea is to recursively split the edges of $G$ into subgraphs of progressively smaller maximum degree, eventually reducing the degree to a constant, and then run the algorithm from \cite{bernshteyn2023fast} on each constant-degree subgraph separately. The factor of $\log \Delta$ in the running time arises from the number of degree-splitting steps necessary to bring the maximum degree down to a constant. Hence, such dependence on $\Delta$ seems unavoidable if one wishes to reduce the problem to applying the algorithm from \cite{bernshteyn2023fast} in a black-box manner on bounded degree graphs. 
    
    To achieve a linear runtime independent of $\Delta$, we instead implement (a variant of) the algorithm from \cite{bernshteyn2023fast} on $G$ \emph{directly} and make use of the larger set of available colors---i.e., $(1+\epsilon)\Delta$ rather than $\Delta + 1$---in the analysis. The resulting argument is similar to the one in \cite{bernshteyn2023fast} but considerably more intricate. As the algorithm and its analysis are quite involved, we will ignore a number of technical details in this informal overview (in particular, the actual algorithms and rigorous definitions given in the rest of the paper may slightly diverge from how they are described here). At the same time, we will attempt to intuitively explain the key ideas of the proof and especially highlight the differences between this work and \cite{bernshteyn2023fast}.

    Throughout the rest of the paper, we fix a graph $G$ with $n$ vertices, $m$ edges, and maximum degree $\Delta$, and we write $V \defeq V(G)$ and $E \defeq E(G)$. We shall assume that $\Delta \geq 2$ (edge-coloring graphs of maximum degree $1$ is a trivial problem). In all our algorithms, we tacitly treat $G$ and $\Delta$ as global variables included in the input. We also fix $\epsilon \in (0,1)$ such that $\Delta \geq 1/\epsilon$ and set $q \defeq (1+\epsilon)\Delta$. Without loss of generality, we assume that $q$ is an integer. 
    We call a function $\phi \colon E\to [q]\cup \{\blank\}$ a \emphd{partial $q$-edge-coloring} (or simply a \emphd{partial coloring}) of $G$. Here $\phi(e) = \blank$ indicates that the edge $e$ is uncolored. As usual, $\dom(\phi)$ denotes the \emphd{domain} of $\phi$, i.e., the set of all colored edges.
    
    Most edge-coloring algorithms rely on \emph{augmenting subgraphs}, i.e., certain gadgets that allow increasing the domain of a partial coloring. Indeed, of the algorithms listed in Table~\ref{table:history} with fewer than $2\Delta - 1$ colors, only the ones from \cite{BCPS_nibble} (which uses the nibble method) and \cite{Assadi} (which are based on a particular way to find large matchings) do not invoke augmenting subgraphs. The following formal definition is taken, slightly modified, from \cite[Definition 2.1]{bernshteyn2023fast}.

    \begin{defn}[Augmenting subgraphs]\label{defn:aug}
        Let $\phi \colon E \to [q] \cup \set{\blank}$ be a proper partial $q$-edge-coloring with domain $\dom(\phi) \subset E$. A subgraph $H \subseteq G$ is \emphd{$e$-augmenting} for an uncolored edge $e \in E \setminus \dom(\phi)$ if $e \in E(H)$ and there is a proper partial coloring $\phi'$ with $\dom(\phi') = \dom(\phi) \cup \set{e}$ that agrees with $\phi$ on the edges that are not in $E(H)$; in other words, by only modifying the colors of the edges of $H$, it is possible to add $e$ to the set of colored edges. We refer to such modification operation as \emphd{augmenting} $\phi$ using $H$.
    \end{defn}

    To find a proper $q$-edge-coloring of $G$, we employ the following simple algorithm template:

    {
    \floatname{algorithm}{Algorithm Template}
    \begin{algorithm}[H]\algsize
        \caption{A randomized $q$-edge-coloring algorithm}\label{temp:seq}
        \begin{flushleft}
            \textbf{Input}: A graph $G = (V,E)$ of maximum degree $\Delta$. \\
            \textbf{Output}: A proper $q$-edge-coloring of $G$.
        \end{flushleft}
        \begin{algorithmic}[1]
            \State $\phi \gets$ the empty coloring
            \While{there are uncolored edges}
                \State Pick an uncolored edge $e$ uniformly at random.\label{step:uar}
                \State Find an $e$-augmenting graph $H$.\label{step:H}
                \State Augment $\phi$ using $H$.
            \EndWhile
            \State \Return $\phi$
        \end{algorithmic}
    \end{algorithm}
    }

    The running time of an algorithm that follows Template~\ref{temp:seq} is largely determined by the complexity of finding an $e$-augmenting subgraph $H$ on Step~\ref{step:H} and then augmenting $\phi$ using $H$. In practice, both of these operations are typically performed in time proportional to the number of edges in $H$. Thus, to make the algorithm faster, we need to be able to construct, given an uncolored edge $e$, a ``small'' $e$-augmenting subgraph. This is achieved using a modified version of the procedure from \cite{bernshteyn2023fast} called the \textsf{Multi-Step Vizing Algorithm}. Given a partial coloring $\phi$ and an uncolored edge $e$, it outputs an $e$-augmenting subgraph $H$ of a special form, called a \emphd{multi-step Vizing chain}, which is assembled of a sequence of \emphd{fans}---i.e., sets of edges incident to a common vertex---and \emphd{alternating paths}---i.e., paths whose edge colors form the sequence $\alpha$, $\beta$, $\alpha$, $\beta$, \ldots{} for some $\alpha$, $\beta \in [q]$; see Fig.~\ref{fig:MSVC_illustration} for an illustration. \emph{One-step} Vizing chains (i.e., those consisting of a single fan and an alternating path) were introduced by Vizing in the seminal paper \cite{Vizing} for the purpose of proving Theorem~\ref{theo:Vizing} and were used in a 
    variety of edge-coloring algorithms \cites[94]{Bollobas}{Arjomandi}{GNKLT}{RD}{MG}{DHZ_dynamic}{SV}{Sinnamon}{BCCSZ_cube_root}. In their breakthrough paper \cite{GP}, Greb\'ik and Pikhurko analyzed \emph{two-step} Vizing chains, and the general multi-step version for $(\Delta + 1)$-edge-coloring was developed in \cite{VizingChain} by the first author. (A variant of multi-step Vizing chains for $(1+\eps)\Delta$-edge coloring appeared in the earlier work of Duan, He, and Zhang \cite{DHZ_dynamic} on dynamic edge-coloring, although they did not use this terminology; see also \cite{dhawan2024fast} for an adaptation of their approach to static edge-coloring.) Since then, this concept proved extremely useful in the study of ``constructive'' facets of edge-coloring \cite{bernshteyn2023fast,GmeasVizing,SubexpVizing,dhawan2024edge,Christ}. 

    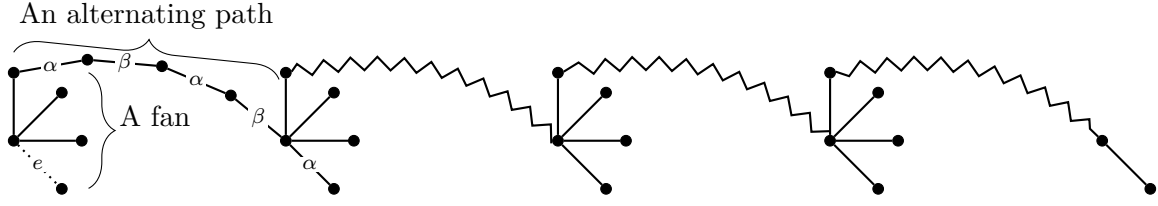
\begin{figure}[t]
	\centering
	\begin{tikzpicture}[scale=0.9]
        \node[circle,fill=black,draw,inner sep=0pt,minimum size=4pt] (a) at (0,0) {};
        	\path (a) ++(-45:1) node[circle,fill=black,draw,inner sep=0pt,minimum size=4pt] (b) {};
        	\path (a) ++(0:1) node[circle,fill=black,draw,inner sep=0pt,minimum size=4pt] (c) {};
        	\path (a) ++(45:1) node[circle,fill=black,draw,inner sep=0pt,minimum size=4pt] (d) {};
        	\path (a) ++(90:1) node[circle,fill=black,draw,inner sep=0pt,minimum size=4pt] (e) {};

                \path (e) ++(10:1.1) node[circle,fill=black,draw,inner sep=0pt,minimum size=4pt] (w) {};
                \path (w) ++(-5:1.1) node[circle,fill=black,draw,inner sep=0pt,minimum size=4pt] (x) {};
                \path (x) ++(-23:1.1) node[circle,fill=black,draw,inner sep=0pt,minimum size=4pt] (y) {};

            \path (c) ++(0:3) node[circle,fill=black,draw,inner sep=0pt,minimum size=4pt] (f) {};
            
        	\path (f) ++(-45:1) node[circle,fill=black,draw,inner sep=0pt,minimum size=4pt] (g) {};
        	\path (f) ++(0:1) node[circle,fill=black,draw,inner sep=0pt,minimum size=4pt] (h) {};
        	\path (f) ++(45:1) node[circle,fill=black,draw,inner sep=0pt,minimum size=4pt] (i) {};
        	\path (f) ++(90:1) node[circle,fill=black,draw,inner sep=0pt,minimum size=4pt] (j) {};
        	
        	\path (h) ++(0:3) node[circle,fill=black,draw,inner sep=0pt,minimum size=4pt] (k) {};

            \path (k) ++(-45:1) node[circle,fill=black,draw,inner sep=0pt,minimum size=4pt] (l) {};
        	\path (k) ++(0:1) node[circle,fill=black,draw,inner sep=0pt,minimum size=4pt] (m) {};
        	\path (k) ++(45:1) node[circle,fill=black,draw,inner sep=0pt,minimum size=4pt] (n) {};
        	\path (k) ++(90:1) node[circle,fill=black,draw,inner sep=0pt,minimum size=4pt] (o) {};
        	
        	\path (m) ++(0:3) node[circle,fill=black,draw,inner sep=0pt,minimum size=4pt] (p) {};
        	
        	\path (p) ++(-45:1) node[circle,fill=black,draw,inner sep=0pt,minimum size=4pt] (q) {};
        	\path (p) ++(0:1) node[circle,fill=black,draw,inner sep=0pt,minimum size=4pt] (r) {};
        	\path (p) ++(45:1) node[circle,fill=black,draw,inner sep=0pt,minimum size=4pt] (s) {};
        	\path (p) ++(90:1) node[circle,fill=black,draw,inner sep=0pt,minimum size=4pt] (t) {};
        	
        	\path (r) ++(0:3) node[circle,fill=black,draw,inner sep=0pt,minimum size=4pt] (u) {};
        	
        	\path (u) ++(-45:1) node[circle,fill=black,draw,inner sep=0pt,minimum size=4pt] (v) {};
        	
        	\draw[thick,dotted] (a) to node[font=\fontsize{8}{8},midway,inner sep=1pt,outer sep=1pt,minimum size=4pt,fill=white] {$e$} (b);
        	
        	\draw[thick, decorate,decoration=zigzag] (j) to[out=10,in=135] (k) (o) to[out=10,in=135] (p) (t) to[out=10,in=135] (u);
        	
        	\draw[thick] (a) -- (c) (a) -- (d) (a) -- (e) (f) to node[font=\fontsize{8}{8},midway,inner sep=1pt,outer sep=1pt,minimum size=4pt,fill=white] {$\alpha$} (g) (f) -- (h) (f) -- (i) (f) -- (j) (k) -- (l) (k) -- (m) (k) -- (n) (k) -- (o) (p) -- (q) (p) -- (r) (p) -- (s) (p) -- (t) (u) -- (v) (e) to node[font=\fontsize{8}{8},midway,inner sep=1pt,outer sep=1pt,minimum size=4pt,fill=white] {$\alpha$} (w) to node[font=\fontsize{8}{8},midway,inner sep=1pt,outer sep=1pt,minimum size=4pt,fill=white] {$\beta$} (x) to node[font=\fontsize{8}{8},midway,inner sep=1pt,outer sep=1pt,minimum size=4pt,fill=white] {$\alpha$} (y) to node[font=\fontsize{8}{8},midway,inner sep=1pt,outer sep=1pt,minimum size=4pt,fill=white] {$\beta$} (f);

            \draw[decoration={brace,amplitude=10pt},decorate] (1.1,1) -- node [midway,above,yshift=-2pt,xshift=25pt] {A fan} (1.1, -0.7);
            \draw[decoration={brace,amplitude=10pt},decorate] (0,1.3) -- node [midway,above,yshift=10pt,xshift=0pt] {An alternating path} (3.9, 0.9);
	\end{tikzpicture}
	\caption{A multi-step Vizing chain.}\label{fig:MSVC_illustration}
    \end{figure}

    Before we sketch our \textsf{Multi-Step Vizing Algorithm}, let us briefly mention two of its key features.

    \begin{itemize}
        \item At the start of the algorithm we fix a parameter $\ell$. The algorithm guarantees that every alternating path in the chain has length at most $\ell$. This is achieved by truncating each path whose length exceeds this bound at a random location among the first $\ell$ edges and growing a new fan and alternating path starting from that edge.\footnote{Technically, it ends up being more convenient 
        to let each alternating path be of length up to $2\ell$. Moreover, we truncate each ``long'' path to a random length between $\ell$ and $2\ell - 1$ (this ensures that the resulting path is neither too long nor too short). This is one of the minor details that we are for simplicity ignoring in this informal discussion.}

        \item When working with multi-step Vizing chains, it turns out to be crucial to require them to be \emphd{non-intersecting}. The precise definition of this notion is somewhat technical (see Definition~\ref{defn:non-int}), but, roughly, it prohibits certain types of overlaps between the fans and paths out of which the multi-step Vizing chain is built. For example, no two paths in the construction are allowed to share an {internal} edge. We achieve this by going back to an earlier stage of the algorithm whenever an intersection occurs.
    \end{itemize}
    We can now present an outline of our algorithm:

    \vspace{0.1in}

    {
    \floatname{algorithm}{Algorithm Sketch}
    \begin{breakablealgorithm}
        \caption{A randomized algorithm for non-intersecting multi-step Vizing chains}\label{inf:MSVC}
        \algsize
        \begin{flushleft}
            \textbf{Input}: 
            A proper partial $q$-edge-coloring $\phi$, an uncolored edge $e$, and a parameter $\ell \in \N$. \\
            \textbf{Output}: A non-intersecting chain $C = F_0 + P_0 + \cdots + F_k + P_k$ with $\length(P_i) \leq \ell$ for all $i$.
        \end{flushleft}
        \begin{algorithmic}[1]
            \State $e_0 \gets e$, \quad $k \gets 0$
            \While{true} \Comment{Basic \textsf{while} loop}
                \State $F + P \gets$ an $e_k$-augmenting one-step Vizing chain\label{step:VC}
                \If{$F + P$ intersects $F_j + P_j$ for some $j < k$}
                    \State $k \gets j$ \label{line:back1} \Comment{Return to step $j$}
                \ElsIf{$\length(P) \leq \ell$}
                        \State $F_k \gets F$, \quad $P_k \gets P$
                        \State \Return $C = F_0 + P_0 + \cdots + F_k + P_k$
                \Else
                    \State Pick a random natural number $\ell' \leq \ell$. \label{line:forward1}\label{line:short}
                    \State $F_k \gets F$, \quad $P_k \gets $ the first $\ell'$ edges of $P$ \Comment{Randomly shorten the path}
                    \State $e_{k+1} \gets $ the last edge of $P_k$, \quad $k \gets k+1$ \Comment{Move on to the next step}
                \EndIf
            \EndWhile
        \end{algorithmic}
    \end{breakablealgorithm}
    }

    \vspace{0.1in}
    
    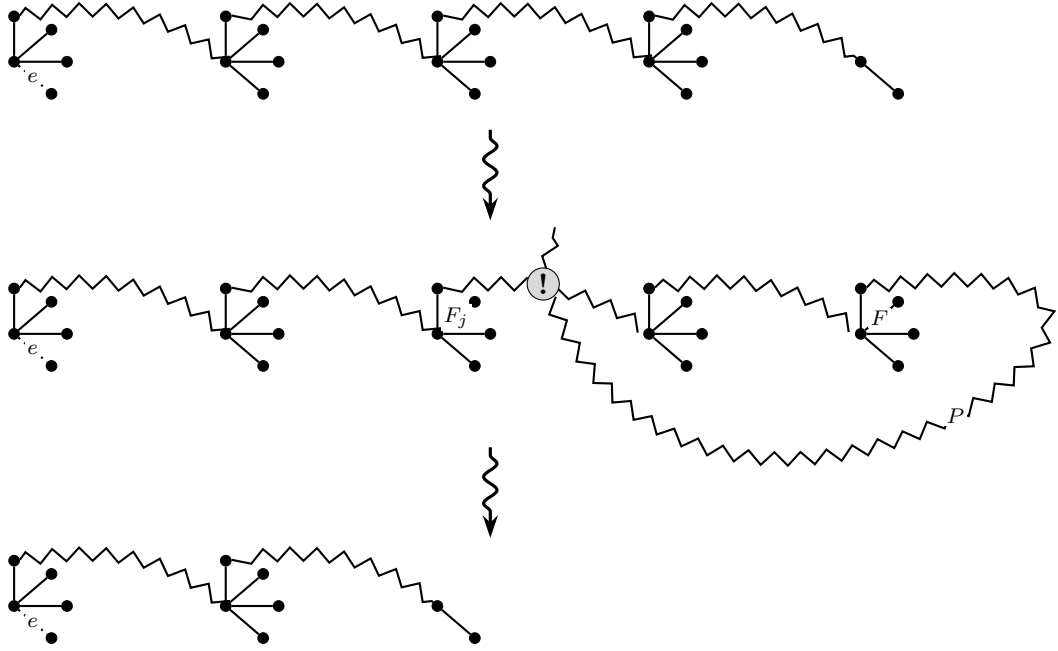
\begin{figure}[htb!]
    \centering
    \begin{tikzpicture}[xscale = 0.7,yscale=0.6]
            \clip (-0.5,-7) rectangle (20,7.5);
    
            \node[circle,fill=black,draw,inner sep=0pt,minimum size=4pt] (a) at (0,0) {};
        	\path (a) ++(-45:1) node[circle,fill=black,draw,inner sep=0pt,minimum size=4pt] (b) {};
        	\path (a) ++(0:1) node[circle,fill=black,draw,inner sep=0pt,minimum size=4pt] (c) {};
        	\path (a) ++(45:1) node[circle,fill=black,draw,inner sep=0pt,minimum size=4pt] (d) {};
        	\path (a) ++(90:1) node[circle,fill=black,draw,inner sep=0pt,minimum size=4pt] (e) {};

            \path (c) ++(0:3) node[circle,fill=black,draw,inner sep=0pt,minimum size=4pt] (f) {};
            
        	\path (f) ++(-45:1) node[circle,fill=black,draw,inner sep=0pt,minimum size=4pt] (g) {};
        	\path (f) ++(0:1) node[circle,fill=black,draw,inner sep=0pt,minimum size=4pt] (h) {};
        	\path (f) ++(45:1) node[circle,fill=black,draw,inner sep=0pt,minimum size=4pt] (i) {};
        	\path (f) ++(90:1) node[circle,fill=black,draw,inner sep=0pt,minimum size=4pt] (j) {};
        	
        	\path (h) ++(0:3) node[circle,fill=black,draw,inner sep=0pt,minimum size=4pt] (k) {};

            \path (k) ++(-45:1) node[circle,fill=black,draw,inner sep=0pt,minimum size=4pt] (l) {};
        	\path (k) ++(0:1) node[circle,fill=black,draw,inner sep=0pt,minimum size=4pt] (m) {};
        	\path (k) ++(45:1) node[circle,fill=black,draw,inner sep=0pt,minimum size=4pt] (n) {};
        	\path (k) ++(90:1) node[circle,fill=black,draw,inner sep=0pt,minimum size=4pt] (o) {};
        	
        	\path (m) ++(0:3) node[circle,fill=black,draw,inner sep=0pt,minimum size=4pt] (p) {};
        	
        	\path (p) ++(-45:1) node[circle,fill=black,draw,inner sep=0pt,minimum size=4pt] (q) {};
        	\path (p) ++(0:1) node[circle,fill=black,draw,inner sep=0pt,minimum size=4pt] (r) {};
        	\path (p) ++(45:1) node[circle,fill=black,draw,inner sep=0pt,minimum size=4pt] (s) {};
        	\path (p) ++(90:1) node[circle,fill=black,draw,inner sep=0pt,minimum size=4pt] (t) {};
        	
        	\path (r) ++(0:3) node[circle,fill=black,draw,inner sep=0pt,minimum size=4pt] (u) {};
        	
        	\path (u) ++(-45:1) node[circle,fill=black,draw,inner sep=0pt,minimum size=4pt] (v) {};sep=0pt,minimum size=4pt] (q) {};
        	\path (u) ++(0:1) node[circle,fill=black,draw,inner sep=0pt,minimum size=4pt] (w) {};
        	\path (u) ++(45:1) node[circle,fill=black,draw,inner sep=0pt,minimum size=4pt] (x) {};
        	\path (u) ++(90:1) node[circle,fill=black,draw,inner sep=0pt,minimum size=4pt] (y) {};
        	
        	\node[circle,inner sep=1pt] (z) at (10, 1.1) {\textbf{!}};
        	\path (z) ++(80:1.5) node (aa) {};

        	\draw[thick,dotted] (a) to node[font=\fontsize{8}{8},midway,inner sep=1pt,outer sep=1pt,minimum size=4pt,fill=white] {$e$} (b);
        	
        	\draw[thick, decorate,decoration=zigzag] (e) to[out=10,in=135] (f) (j) to[out=10,in=135] (k) (o) to[out=10,in=180] (z) to[out=-20,in=135] (p) (t) to[out=10,in=135] (u) (y) to[out=10,in=-70, looseness=4] node[font=\fontsize{8}{8},midway,inner sep=1pt,outer sep=1pt,minimum size=3pt,fill=white] {$P$} (z) -- (aa);
        	
        	\draw[thick] (a) -- (c) (a) -- (d) (a) -- (e) (f) -- (g) (f) -- (h) (f) -- (i) (f) -- (j) (k) -- (l) (k) -- (m) (k) to node[font=\fontsize{8}{8},midway,inner sep=1pt,outer sep=1pt,minimum size=4pt,fill=white] {$F_j$} (n) (k) -- (o) (p) -- (q) (p) -- (r) (p) -- (s) (p) -- (t) (u) -- (v) (u) -- (w) (u) to node[font=\fontsize{8}{8},midway,inner sep=1pt,outer sep=1pt,minimum size=4pt,fill=white] {$F$} (x) (u) -- (y);

            \node[circle,fill=gray!30,draw,inner sep=1.3pt] at (10, 1.1) {\textbf{!}};
        
        \begin{scope}[yshift=-3.5cm]
            \draw[-{Stealth[length=3mm,width=2mm]},very thick,decoration = {snake,pre length=3pt,post length=7pt,},decorate] (9,1) -- (9,-1);
        \end{scope}
        
        \begin{scope}[yshift=-6cm]
            \node[circle,fill=black,draw,inner sep=0pt,minimum size=4pt] (a) at (0,0) {};
        	\path (a) ++(-45:1) node[circle,fill=black,draw,inner sep=0pt,minimum size=4pt] (b) {};
        	\path (a) ++(0:1) node[circle,fill=black,draw,inner sep=0pt,minimum size=4pt] (c) {};
        	\path (a) ++(45:1) node[circle,fill=black,draw,inner sep=0pt,minimum size=4pt] (d) {};
        	\path (a) ++(90:1) node[circle,fill=black,draw,inner sep=0pt,minimum size=4pt] (e) {};

            \path (c) ++(0:3) node[circle,fill=black,draw,inner sep=0pt,minimum size=4pt] (f) {};
            
        	\path (f) ++(-45:1) node[circle,fill=black,draw,inner sep=0pt,minimum size=4pt] (g) {};
        	\path (f) ++(0:1) node[circle,fill=black,draw,inner sep=0pt,minimum size=4pt] (h) {};
        	\path (f) ++(45:1) node[circle,fill=black,draw,inner sep=0pt,minimum size=4pt] (i) {};
        	\path (f) ++(90:1) node[circle,fill=black,draw,inner sep=0pt,minimum size=4pt] (j) {};
        	
        	\path (h) ++(0:3) node[circle,fill=black,draw,inner sep=0pt,minimum size=4pt] (k) {};

            \path (k) ++(-45:1) node[circle,fill=black,draw,inner sep=0pt,minimum size=4pt] (l) {};

        	\draw[thick,dotted] (a) to node[font=\fontsize{8}{8},midway,inner sep=1pt,outer sep=1pt,minimum size=4pt,fill=white] {$e$} (b);
        	
        	\draw[thick, decorate,decoration=zigzag] (e) to[out=10,in=135] (f) (j) to[out=10,in=135] (k);
        	
        	\draw[thick] (a) -- (c) (a) -- (d) (a) -- (e) (f) -- (g) (f) -- (h) (f) -- (i) (f) -- (j) (k) -- (l);
        \end{scope}

        \begin{scope}[yshift=6cm]
            \node[circle,fill=black,draw,inner sep=0pt,minimum size=4pt] (a) at (0,0) {};
        	\path (a) ++(-45:1) node[circle,fill=black,draw,inner sep=0pt,minimum size=4pt] (b) {};
        	\path (a) ++(0:1) node[circle,fill=black,draw,inner sep=0pt,minimum size=4pt] (c) {};
        	\path (a) ++(45:1) node[circle,fill=black,draw,inner sep=0pt,minimum size=4pt] (d) {};
        	\path (a) ++(90:1) node[circle,fill=black,draw,inner sep=0pt,minimum size=4pt] (e) {};

            \path (c) ++(0:3) node[circle,fill=black,draw,inner sep=0pt,minimum size=4pt] (f) {};
            
        	\path (f) ++(-45:1) node[circle,fill=black,draw,inner sep=0pt,minimum size=4pt] (g) {};
        	\path (f) ++(0:1) node[circle,fill=black,draw,inner sep=0pt,minimum size=4pt] (h) {};
        	\path (f) ++(45:1) node[circle,fill=black,draw,inner sep=0pt,minimum size=4pt] (i) {};
        	\path (f) ++(90:1) node[circle,fill=black,draw,inner sep=0pt,minimum size=4pt] (j) {};
        	
        	\path (h) ++(0:3) node[circle,fill=black,draw,inner sep=0pt,minimum size=4pt] (k) {};

            \path (k) ++(-45:1) node[circle,fill=black,draw,inner sep=0pt,minimum size=4pt] (l) {};
        	\path (k) ++(0:1) node[circle,fill=black,draw,inner sep=0pt,minimum size=4pt] (m) {};
        	\path (k) ++(45:1) node[circle,fill=black,draw,inner sep=0pt,minimum size=4pt] (n) {};
        	\path (k) ++(90:1) node[circle,fill=black,draw,inner sep=0pt,minimum size=4pt] (o) {};
        	
        	\path (m) ++(0:3) node[circle,fill=black,draw,inner sep=0pt,minimum size=4pt] (p) {};
        	
        	\path (p) ++(-45:1) node[circle,fill=black,draw,inner sep=0pt,minimum size=4pt] (q) {};
        	\path (p) ++(0:1) node[circle,fill=black,draw,inner sep=0pt,minimum size=4pt] (r) {};
        	\path (p) ++(45:1) node[circle,fill=black,draw,inner sep=0pt,minimum size=4pt] (s) {};
        	\path (p) ++(90:1) node[circle,fill=black,draw,inner sep=0pt,minimum size=4pt] (t) {};
        	
        	\path (r) ++(0:3) node[circle,fill=black,draw,inner sep=0pt,minimum size=4pt] (u) {};
        	
        	\path (u) ++(-45:1) node[circle,fill=black,draw,inner sep=0pt,minimum size=4pt] (v) {};
        	
        	\draw[thick,dotted] (a) to node[font=\fontsize{8}{8},midway,inner sep=1pt,outer sep=1pt,minimum size=4pt,fill=white] {$e$} (b);
        	
        	\draw[thick, decorate,decoration=zigzag] (e) to[out=10,in=135] (f) (j) to[out=10,in=135] (k) (o) to[out=10,in=135] (p) (t) to[out=10,in=135] (u);
        	
        	\draw[thick] (a) -- (c) (a) -- (d) (a) -- (e) (f) -- (g) (f) -- (h) (f) -- (i) (f) -- (j) (k) -- (l) (k) -- (m) (k) -- (n) (k) -- (o) (p) -- (q) (p) -- (r) (p) -- (s) (p) -- (t) (u) -- (v);
        \end{scope}

        \begin{scope}[yshift=3.5cm]
            \draw[-{Stealth[length=3mm,width=2mm]},very thick,decoration = {snake,pre length=3pt,post length=7pt,},decorate] (9,1) -- (9,-1);
        \end{scope}
        \end{tikzpicture}
    \caption{An iteration of the basic \textsf{while} loop in Algorithm \ref{inf:MSVC} when there is an intersection between $F + P$ and $F_j + P_j$.}
    \label{fig:iteration_illustration}
\end{figure}

    We refer to the \textsf{while} loop in Algorithm~\ref{inf:MSVC} as the \emphd{basic loop}. Fig.~\ref{fig:iteration_illustration} shows a cartoon of an iteration of the basic loop. Inserting Algorithm~\ref{inf:MSVC} into Template~\ref{temp:seq}, we obtain the following final form of our edge-coloring algorithm:

    {
    \floatname{algorithm}{Algorithm Sketch}
    \begin{algorithm}[h]\algsize
        \caption{A randomized $q$-edge-coloring algorithm using multi-step Vizing chains}\label{inf:seq}
        \begin{flushleft}
            \textbf{Input}: A parameter $\ell \in \N$. \\
            \textbf{Output}: A proper $q$-edge-coloring of $G$.
        \end{flushleft}
        \begin{algorithmic}[1]
            \State $\phi \gets$ the empty coloring
            \While{there are uncolored edges}
                \State Pick an uncolored edge $e$ uniformly at random.\label{step:uar1}
                \State Find an $e$-augmenting multi-step Vizing chain $C$ using Algorithm~\ref{inf:MSVC}.\label{step:H1}
                \State Augment $\phi$ using $C$.
            \EndWhile
            \State \Return $\phi$
        \end{algorithmic}
    \end{algorithm}
    }

    In \cite{bernshteyn2023fast}, it was shown that when $q = \Delta + 1$ and $\ell = \Theta(\Delta^{16})$, this algorithm with high probability runs in time at most  $O(m\,\Delta^{17})$. The problem we are now facing is this:

    \begin{shaded}
        \begin{quote}\textsl{Can we exploit the larger number of available colors, i.e., $q = (1+\epsilon)\Delta$, to reduce the running time of Algorithm~\ref{inf:seq} to just $m \, \poly(1/\epsilon)$?}
        \end{quote}
    \end{shaded}

    To explain our solution, we need to say a few words about how Algorithms~\ref{inf:MSVC} and \ref{inf:seq} are analyzed. We take $\ell = \poly(1/\epsilon)$. There are two main factors affecting the runtime of Algorithm~\ref{inf:seq}:
    \begin{itemize}
        \item the number of iterations of the basic \textsf{while} loop, and
        \item the runtime of each iteration.
    \end{itemize}
    Throughout the entire execution of Algorithm~\ref{inf:seq}, our aim is to have at most $m\,\poly(1/\epsilon)$ iterations of the basic loop and spend (on average) at most $\poly(1/\epsilon)$ time per iteration, leading to the total running time of $m\,\poly(1/\epsilon)$. Both the number of iterations and the runtime per iteration present interesting challenges.

    \subsubsection*{Runtime per iteration: the Random Fan Algorithm}

    Let us first look at the running time of an individual iteration of the basic loop in Algorithm~\ref{inf:MSVC}. Most of the work there takes place on Step~\ref{step:VC}, where we obtain a fan $F$ and an alternating path $P$ such that the chain $F+P$ is $e_k$-augmenting. Note that the path $P$ could, in principle, be of length up to $\Theta(n)$, which would make it too time-consuming to list all of its edges. We circumvent this issue by only keeping track of the initial segment $\tilde{P}$ of $P$ of length $\min\set{\length(P), \,\ell}$. 
    
    To simplify this informal discussion, we may view the number of edges in $F+\tilde{P}$ as a convenient proxy for the runtime of Step~\ref{step:VC}. By definition, $\tilde{P}$ has length at most $\ell$, while $F$ has size at most $\Delta$. Thus, $F + \tilde{P}$ contains at most $\Delta + \ell$ edges. Unfortunately, this bound is too weak for our purposes; instead, we wish $F+ \tilde{P}$ to only have at most $\poly(1/\epsilon)$ edges. As $\ell = \poly(1/\epsilon)$ by our choice, what's needed is a $\poly(1/\epsilon)$ upper bound on the number of edges in the fan $F$. (Note that this problem does not arise in \cite{bernshteyn2023fast}, where the target runtime per iteration is $\poly(\Delta)$; in general, fans do not significantly affect the analysis in \cite{bernshteyn2023fast}, but are a major source of difficulties in this paper.)

    We address this problem by introducing a \emph{randomized} algorithm for building fans (see Algorithm~\ref{alg:rand_fan}). This is in contrast to the way fans are constructed in previous applications of the multi-step Vizing chain technique \cite{GP,VizingChain,bernshteyn2023fast,GmeasVizing}, where they are built deterministically. It is crucial for this randomized approach that we are using $(1+\epsilon)\Delta$ colors, which implies that for each vertex $v \in V(G)$, there are at least $\epsilon \Delta$ colors that are not used on any of its incident edges. We call such colors \emphd{missing} at $v$. To construct a fan, we repeatedly pick colors missing at certain vertices. In our algorithm, the missing colors are chosen \emph{uniformly at random}. To guarantee that the resulting fan $F$ has size at most $\poly(1/\epsilon)$, we fix a parameter $k_\mathrm{max} = \poly(1/\epsilon)$ and restart the procedure from scratch whenever the fan's size exceeds $k_\mathrm{max}$. The flexibility provided by the random choice of the missing colors allows us to establish an upper bound on the probability that the process has to restart ``too many'' times (see Proposition~\ref{lemma:fan_number_of_tries}). In our final analysis, this is used to conclude that with high probability, the average time spent constructing each fan during the execution of Algorithm~\ref{inf:seq} is at most $\poly(1/\epsilon)$, as desired.

    Here we should make another important remark. The fact that the running time of our algorithm must be linear in $m$ regardless of $\Delta$ creates subtle complications even in seemingly innocuous parts of the analysis. For example, as mentioned above, our algorithm for building fans requires choosing a random missing color at a given vertex $v \in V(G)$. However, we cannot just search through all the colors to find a missing one without introducing an unwanted dependence on $\Delta$ in the runtime. Our solution is to pick a uniformly random color from $[q]$ and then check whether it is missing at $v$ (this takes $O(1)$ time); if yes, we have succeeded, and if not, we pick a random color again and continue iterating until a missing color is found (see Algorithm \ref{alg:rand_col}). Since there are at least $\epsilon\Delta$ colors that are missing at $v$, the expected number of attempts until a missing color is found is at most $1/\epsilon$, and it is not hard to bound the probability that the number of attempts is ``large'' (see Lemma~\ref{lemma:rand_color_runtime}). At the end of the day, this is used to show that, while executing Algorithm~\ref{inf:seq}, we spend on average $\poly(1/\epsilon)$ time to pick a random missing color at a vertex, with high probability.

    \subsubsection*{Number of iterations: entropy compression}

    To bound the number of iterations of the basic \textsf{while} loop, we employ the so-called \emph{entropy compression method}. (The idea of using it in this setting was introduced in \cite{bernshteyn2023fast}.) This method was invented by Moser and Tardos \cite{MT} in order to prove an algorithmic version of the Lov\'asz Local Lemma \ep{the name ``entropy compression method'' was coined by Tao~\cite{Tao}}. Grytczuk, Kozik, and Micek \cite{Grytczuk} discovered that the entropy compression method can sometimes lead to improved combinatorial results if applied directly, with no explicit mention of the Local Lemma. Since then, the method has found numerous applications in combinatorics, probability, and computer science; see \cite{Duj,Esperet,acyclic_entropy,entropy_permutations,Achlioptas,Oracles,Physics} and the references therein for a small sample of the vast related literature.

    The entropy compression method is used to prove that a given randomized algorithm terminates after a specified time with high probability. The idea of the method is to encode the execution process of the algorithm in such a way that the original sequence of random inputs can be recovered from the resulting encoding. One then shows that if the algorithm runs for too long, the space of possible codes becomes smaller than the space of inputs, which leads to a contradiction.
    In our application, for each iteration of the basic loop in Algorithm~\ref{inf:MSVC}, the encoding includes information such as, for example, whether the algorithm went to line \ref{line:back1} or \ref{line:forward1} on that iteration, and if the algorithm went to line \ref{line:back1}, the value $k - j$ (i.e., how many steps back were taken). Additionally, we record certain data pertaining to the final state of the algorithm, such as the edge $e_k$ after the last iteration. This information is insufficient to fully reconstruct the execution  process, but we can establish an upper bound on the number of different processes that can result in the same encoding.

    The random inputs in Algorithm~\ref{inf:MSVC} are the numbers $\ell'$ generated at line~\ref{line:short}. Hence, there are at most about $\ell^t$ random input sequences for which the basic loop is iterated $t$ times. On the other hand, it is shown in \cite{bernshteyn2023fast} that each encoding of $t$ iterations of the basic loop can arise from at most about $\Delta^{O(t)}$ different processes. Roughly speaking, the analysis in \cite{bernshteyn2023fast} succeeds because there one can take $\ell$ so large that $\ell^t \gg \Delta^{O(t)}$. Unfortunately, here we only have $\ell = \poly(1/\epsilon)$, so this calculation breaks down. Thankfully, we have an ace up our sleeve, namely the randomized algorithm for building fans discussed above. The fact that the fans are constructed in a \emph{randomized} fashion means that the algorithm takes in more randomness than just the numbers $\ell'$ generated in line~\ref{line:short}. Indeed, each randomly chosen missing color is picked from a set of at least $\epsilon \Delta$ options. Intuitively, this means that when we pick a random missing color, this allows us to tolerate an extra factor of $\Delta$ in the bound on the number of encodings/processes per encoding. With this observation in hand, we are able to carefully balance the powers of $\Delta$ on both sides of the inequality to ultimately eliminate the dependence on $\Delta$ in the necessary lower bound on $\ell$.

    We should stress that the resulting argument is quite delicate, since the factors of $\Delta$ have to be balanced exactly (even a single surplus factor of $\Delta$ somewhere in the analysis would ultimately lead to a dependence on $\Delta$ in the algorithm's runtime). For example, it is necessary for us to include enough information in the encoding to be able to fully reconstruct not only the alternating paths, but also the fans generated by the algorithm. (By contrast, in \cite{bernshteyn2023fast} the fans were built deterministically, so they could be recovered without any extra information.) Furthermore, we must take care to avoid super-polynomial dependence on $\epsilon$. In particular, using the naive lower bound of $\epsilon \Delta$ on the number of missing colors at a vertex would result in a factor of roughly $(1/\epsilon)^{O(1/\epsilon)}$ in the running time of the algorithm. To achieve polynomial dependence on $1/\epsilon$, we take into account the exact number of missing colors at each vertex and make use of the fact that it may exceed $\epsilon \Delta$.

    \medskip

    \noindent In summary, the central difference between our version of Algorithm~\ref{inf:MSVC} and the one studied in \cite{bernshteyn2023fast} is that we use a randomized construction to build the fan $F$ in Step~\ref{step:VC}. With a careful analysis, this extra randomness allows us to both bound the time needed to construct the fan by $\poly(1/\epsilon)$ and to remove the dependence on $\Delta$ in the number of iterations of the basic loop.
    
    Before finishing the introduction, we would like to discuss two more aspects of our algorithm: its success probability and the size of the augmenting subgraphs it generates.

    \subsubsection*{Success probability and restrictions on $\Delta$} 
    The randomized algorithms for $(1+\epsilon)\Delta$-edge-coloring developed in \cite{BCPS_nibble,Assadi} rely, among other ingredients, on performing a randomized construction on $G$ (such as randomly splitting the edges of $G$ into subgraphs $G_1$, \ldots, $G_k$) and then arguing that certain desirable properties hold with high probability for every vertex of $G$ (for example, the maximum degree of each $G_i$ is bounded by roughly $\Delta/k$). The desired property for a particular vertex $v$ is determined by independent random choices in the neighborhood of $v$, so the best one can hope for is to bound the failure probability at $v$ by a quantity of the type $\exp(-\Theta(\Delta))$ via a Chernoff-style concentration bound. One then takes the union bound over all vertices to argue that the total failure probability is at most $n \exp(-\Theta(\Delta))$. For this quantity to be much less then $1$, it is necessary to have $\Delta = \Omega(\log n)$---this explains the appearance of the lower bounds on $\Delta$ such as \eqref{eq:Delta_lower_bound} in these results and indicates that a lower bound on $\Delta$ is unavoidable with these approaches. Furthermore, for $\Delta = \Theta(\log n)$, the resulting failure probability is only polynomially small.

    Our analysis is different in that we do not require any favorable outcomes to occur at every vertex or edge of $G$. For example, we do not claim that with high probability, it takes only $\poly(1/\epsilon)$ time to color \emph{each} edge. Instead, we only control the \emph{average} runtime over all edges throughout the execution of Algorithm~\ref{inf:seq}; and it is perhaps not surprising that the average of $m$ random variables would be concentrated around its mean with an exponentially small in $m$ tail probability. Moreover, such ``global'' analysis does not require a lower bound on $\Delta$ and remains valid even if $\Delta$ is constant.

    \subsubsection*{Small augmenting subgraphs}

    A byproduct of our analysis is a construction of ``small'' augmenting subgraphs for $(1+\epsilon)\Delta$-edge-coloring. More precisely, we consider the following question:

    \begin{shaded}
        \begin{quote}\textsl{Let $\phi \colon E \to [q] \cup \set{\blank}$ be a proper partial coloring and let $e$ be an uncolored edge. What is the minimum number of edges in an $e$-augmenting subgraph $H$?}
        \end{quote}
    \end{shaded}

    \noindent For $q = \Delta + 1$, Vizing's original proof of Theorem~\ref{theo:Vizing} gives augmenting subgraphs with $O(n)$ edges \cite{Vizing}. The first improvement to this bound comes from the paper of Greb\'ik and Pikhurko \cite{GP}, whose two-step Vizing chain construction yields augmenting subgraphs with $\poly(\Delta) \sqrt{n}$ edges.\footnote{Although this result is not explicitly stated in \cite{GP}, it follows easily from the arguments presented there; see \cite[\S3]{VizingChain} for a detailed discussion.} By using multi-step Vizing chains, the first author was able to build augmenting subgraphs with $\poly(\Delta) \log^2 n$ edges \cite{VizingChain}. This was improved to just $O(\Delta^7 \log n)$ by Christiansen \cite{Christ}.

    When $q = (1+\epsilon)\Delta$ with $\epsilon \gg 1/\Delta$, we may hope to lower the bound even further. In particular, it is desirable to remove the polynomial dependence on $\Delta$. By the following result of Chang, He, Li, Pettie, and Uitto, the best one can hope for is $\Theta(\log (\epsilon n)/\epsilon)$:

    \begin{theo}[{Chang--He--Li--Pettie--Uitto \cite[Theorem 7.1]{CHLPU}}]\label{theo:CHLPU}
        For any $1/\Delta \leq \epsilon \leq 1/3$ and $n > \Delta$, there exists an $n$-vertex graph $G$ of maximum degree $\Delta$ with a proper partial $(1+\epsilon)\Delta$-edge-coloring $\phi$ and an uncolored edge $e$ such that every $e$-augmenting subgraph $H$ of $G$ has diameter $\Omega(\log (\epsilon n)/\epsilon)$ \ep{and hence it has $\Omega(\log (\epsilon n)/\epsilon)$ edges}.
    \end{theo}

    Here we obtain an upper bound that matches Theorem~\ref{theo:CHLPU} up to a $\poly(1/\epsilon)$ factor:

    \begin{tcolorbox}
    \begin{theo}[Small augmenting subgraphs for $(1+\epsilon)\Delta$-edge-coloring]\label{theo:augment_small}
            Let $G$ be an $n$-vertex graph with maximum degree $\Delta \geq 1/\eps$, where $0 < \epsilon < 1$. Then for any proper partial $(1+\epsilon)\Delta$-edge-coloring $\phi$ of $G$ and every uncolored edge $e$, there exists an $e$-augmenting subgraph $H$ of $G$ with $O((\log n)/\epsilon^4)$ edges.
    \end{theo}
    \end{tcolorbox}

    Theorem~\ref{theo:augment_small} follows as a consequence from our analysis of the \textsf{Multi-Step Vizing Algorithm}; see \S\ref{subsec: rmsva analysis} for a proof. We remark that \emph{on average} the augmenting subgraphs generated during the execution of Algorithm~\ref{inf:seq} have only $\poly(1/\epsilon)$ edges; that is why the total running time of the algorithm is linear in $m$ instead of $O(m \log n)$.

    \bigskip

    \noindent We finish this introduction by an outline of the structure of the paper. In \S\ref{sec:notation} we introduce the notation and terminology pertaining to fans, paths, and Vizing chains, which is largely borrowed from \cite{VizingChain,bernshteyn2023fast}. Then, in \S\ref{section: data_structures}, we carefully describe the data structures used to store the graph, the current partial coloring, and other relevant information throughout the execution of our algorithm. A judicious choice of the data structures is important to achieve a linear runtime (for example, checking whether a given color is missing at a vertex must be doable in $O(1)$ time). The \textsf{Multi-Step Vizing Algorithm} (i.e., the formal version of Algorithm~\ref{inf:MSVC}) is described in \S\ref{sec:MSVC}. That section also includes the proof of the algorithm's correctness. The runtime of the \textsf{Multi-Step Vizing Algorithm} is analyzed in \S\S\ref{subsec: fan analysis}--\ref{subsec: rmsva analysis}. Specifically, \S\ref{subsec: fan analysis} studies our randomized algorithm for building fans, while \S\ref{subsec: rmsva analysis} contains the entropy compression-based analysis of the number of iterations of the basic loop. Finally, we put everything together and complete the proof of Theorem~\ref{theo:main_theo} in \S\ref{section: sequential}.

\section{Notation and Preliminaries on Augmenting Chains}\label{sec:notation}

In this section, we define the augmenting subgraphs we will be constructing.
The definitions in this section are taken from \cite[\S3]{bernshteyn2023fast} and so the familiar reader may skip on to \S\ref{section: data_structures}.

\subsection{Chains: general definitions}

Given a partial coloring $\phi$ and $x\in V$, we let \[M(\phi, x) \defeq [q]\setminus\{\phi(xy)\,:\, xy \in E\}\] be the set of all the \emphd{missing} colors at $x$ under the coloring $\phi$. Since there are more than $\Delta$ available colors, $M(\phi, x)$ is always nonempty. An uncolored edge $xy$ is \emphd{$\phi$-happy} if $M(\phi, x)\cap M(\phi, y)\neq \0$. If $e = xy$ is $\phi$-happy, we can extend the coloring $\phi$ by assigning any color in $M(\phi, x)\cap M(\phi, y)$ to $e$.

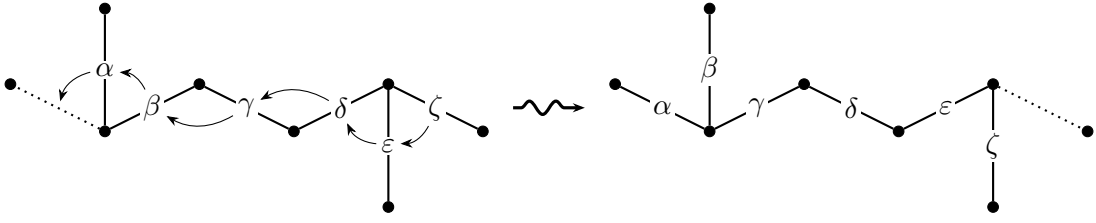
\begin{figure}[t]
	\centering
	\begin{tikzpicture}
	\begin{scope}
	\node[circle,fill=black,draw,inner sep=0pt,minimum size=4pt] (a) at (0,0) {};
	\node[circle,fill=black,draw,inner sep=0pt,minimum size=4pt] (b) at (-1.25,0.625) {};
	\node[circle,fill=black,draw,inner sep=0pt,minimum size=4pt] (c) at (0,1.625) {};
	\node[circle,fill=black,draw,inner sep=0pt,minimum size=4pt] (d) at (1.25,0.625) {};
	\node[circle,fill=black,draw,inner sep=0pt,minimum size=4pt] (e) at (2.5,0) {};
	\node[circle,fill=black,draw,inner sep=0pt,minimum size=4pt] (f) at (3.75,0.625) {};
	\node[circle,fill=black,draw,inner sep=0pt,minimum size=4pt] (g) at (3.75,-1) {};
	\node[circle,fill=black,draw,inner sep=0pt,minimum size=4pt] (h) at (5,0) {};
	
	\draw[thick,dotted] (a) to node[midway,inner sep=0pt,minimum size=4pt] (i) {} (b);
	\draw[thick] (a) to node[midway,inner sep=1pt,outer sep=1pt,minimum size=4pt,fill=white] (j) {$\alpha$} (c);
	\draw[thick] (a) to node[midway,inner sep=1pt,outer sep=1pt,minimum size=4pt,fill=white] (k) {$\beta$} (d);
	\draw[thick] (d) to node[midway,inner sep=1pt,outer sep=1pt,minimum size=4pt,fill=white] (l) {$\gamma$} (e);
	\draw[thick] (f) to node[midway,inner sep=1pt,outer sep=1pt,minimum size=4pt,fill=white] (m) {$\delta$} (e);
	\draw[thick] (f) to node[midway,inner sep=1pt,outer sep=1pt,minimum size=4pt,fill=white] (n) {$\epsilon$} (g);
	\draw[thick] (f) to node[midway,inner sep=1pt,outer sep=1pt,minimum size=4pt,fill=white] (o) {$\zeta$} (h);
	
	\draw[-{Stealth[length=1.6mm]}] (j) to[bend right] (i);
	\draw[-{Stealth[length=1.6mm]}] (k) to[bend right] (j);
	\draw[-{Stealth[length=1.6mm]}] (l) to[bend left] (k);
	\draw[-{Stealth[length=1.6mm]}] (m) to[bend right] (l);
	\draw[-{Stealth[length=1.6mm]}] (n) to[bend left] (m);
	\draw[-{Stealth[length=1.6mm]}] (o) to[bend left] (n);
	\end{scope}
	
	\draw[-{Stealth[length=1.6mm]},very thick,decoration = {snake,pre length=3pt,post length=7pt,},decorate] (5.4,0.3125) -- (6.35,0.3125);
	
	\begin{scope}[xshift=8cm]
	\node[circle,fill=black,draw,inner sep=0pt,minimum size=4pt] (a) at (0,0) {};
	\node[circle,fill=black,draw,inner sep=0pt,minimum size=4pt] (b) at (-1.25,0.625) {};
	\node[circle,fill=black,draw,inner sep=0pt,minimum size=4pt] (c) at (0,1.625) {};
	\node[circle,fill=black,draw,inner sep=0pt,minimum size=4pt] (d) at (1.25,0.625) {};
	\node[circle,fill=black,draw,inner sep=0pt,minimum size=4pt] (e) at (2.5,0) {};
	\node[circle,fill=black,draw,inner sep=0pt,minimum size=4pt] (f) at (3.75,0.625) {};
	\node[circle,fill=black,draw,inner sep=0pt,minimum size=4pt] (g) at (3.75,-1) {};
	\node[circle,fill=black,draw,inner sep=0pt,minimum size=4pt] (h) at (5,0) {};
	
	\draw[thick] (a) to node[midway,inner sep=1pt,outer sep=1pt,minimum size=4pt,fill=white] (i) {$\alpha$} (b);
	\draw[thick] (a) to node[midway,inner sep=1pt,outer sep=1pt,minimum size=4pt,fill=white] (j) {$\beta$} (c);
	\draw[thick] (a) to node[midway,inner sep=1pt,outer sep=1pt,minimum size=4pt,fill=white] (k) {$\gamma$} (d);
	\draw[thick] (d) to node[midway,inner sep=1pt,outer sep=1pt,minimum size=4pt,fill=white] (l) {$\delta$} (e);
	\draw[thick] (f) to node[midway,inner sep=1pt,outer sep=1pt,minimum size=4pt,fill=white] (m) {$\epsilon$} (e);
	\draw[thick] (f) to node[midway,inner sep=1pt,outer sep=1pt,minimum size=4pt,fill=white] (n) {$\zeta$} (g);
	\draw[thick,dotted] (f) to node[midway,inner sep=0pt,minimum size=4pt] (o) {} (h);
	\end{scope}
	\end{tikzpicture}
	\caption{Shifting a coloring along a chain (Greek letters represent colors).}\label{fig:shift}
\end{figure}

Given a proper partial coloring, we wish to modify it in order to create a partial coloring with a happy edge. 
We will do so by finding small augmenting subgraphs called chains.
A \emphd{chain} of length $k$ is a sequence of distinct edges $C = (e_0, \ldots, e_{k-1})$ such that $e_i$ and $e_{i+1}$ are adjacent for each $0 \leq i < k-1$. 
Let $\Start(C) \defeq e_0$ and $\End(C) \defeq e_{k-1}$ denote the first and the last edges of $C$ respectively and let $\length(C) \defeq k$ be the length of $C$. 
We also let $E(C) \defeq \set{e_0, \ldots, e_{k-1}}$ be the edge set of $C$ and $V(C)$ be the set of vertices incident to the edges in $E(C)$.
To assist with defining our augmenting process, we will define the $\Shift$ operation (see Fig.~\ref{fig:shift}). 
Given a chain $C = (e_0, \ldots, e_{k-1})$, we define the coloring $\Shift(\phi, C)$ as follows:
\[\Shift(\phi, C)(e) \defeq \left\{\begin{array}{cc}
   \phi(e_{i+1})  & e = e_i, \, 0 \leq i < k-1; \\
    \blank & e = e_{k-1}; \\
    \phi(e) & \text{otherwise.}
\end{array}\right.\]
In other words, $\Shift(\phi, C)$ ``shifts'' the color from $e_{i+1}$ to $e_i$, leaves $e_{k-1}$ uncolored, and keeps the coloring on the other edges unchanged.
We call $C$ a \emphd{$\phi$-shiftable chain} if $\phi(e_0) = \blank$ and the coloring $\Shift(\phi, C)$ is proper.
For $C = (e_0, \ldots, e_{k-1})$, we let $C^* = (e_{k-1}, \ldots, e_0)$. 
It is clear that if $C$ is $\phi$-shiftable, then $C^*$ is $\Shift(\phi, C)$-shiftable and $\Shift(\Shift(\phi, C), C^*) = \phi$.

The next definition captures the class of chains that can be used to create a happy edge:

\begin{defn}[Happy chains]
    We say that a chain $C$ is \emphd{$\phi$-happy} for a partial coloring $\phi$ if it is $\phi$-shiftable and the edge $\End(C)$ is $\Shift(\phi, C)$-happy.
\end{defn}

Using the terminology of Definition~\ref{defn:aug}, we observe that a $\phi$-happy chain with $\Start(C) = e$ forms an $e$-augmenting subgraph of $G$ with respect to the partial coloring $\phi$. 
For a $\phi$-happy chain $C$ and a \textit{valid} color $\alpha$, we let $\aug(\phi, C, \alpha)$ be a proper coloring obtained from $\Shift(\phi, C)$ by assigning to $\End(C)$ the color $\alpha$ ($\alpha$ is valid if the resulting color is proper).
Our aim now becomes to develop algorithms for constructing $\phi$-happy chains for a given partial coloring $\phi$.

For a chain $C = (e_0, \ldots, e_{k-1})$ and $1 \leq j \leq k$, we define the \emphd{initial segment} of $C$ as
\[C|j \,\defeq\, (e_0, \ldots, e_{j-1}).\]
If $C$ is $\phi$-shiftable, then $C|j$ is also $\phi$-shiftable. 
Given two chains $C = (e_0, \ldots, e_{k-1})$ and $C' = (f_0, \ldots, f_{j-1})$ with $e_{k-1} = f_0$, we define their \emphd{sum} as follows:
\[C+C' \,\defeq\, (e_0, \ldots, e_{k-1} = f_0, \ldots, f_{j-1}).\]
Note that if $C$ is $\phi$-shiftable and $C'$ is $\Shift(\phi, C)$-shiftable, then $C+C'$ is $\phi$-shiftable.

With these definitions in mind, we are now ready to describe the types of chains that will be used as building blocks in our algorithms.

\subsection{Path chains}\label{subsec:pathchains}

The first special type of chains we will consider are path chains:

\begin{defn}[Path chains]
    A chain $P = (e_0, \ldots, e_{k-1})$ is a \emphd{path chain} if the edges $e_1$, \ldots, $e_{k-1}$ form a path in $G$, i.e., if there exist distinct vertices $x_1$, \ldots, $x_k$ such that $e_i = x_ix_{i+1}$ for all $1 \leq i \leq k-1$. We let $x_0$ be the vertex in $e_0$ which is not in $e_1$ and let $\vstart(P) \defeq x_0$, $\vend(P) \defeq x_k$ denote the first and last vertices on the path chain respectively. Note that the vertex $x_0$ may coincide with $x_i$ for some $3 \leq i \leq k$; see Fig.~\ref{subfig:unsucc} for an example. Furthermore, the vertices $\vstart(P)$ and $\vend(P)$ are uniquely determined unless $P$ is a single edge.
\end{defn}

\begin{defn}[Internal vertices and edges]\label{defn:internal}
    An edge of a path chain $P$ is \emphd{internal} if it is distinct from $\Start(P)$ and $\End(P)$. We let $\IE(P)$ denote the set of all internal edges of $P$. A vertex of $P$ is \emphd{internal} if it is not incident to $\Start(P)$ or $\End(P)$. We let $\IV(P)$ denote the set of all internal vertices of $P$.
\end{defn}

We are particularly interested in path chains containing at most $2$ colors, except possibly at their first edge.
We refer to such chains as \emphd{alternating paths}.
Specifically, given a proper partial coloring $\phi$ and $\alpha$, $\beta \in [q]$, we say that a path chain $P$ is an \emphd{$\alpha\beta$-path} under $\phi$ if all edges of $P$ except possibly $\Start(P)$ are colored $\alpha$ or $\beta$. To simplify working with $\alpha\beta$-paths, we introduce the following notation.
Let $G(\phi, \alpha\beta)$ be the spanning subgraph of $G$ with 
\[E(G(\phi, \alpha\beta)) \,\defeq\, \{e\in E\,:\, \phi(e) \in \{\alpha, \beta\}\}.\]
Since $\phi$ is proper, the maximum degree of $G(\phi, \alpha\beta)$ is at most $2$. 
Hence, the connected components of $G(\phi, \alpha\beta)$ are paths or cycles (an isolated vertex is a path of length $0$).
For $x\in V$, let $G(x;\phi, \alpha\beta)$ denote the connected component of $G(\phi, \alpha\beta)$ containing $x$ and $\deg(x; \phi, \alpha\beta)$ denote the degree of $x$ in $G(\phi, \alpha\beta)$. We say that $x$, $y \in V$ are \emphd{$(\phi, \alpha\beta)$-related} if $G(x;\phi, \alpha\beta) = G(y;\phi, \alpha\beta)$, i.e., if $y$ is reachable from $x$ by a path in $G(\phi, \alpha\beta)$.

\begin{defn}[Hopeful and successful edges]
    Let $\phi$ be a proper partial coloring of $G$ and let $\alpha$, $\beta \in [q]$. Let $xy \in E$ be an edge such that $\phi(xy) = \blank$. We say that $xy$ is \emphd{$(\phi, \alpha\beta)$-hopeful} if $\deg(x;\phi, \alpha\beta) < 2$ and $\deg(y;\phi, \alpha\beta) < 2$. Further, we say that $xy$ is \emphd{$(\phi, \alpha\beta)$-successful} if it is $(\phi, \alpha\beta)$-hopeful and $x$ and $y$ are not $(\phi, \alpha\beta)$-related.
\end{defn}

Let $\phi$ be a proper partial coloring and let $\alpha$, $\beta \in [q]$. Consider a $(\phi, \alpha\beta)$-hopeful edge $e = xy$. Depending on the degrees of $x$ and $y$ in $G(\phi, \alpha\beta)$, the following two situations are possible:

\begin{figure}[t]
	\centering
	\begin{subfigure}[t]{.4\textwidth}
		\centering
		\begin{tikzpicture}
		\node[draw=none,minimum size=2.5cm,regular polygon,regular polygon sides=7] (P) {};

		\node[circle,fill=black,draw,inner sep=0pt,minimum size=4pt] (x) at (P.corner 4) {};
		\node[circle,fill=black,draw,inner sep=0pt,minimum size=4pt] (y) at (P.corner 5) {};
		\node[circle,fill=black,draw,inner sep=0pt,minimum size=4pt] (a) at (P.corner 6) {};
		\node[circle,fill=black,draw,inner sep=0pt,minimum size=4pt] (b) at (P.corner 7) {};
		\node[circle,fill=black,draw,inner sep=0pt,minimum size=4pt] (c) at (P.corner 1) {};
		\node[circle,fill=black,draw,inner sep=0pt,minimum size=4pt] (d) at (P.corner 2) {};
		\node[circle,fill=black,draw,inner sep=0pt,minimum size=4pt] (e) at (P.corner 3) {};
		
		\node[anchor=north] at (x) {$x$};
		\node[anchor=north] at (y) {$y$};
		
		\draw[thick,dotted] (x) to (y);
		\draw[thick] (y) to node[midway,inner sep=1pt,outer sep=1pt,minimum size=4pt,fill=white] {$\alpha$} (a);
		\draw[thick] (a) to node[midway,inner sep=1pt,outer sep=1pt,minimum size=4pt,fill=white] {$\beta$} (b);
		\draw[thick] (b) to node[midway,inner sep=1pt,outer sep=1pt,minimum size=4pt,fill=white] {$\alpha$} (c);
		\draw[thick] (c) to node[midway,inner sep=1pt,outer sep=1pt,minimum size=4pt,fill=white] {$\beta$} (d);
		\draw[thick] (d) to node[midway,inner sep=1pt,outer sep=1pt,minimum size=4pt,fill=white] {$\alpha$} (e);
		\draw[thick] (e) to node[midway,inner sep=1pt,outer sep=1pt,minimum size=4pt,fill=white] {$\beta$} (x);
		\end{tikzpicture}
		\caption{The edge $xy$ is not $(\phi, \alpha \beta)$-successful.}\label{subfig:unsucc}
	\end{subfigure}%
	\qquad%
	\begin{subfigure}[t]{.4\textwidth}
		\centering
		\begin{tikzpicture}
		\node[draw=none,minimum size=2.5cm,regular polygon,regular polygon sides=7] (P) {};

		\node[circle,fill=black,draw,inner sep=0pt,minimum size=4pt] (x) at (P.corner 4) {};
		\node[circle,fill=black,draw,inner sep=0pt,minimum size=4pt] (y) at (P.corner 5) {};
		\node[circle,fill=black,draw,inner sep=0pt,minimum size=4pt] (a) at (P.corner 6) {};
		\node[circle,fill=black,draw,inner sep=0pt,minimum size=4pt] (b) at (P.corner 7) {};
		\node[circle,fill=black,draw,inner sep=0pt,minimum size=4pt] (c) at (P.corner 1) {};
		\node[circle,fill=black,draw,inner sep=0pt,minimum size=4pt] (d) at (P.corner 2) {};
		\node[circle,fill=black,draw,inner sep=0pt,minimum size=4pt] (e) at (P.corner 3) {};
		\node[circle,fill=black,draw,inner sep=0pt,minimum size=4pt] (f) at (-2.2,0) {}; 
		
		\node[anchor=north] at (x) {$x$};
		\node[anchor=north] at (y) {$y$};
		
		\draw[thick,dotted] (x) to (y);
		\draw[thick] (y) to node[midway,inner sep=1pt,outer sep=1pt,minimum size=4pt,fill=white] {$\alpha$} (a);
		\draw[thick] (a) to node[midway,inner sep=1pt,outer sep=1pt,minimum size=4pt,fill=white] {$\beta$} (b);
		\draw[thick] (b) to node[midway,inner sep=1pt,outer sep=1pt,minimum size=4pt,fill=white] {$\alpha$} (c);
		\draw[thick] (c) to node[midway,inner sep=1pt,outer sep=1pt,minimum size=4pt,fill=white] {$\beta$} (d);
		\draw[thick] (d) to node[midway,inner sep=1pt,outer sep=1pt,minimum size=4pt,fill=white] {$\alpha$} (e);
		\draw[thick] (e) to node[midway,inner sep=1pt,outer sep=1pt,minimum size=4pt,fill=white] {$\beta$} (f);
		\end{tikzpicture}
		\caption{The edge $e = xy$ is $(\phi, \alpha \beta)$-successful.}
	\end{subfigure}
	\caption{The chain $P(e; \phi, \alpha\beta)$.}\label{fig:path}
\end{figure}
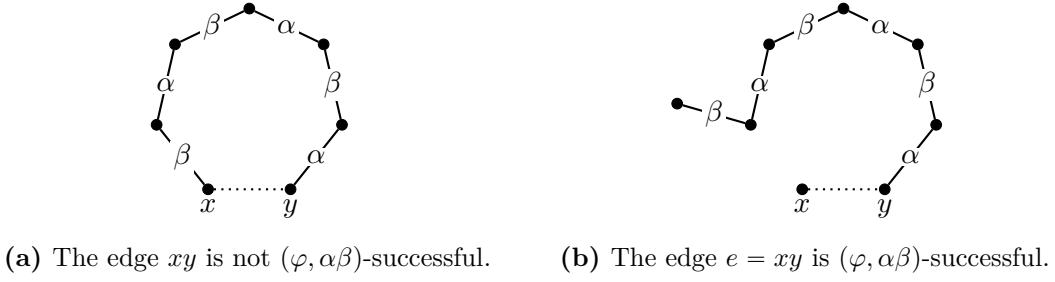

\begin{enumerate}[label=\ep{\textbf{Case \arabic*}},wide]
    \item If $\deg(x;\phi, \alpha\beta) = 0$ or $\deg(y; \phi, \alpha\beta) = 0$ (or both), then either $\alpha$ or $\beta$ is missing at both $x$ and $y$, so $e$ is $\phi$-happy.

    \item If $\deg(x;\phi, \alpha\beta) = \deg(y; \phi, \alpha\beta) = 1$, then both $x$ and $y$ miss exactly one of the colors $\alpha$, $\beta$. If they miss the same color, $e$ is $\phi$-happy. Otherwise suppose that, say, $\alpha \in M(\phi, x)$ and $\beta \in M(\phi, y)$.
    We define a path chain $P(e; \phi, \alpha\beta)$ by
    \[
        P(e; \phi, \alpha\beta) \,\defeq\, (e_0=e, e_1, \ldots, e_k),
    \]
    where $(e_1, \ldots, e_k)$ is the maximal path in $G(\phi, \alpha\beta)$ starting at $y$. Note that we have $\phi(e_1) = \alpha$ (in particular, the order of $\alpha$ and $\beta$ in the notation $P(e; \phi, \alpha\beta)$ matters, as $P(e; \phi, \beta\alpha)$ is the maximal path in $G(\phi, \alpha\beta)$ starting at $x$).
    The chain $P(e; \phi, \alpha\beta)$ is $\phi$-shiftable \cite[Fact 4.4]{VizingChain}. Moreover, if $e$ is $(\phi, \alpha\beta)$-successful, then $P(e; \phi, \alpha\beta)$ is $\phi$-happy \cite[Fact 4.5]{VizingChain}. See Fig.~\ref{fig:path} for an illustration.
\end{enumerate}

\subsection{Fan chains}

The second type of chains we will be working with are fans:

\begin{defn}[Fans]
    A \emphd{fan} is a chain of the form $F = (xy_0, \ldots, xy_{k-1})$ where $x$ is referred to as the \emphd{pivot} of the fan and $y_0$, \ldots, $y_{k-1}$ are distinct neighbors of $x$. We let $\Pivot(F) \defeq x$, $\vstart(F) \defeq y_0$ and $\vend(F) \defeq y_{k-1}$ denote the pivot, start, and end vertices of a fan $F$. (This notation is uniquely determined unless $F$ is a single edge.)
\end{defn}

The process of shifting a fan is illustrated in Fig.~\ref{fig:fan}.

\begin{figure}[htb!]
	\centering
	\begin{tikzpicture}[xscale=1.1]
	\begin{scope}
	\node[circle,fill=black,draw,inner sep=0pt,minimum size=4pt] (x) at (0,0) {};
	\node[anchor=north] at (x) {$x$};
	
	\coordinate (O) at (0,0);
	\def\radius{2.6cm}
	
	\node[circle,fill=black,draw,inner sep=0pt,minimum size=4pt] (y0) at (190:\radius) {};
	\node at (190:2.9) {$y_0$};
	
	\node[circle,fill=black,draw,inner sep=0pt,minimum size=4pt] (y1) at (165:\radius) {};
	\node at (165:2.9) {$y_1$};
	
	\node[circle,fill=black,draw,inner sep=0pt,minimum size=4pt] (y2) at (140:\radius) {};
	\node at (140:2.9) {$y_2$};
	
	\node[circle,fill=black,draw,inner sep=0pt,minimum size=4pt] (y4) at (90:\radius) {};
	\node at (90:2.9) {$y_{i-1}$};
	
	\node[circle,fill=black,draw,inner sep=0pt,minimum size=4pt] (y5) at (65:\radius) {};
	\node at (65:2.9) {$y_i$};
	
	\node[circle,fill=black,draw,inner sep=0pt,minimum size=4pt] (y6) at (40:\radius) {};
	\node at (40:3) {$y_{i+1}$};
	
	\node[circle,fill=black,draw,inner sep=0pt,minimum size=4pt] (y8) at (-10:\radius) {};
	\node at (-10:3.1) {$y_{k-1}$};
	
	\node[circle,inner sep=0pt,minimum size=4pt] at (115:2) {$\ldots$}; 
	\node[circle,inner sep=0pt,minimum size=4pt] at (15:2) {$\ldots$}; 
	
	\draw[thick,dotted] (x) to (y0);
	\draw[thick] (x) to node[midway,inner sep=1pt,outer sep=1pt,minimum size=4pt,fill=white] {$\beta_0$} (y1);
	\draw[thick] (x) to node[midway,inner sep=1pt,outer sep=1pt,minimum size=4pt,fill=white] {$\beta_1$} (y2);
	
	\draw[thick] (x) to node[midway,inner sep=1pt,outer sep=1pt,minimum size=4pt,fill=white] {$\beta_{i-2}$} (y4);
	\draw[thick] (x) to node[pos=0.75,inner sep=1pt,outer sep=1pt,minimum size=4pt,fill=white] {$\beta_{i-1}$} (y5);
	\draw[thick] (x) to node[midway,inner sep=1pt,outer sep=1pt,minimum size=4pt,fill=white] {$\beta_i$} (y6);
	
	\draw[thick] (x) to node[midway,inner sep=1pt,outer sep=1pt,minimum size=4pt,fill=white] {$\beta_{k-2}$} (y8);
	\end{scope}
	
	\draw[-{Stealth[length=1.6mm]},very thick,decoration = {snake,pre length=3pt,post length=7pt,},decorate] (2.9,1) to node[midway,anchor=south]{$\Shift$} (5,1);
	
	\begin{scope}[xshift=8.3cm]
	\node[circle,fill=black,draw,inner sep=0pt,minimum size=4pt] (x) at (0,0) {};
	\node[anchor=north] at (x) {$x$};
	
	\coordinate (O) at (0,0);
	\def\radius{2.6cm}
	
	\node[circle,fill=black,draw,inner sep=0pt,minimum size=4pt] (y0) at (190:\radius) {};
	\node at (190:2.9) {$y_0$};
	
	\node[circle,fill=black,draw,inner sep=0pt,minimum size=4pt] (y1) at (165:\radius) {};
	\node at (165:2.9) {$y_1$};
	
	\node[circle,fill=black,draw,inner sep=0pt,minimum size=4pt] (y2) at (140:\radius) {};
	\node at (140:2.9) {$y_2$};
	
	\node[circle,fill=black,draw,inner sep=0pt,minimum size=4pt] (y4) at (90:\radius) {};
	\node at (90:2.9) {$y_{i-1}$};
	
	\node[circle,fill=black,draw,inner sep=0pt,minimum size=4pt] (y5) at (65:\radius) {};
	\node at (65:2.9) {$y_i$};
	
	\node[circle,fill=black,draw,inner sep=0pt,minimum size=4pt] (y6) at (40:\radius) {};
	\node at (40:3) {$y_{i+1}$};
	
	\node[circle,fill=black,draw,inner sep=0pt,minimum size=4pt] (y8) at (-10:\radius) {};
	\node at (-10:3.1) {$y_{k-1}$};
	
	\node[circle,inner sep=0pt,minimum size=4pt] at (115:2) {$\ldots$}; 
	\node[circle,inner sep=0pt,minimum size=4pt] at (15:2) {$\ldots$}; 
	
	\draw[thick] (x) to node[midway,inner sep=1pt,outer sep=1pt,minimum size=4pt,fill=white] {$\beta_0$} (y0);
	\draw[thick] (x) to node[midway,inner sep=1pt,outer sep=1pt,minimum size=4pt,fill=white] {$\beta_1$} (y1);
	\draw[thick] (x) to node[midway,inner sep=1pt,outer sep=1pt,minimum size=4pt,fill=white] {$\beta_2$} (y2);
	
	\draw[thick] (x) to node[midway,inner sep=1pt,outer sep=1pt,minimum size=4pt,fill=white] {$\beta_{i-1}$} (y4);
	\draw[thick] (x) to node[pos=0.75,inner sep=1pt,outer sep=1pt,minimum size=4pt,fill=white] {$\beta_i$} (y5);
	\draw[thick] (x) to node[midway,inner sep=1pt,outer sep=1pt,minimum size=4pt,fill=white] {$\beta_{i+1}$} (y6);
	
	\draw[thick, dotted] (x) to (y8);
	\end{scope}
	\end{tikzpicture}
	\caption{The process of shifting a fan.}\label{fig:fan}
\end{figure}
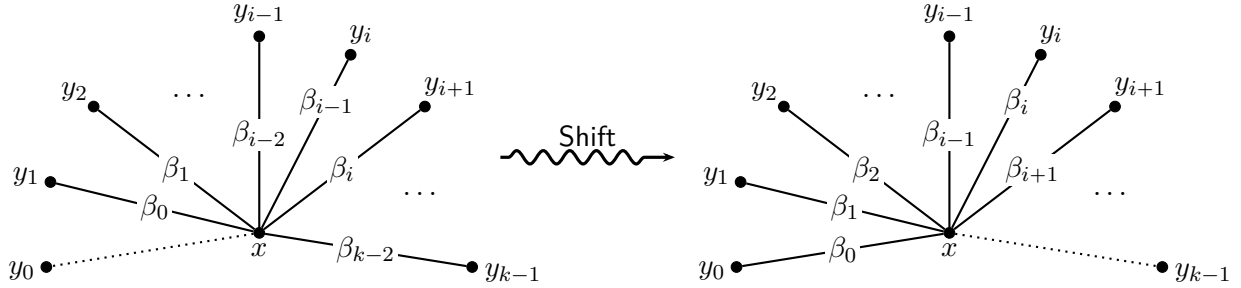


\begin{defn}[Hopeful, successful, and disappointed fans]\label{defn:hsf}
    Let $\phi$ be a proper partial coloring and let $\alpha$, $\beta \in [q]$. Let $F$ be a $\phi$-shiftable fan with $x \defeq \Pivot(F)$ and $y \defeq \vend(F)$ and suppose that $F$ is not $\phi$-happy (which means that the edge $xy$ is not $\Shift(\phi, F)$-happy). We say that $F$ is:
    \begin{itemize}
        \item \emphd{$(\phi, \alpha\beta)$-hopeful} if $\deg(x;\phi, \alpha\beta) < 2$ and $\deg(y;\phi, \alpha\beta) < 2$;
        \item \emphd{$(\phi, \alpha\beta)$-successful} if it is $(\phi, \alpha\beta)$-hopeful and $x$ and $y$ are not $(\Shift(\phi, F), \alpha\beta)$-related;
        \item \emphd{$(\phi, \alpha\beta)$-disappointed} if it is $(\phi, \alpha\beta)$-hopeful but not $(\phi, \alpha\beta)$-successful.
    \end{itemize}
\end{defn}

Note that, in the setting of Definition~\ref{defn:hsf}, we have $M(\phi, x) = M(\Shift(\phi, F), x)$ and $M(\phi, y) \subseteq M(\Shift(\phi, F), y)$. Therefore, if $F$ is $(\phi, \alpha\beta)$-hopeful (resp. successful), then
\[\deg(x;\phi, \alpha\beta) = \deg(x;\Shift(\phi, F), \alpha\beta) < 2, \quad \deg(y;\Shift(\phi, F), \alpha\beta) \leq \deg(y;\phi, \alpha\beta) < 2,\]
and so $xy$ is $(\Shift(\phi, F), \alpha\beta)$-hopeful (resp. successful) \cite[Fact 4.7]{VizingChain}.

\subsection{Multi-Step Vizing chains}\label{subsec:Vizingdefn}

The third type of chain we will be working with are Vizing chains, which are formed by combining a fan and an alternating path:

\begin{defn}[Vizing chains]
    A \emphd{Vizing chain} in a proper partial coloring $\phi$ is a chain of the form $F + P$, where $F$ is a $(\phi, \alpha\beta)$-hopeful fan for some $\alpha$, $\beta \in [q]$ and $P$ is an initial segment of the path chain $P(\End(F);  \Shift(\phi, F), \alpha\beta)$ with $\vstart(P) = \Pivot(F)$.
    In particular, letting $x \defeq \Pivot(F)$ and $y \defeq \vend(F)$, the chain $P$ consists of the edge $xy$ followed by a (not necessarily maximal) path starting at $y$ whose edges are colored $\alpha$ and $\beta$ under the coloring $\Shift(\phi, F)$ such that $\beta \in M(\Shift(\phi, F), y)$. See Fig.~\ref{fig:Viz_col} for an illustration. Note that we allow $P$ to comprise only the single edge $\End(F)$, in which case the Vizing chain coincides with the fan $F$.
\end{defn}

\begin{figure}[ht]
	\centering
	\begin{subfigure}[t]{.2\textwidth}
		\centering
		\begin{tikzpicture}[scale=1.2]
		\node[circle,fill=black,draw,inner sep=0pt,minimum size=4pt] (a) at (0,0) {};
		\node[circle,fill=black,draw,inner sep=0pt,minimum size=4pt] (b) at (1,0) {};
		\path (b) ++(150:1) node[circle,fill=black,draw,inner sep=0pt,minimum size=4pt] (c) {}; 
		\path (b) ++(120:1) node[circle,fill=black,draw,inner sep=0pt,minimum size=4pt] (d) {}; 
		\node[circle,fill=black,draw,inner sep=0pt,minimum size=4pt] (e) at (1,1) {};
		\node[circle,fill=black,draw,inner sep=0pt,minimum size=4pt] (f) at (1.8,1.4) {};
		\node[circle,fill=black,draw,inner sep=0pt,minimum size=4pt] (g) at (1,1.8) {};
		\node[circle,fill=black,draw,inner sep=0pt,minimum size=4pt] (h) at (1.8,2.2) {};
		\node[circle,fill=black,draw,inner sep=0pt,minimum size=4pt] (i) at (1,2.6) {};
		\node[circle,fill=black,draw,inner sep=0pt,minimum size=4pt] (j) at (1.8,3) {};
		\node[circle,fill=black,draw,inner sep=0pt,minimum size=4pt] (k) at (1,3.4) {};
		\node[circle,fill=black,draw,inner sep=0pt,minimum size=4pt] (l) at (1.8,3.8) {};
		\node[circle,fill=black,draw,inner sep=0pt,minimum size=4pt] (m) at (1,4.2) {};
		\node[circle,fill=black,draw,inner sep=0pt,minimum size=4pt] (n) at (1.8,4.6) {};
		
		\draw[ thick,dotted] (a) -- (b);
		\draw[ thick] (b) to node[midway,inner sep=1pt,outer sep=1pt,minimum size=4pt,fill=white] {$\epsilon$} (c) (b) to node[midway,inner sep=0.5pt,outer sep=0.5pt,minimum size=4pt,fill=white] {$\zeta$} (d) (b) to node[midway,inner sep=1pt,outer sep=1pt,minimum size=4pt,fill=white] {$\eta$} (e) to node[midway,inner sep=1pt,outer sep=1pt,minimum size=4pt,fill=white] {$\alpha$} (f) to node[midway,inner sep=1pt,outer sep=1pt,minimum size=4pt,fill=white] {$\beta$} (g) to node[midway,inner sep=1pt,outer sep=1pt,minimum size=4pt,fill=white] {$\alpha$} (h) to node[midway,inner sep=1pt,outer sep=1pt,minimum size=4pt,fill=white] {$\beta$} (i) to node[midway,inner sep=1pt,outer sep=1pt,minimum size=4pt,fill=white] {$\alpha$} (j) to node[midway,inner sep=1pt,outer sep=1pt,minimum size=4pt,fill=white] {$\beta$} (k) to node[midway,inner sep=1pt,outer sep=1pt,minimum size=4pt,fill=white] {$\alpha$} (l) to node[midway,inner sep=1pt,outer sep=1pt,minimum size=4pt,fill=white] {$\beta$} (m) to node[midway,inner sep=1pt,outer sep=1pt,minimum size=4pt,fill=white] {$\alpha$} (n);
		
		\draw[decoration={brace,amplitude=10pt,mirror},decorate] (2,0.9) -- node [midway,below,xshift=15pt,yshift=5pt] {$P$} (2,4.7);
		
		\draw[decoration={brace,amplitude=10pt},decorate] (-0.35,0) -- node [midway,above,yshift=3pt,xshift=-10pt] {$F$} (1,1.35);
		
		\node[anchor=west, rotate=-20] at (e) {$y = \vend(F)$};
		\node[anchor=west] at (b) {$x = \Pivot(F)$};
		\end{tikzpicture}
		\caption{Before shifting.}\label{fig:Viz_col:initial}
	\end{subfigure}%
	\qquad\qquad\qquad%
	\begin{subfigure}[t]{.2\textwidth}
		\centering
		\begin{tikzpicture}[scale=1.2]
		\node[circle,fill=black,draw,inner sep=0pt,minimum size=4pt] (a) at (0,0) {};
		\node[circle,fill=black,draw,inner sep=0pt,minimum size=4pt] (b) at (1,0) {};
		\path (b) ++(150:1) node[circle,fill=black,draw,inner sep=0pt,minimum size=4pt] (c) {}; 
		\path (b) ++(120:1) node[circle,fill=black,draw,inner sep=0pt,minimum size=4pt] (d) {}; 
		\node[circle,fill=black,draw,inner sep=0pt,minimum size=4pt] (e) at (1,1) {};
		\node[circle,fill=black,draw,inner sep=0pt,minimum size=4pt] (f) at (1.8,1.4) {};
		\node[circle,fill=black,draw,inner sep=0pt,minimum size=4pt] (g) at (1,1.8) {};
		\node[circle,fill=black,draw,inner sep=0pt,minimum size=4pt] (h) at (1.8,2.2) {};
		\node[circle,fill=black,draw,inner sep=0pt,minimum size=4pt] (i) at (1,2.6) {};
		\node[circle,fill=black,draw,inner sep=0pt,minimum size=4pt] (j) at (1.8,3) {};
		\node[circle,fill=black,draw,inner sep=0pt,minimum size=4pt] (k) at (1,3.4) {};
		\node[circle,fill=black,draw,inner sep=0pt,minimum size=4pt] (l) at (1.8,3.8) {};
		\node[circle,fill=black,draw,inner sep=0pt,minimum size=4pt] (m) at (1,4.2) {};
		\node[circle,fill=black,draw,inner sep=0pt,minimum size=4pt] (n) at (1.8,4.6) {};
		
		\draw[ thick] (a) to node[midway,inner sep=1pt,outer sep=1pt,minimum size=4pt,fill=white] {$\epsilon$} (b);
		\draw[ thick] (b) to node[midway,inner sep=0.5pt,outer sep=0.5pt,minimum size=4pt,fill=white] {$\zeta$} (c) (b) to node[midway,inner sep=1pt,outer sep=1pt,minimum size=4pt,fill=white] {$\eta$} (d) (b) to node[midway,inner sep=1pt,outer sep=1pt,minimum size=4pt,fill=white] {$\alpha$} (e) to node[midway,inner sep=1pt,outer sep=1pt,minimum size=4pt,fill=white] {$\beta$} (f) to node[midway,inner sep=1pt,outer sep=1pt,minimum size=4pt,fill=white] {$\alpha$} (g) to node[midway,inner sep=1pt,outer sep=1pt,minimum size=4pt,fill=white] {$\beta$} (h) to node[midway,inner sep=1pt,outer sep=1pt,minimum size=4pt,fill=white] {$\alpha$} (i) to node[midway,inner sep=1pt,outer sep=1pt,minimum size=4pt,fill=white] {$\beta$} (j) to node[midway,inner sep=1pt,outer sep=1pt,minimum size=4pt,fill=white] {$\alpha$} (k) to node[midway,inner sep=1pt,outer sep=1pt,minimum size=4pt,fill=white] {$\beta$} (l) to node[midway,inner sep=1pt,outer sep=1pt,minimum size=4pt,fill=white] {$\alpha$} (m);
        \draw[ thick,dotted] (m) -- (n);
		\end{tikzpicture}
		\caption{After shifting.}\label{fig:Viz_col:shifted}
	\end{subfigure}%
	\caption{A Vizing chain $C = F + P$ before and after shifting.}\label{fig:Viz_col}
\end{figure}
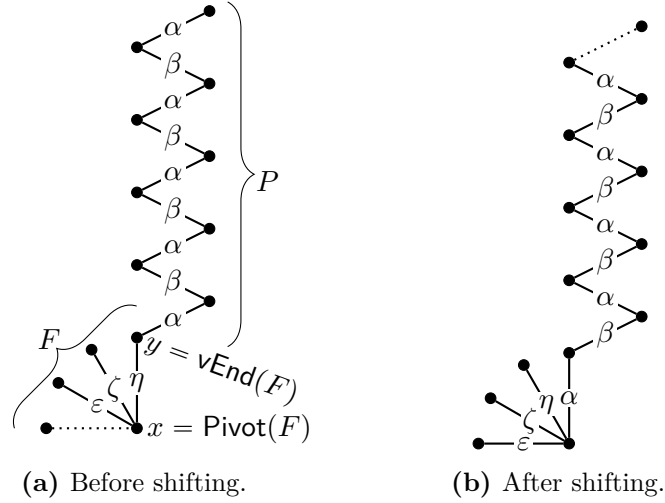

It follows from standard proofs of \hyperref[theo:Vizing]{Vizing's theorem} that for any uncolored edge $e$, one can find a $\phi$-happy Vizing chain $C = F+P$ with $\Start(F) = e$.
Combining several Vizing chains yields a multi-step Vizing chain:

\begin{defn}[Multi-step Vizing chains]
    A \emphd{$k$-step Vizing chain} is a chain of the form \[C \,=\, C_0 + \cdots + C_{k-1},\] where $C_i = F_i + P_i$ is a Vizing chain in the coloring $\Shift(\phi, C_0 + \cdots + C_{i-1})$ such that $\vend(P_i) = \vstart(F_{i+1})$ for all $0 \leq i < k-1$. See Fig.~\ref{fig:multi_Viz_chain} for an illustration.
\end{defn}

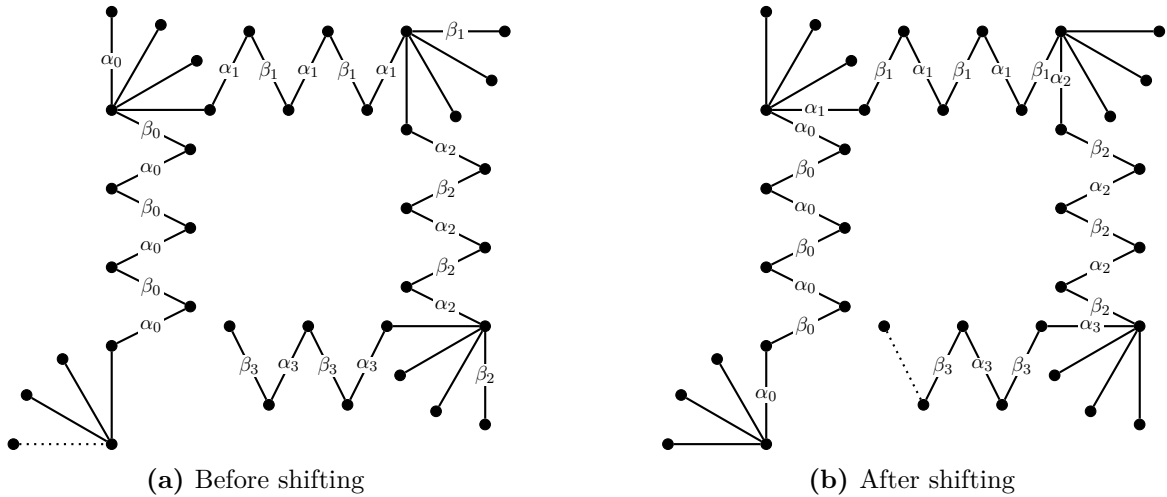
\begin{figure}[htb!]
    \begin{subfigure}[t]{.4\textwidth}
    	\centering
    	\begin{tikzpicture}[scale=1.3]
    	\node[circle,fill=black,draw,inner sep=0pt,minimum size=4pt] (a) at (0,0) {};
    	\node[circle,fill=black,draw,inner sep=0pt,minimum size=4pt] (b) at (1,0) {};
    	\path (b) ++(150:1) node[circle,fill=black,draw,inner sep=0pt,minimum size=4pt] (c) {};
    	\path (b) ++(120:1) node[circle,fill=black,draw,inner sep=0pt,minimum size=4pt] (d) {};
    	\node[circle,fill=black,draw,inner sep=0pt,minimum size=4pt] (e) at (1,1) {};
    	\node[circle,fill=black,draw,inner sep=0pt,minimum size=4pt] (f) at (1.8,1.4) {};
    	\node[circle,fill=black,draw,inner sep=0pt,minimum size=4pt] (g) at (1,1.8) {};
    	\node[circle,fill=black,draw,inner sep=0pt,minimum size=4pt] (h) at (1.8,2.2) {};
    	\node[circle,fill=black,draw,inner sep=0pt,minimum size=4pt] (i) at (1,2.6) {};
    	\node[circle,fill=black,draw,inner sep=0pt,minimum size=4pt] (j) at (1.8,3) {};
    	\node[circle,fill=black,draw,inner sep=0pt,minimum size=4pt] (k) at (1,3.4) {};
    	\node[circle,fill=black,draw,inner sep=0pt,minimum size=4pt] (l) at (1,4.4) {};
    	
    	\path (k) ++(60:1) node[circle,fill=black,draw,inner sep=0pt,minimum size=4pt] (m) {};
    	\path (k) ++(30:1) node[circle,fill=black,draw,inner sep=0pt,minimum size=4pt] (n) {};
    	\path (k) ++(0:1) node[circle,fill=black,draw,inner sep=0pt,minimum size=4pt] (o) {};
    	
    	\node[circle,fill=black,draw,inner sep=0pt,minimum size=4pt] (p) at (2.4,4.2) {};
    	\node[circle,fill=black,draw,inner sep=0pt,minimum size=4pt] (q) at (2.8,3.4) {};
    	\node[circle,fill=black,draw,inner sep=0pt,minimum size=4pt] (r) at (3.2,4.2) {};
    	\node[circle,fill=black,draw,inner sep=0pt,minimum size=4pt] (s) at (3.6,3.4) {};
    	\node[circle,fill=black,draw,inner sep=0pt,minimum size=4pt] (t) at (4,4.2) {};

    	\node[circle,fill=black,draw,inner sep=0pt,minimum size=4pt] (v) at (5,4.2) {};
    	\path (t) ++(-30:1) node[circle,fill=black,draw,inner sep=0pt,minimum size=4pt] (w) {};
    	\path (t) ++(-60:1) node[circle,fill=black,draw,inner sep=0pt,minimum size=4pt] (x) {};
    	\path (t) ++(-90:1) node[circle,fill=black,draw,inner sep=0pt,minimum size=4pt] (y) {};

    	\node[circle,fill=black,draw,inner sep=0pt,minimum size=4pt] (z) at (4.8,2.8) {};
    	\node[circle,fill=black,draw,inner sep=0pt,minimum size=4pt] (aa) at (4,2.4) {};
    	\node[circle,fill=black,draw,inner sep=0pt,minimum size=4pt] (ab) at (4.8,2) {};
    	\node[circle,fill=black,draw,inner sep=0pt,minimum size=4pt] (ac) at (4,1.6) {};
    	\node[circle,fill=black,draw,inner sep=0pt,minimum size=4pt] (ad) at (4.8,1.2) {};
    	\node[circle,fill=black,draw,inner sep=0pt,minimum size=4pt] (ae) at (4.8,0.2) {};
    	
    	\path (ad) ++(-120:1) node[circle,fill=black,draw,inner sep=0pt,minimum size=4pt] (ag) {};
    	\path (ad) ++(-150:1) node[circle,fill=black,draw,inner sep=0pt,minimum size=4pt] (ah) {};
    	\path (ad) ++(-180:1) node[circle,fill=black,draw,inner sep=0pt,minimum size=4pt] (ai) {};
    	
    	\node[circle,fill=black,draw,inner sep=0pt,minimum size=4pt] (aj) at (3.4,0.4) {};
    	\node[circle,fill=black,draw,inner sep=0pt,minimum size=4pt] (ak) at (3,1.2) {};
    	\node[circle,fill=black,draw,inner sep=0pt,minimum size=4pt] (al) at (2.6,0.4) {};
    	\node[circle,fill=black,draw,inner sep=0pt,minimum size=4pt] (am) at (2.2,1.2) {};
    	
    	\begin{scope}[every node/.style={scale=0.7}]
    	\draw[ thick,dotted] (a) -- (b);
    	\draw[ thick] (b) -- (c) (b) -- (d) (b) -- (e) to node[midway,inner sep=1pt,outer sep=1pt,minimum size=4pt,fill=white] {$\alpha_0$} (f) to node[midway,inner sep=1pt,outer sep=1pt,minimum size=4pt,fill=white] {$\beta_0$} (g) to node[midway,inner sep=1pt,outer sep=1pt,minimum size=4pt,fill=white] {$\alpha_0$} (h) to node[midway,inner sep=1pt,outer sep=1pt,minimum size=4pt,fill=white] {$\beta_0$} (i) to node[midway,inner sep=1pt,outer sep=1pt,minimum size=4pt,fill=white] {$\alpha_0$} (j) to node[midway,inner sep=1pt,outer sep=1pt,minimum size=4pt,fill=white] {$\beta_0$} (k) to node[midway,inner sep=1pt,outer sep=1pt,minimum size=4pt,fill=white] {$\alpha_0$} (l);
    	
    	\draw[ thick] (k) -- (m) (k) -- (n) (k) -- (o) to node[midway,inner sep=1pt,outer sep=1pt,minimum size=4pt,fill=white] {$\alpha_1$} (p) to node[midway,inner sep=1pt,outer sep=1pt,minimum size=4pt,fill=white] {$\beta_1$} (q) to node[midway,inner sep=1pt,outer sep=1pt,minimum size=4pt,fill=white] {$\alpha_1$} (r) to node[midway,inner sep=1pt,outer sep=1pt,minimum size=4pt,fill=white] {$\beta_1$} (s) to node[midway,inner sep=1pt,outer sep=1pt,minimum size=4pt,fill=white] {$\alpha_1$} (t) to node[midway,inner sep=1pt,outer sep=1pt,minimum size=4pt,fill=white] {$\beta_1$} (v);
    	
    	\draw[ thick] (t) -- (w) (t) -- (x) (t) -- (y) to node[midway,inner sep=1pt,outer sep=1pt,minimum size=4pt,fill=white] {$\alpha_2$} (z) to node[midway,inner sep=1pt,outer sep=1pt,minimum size=4pt,fill=white] {$\beta_2$} (aa) to node[midway,inner sep=1pt,outer sep=1pt,minimum size=4pt,fill=white] {$\alpha_2$} (ab) to node[midway,inner sep=1pt,outer sep=1pt,minimum size=4pt,fill=white] {$\beta_2$} (ac) to node[midway,inner sep=1pt,outer sep=1pt,minimum size=4pt,fill=white] {$\alpha_2$} (ad) to node[midway,inner sep=1pt,outer sep=1pt,minimum size=4pt,fill=white] {$\beta_2$} (ae);
    	
    	\draw[thick] (ad) -- (ag) (ad) -- (ah) (ad) -- (ai) to node[midway,inner sep=1pt,outer sep=1pt,minimum size=4pt,fill=white] {$\alpha_3$} (aj) to node[midway,inner sep=1pt,outer sep=1pt,minimum size=4pt,fill=white] {$\beta_3$} (ak) to node[midway,inner sep=1pt,outer sep=1pt,minimum size=4pt,fill=white] {$\alpha_3$} (al) to node[midway,inner sep=1pt,outer sep=1pt,minimum size=4pt,fill=white] {$\beta_3$} (am);
        \end{scope}
    	
    	\end{tikzpicture}
    	\caption{Before shifting}
    \end{subfigure}
    \qquad \qquad
    \begin{subfigure}[t]{.45\textwidth}
    	\centering
    	\begin{tikzpicture}[scale=1.3]
    	\node[circle,fill=black,draw,inner sep=0pt,minimum size=4pt] (a) at (0,0) {};
    	\node[circle,fill=black,draw,inner sep=0pt,minimum size=4pt] (b) at (1,0) {};
    	\path (b) ++(150:1) node[circle,fill=black,draw,inner sep=0pt,minimum size=4pt] (c) {};
    	\path (b) ++(120:1) node[circle,fill=black,draw,inner sep=0pt,minimum size=4pt] (d) {};
    	\node[circle,fill=black,draw,inner sep=0pt,minimum size=4pt] (e) at (1,1) {};
    	\node[circle,fill=black,draw,inner sep=0pt,minimum size=4pt] (f) at (1.8,1.4) {};
    	\node[circle,fill=black,draw,inner sep=0pt,minimum size=4pt] (g) at (1,1.8) {};
    	\node[circle,fill=black,draw,inner sep=0pt,minimum size=4pt] (h) at (1.8,2.2) {};
    	\node[circle,fill=black,draw,inner sep=0pt,minimum size=4pt] (i) at (1,2.6) {};
    	\node[circle,fill=black,draw,inner sep=0pt,minimum size=4pt] (j) at (1.8,3) {};
    	\node[circle,fill=black,draw,inner sep=0pt,minimum size=4pt] (k) at (1,3.4) {};
    	\node[circle,fill=black,draw,inner sep=0pt,minimum size=4pt] (l) at (1,4.4) {};
    	
    	\path (k) ++(60:1) node[circle,fill=black,draw,inner sep=0pt,minimum size=4pt] (m) {};
    	\path (k) ++(30:1) node[circle,fill=black,draw,inner sep=0pt,minimum size=4pt] (n) {};
    	\path (k) ++(0:1) node[circle,fill=black,draw,inner sep=0pt,minimum size=4pt] (o) {};
    	
    	\node[circle,fill=black,draw,inner sep=0pt,minimum size=4pt] (p) at (2.4,4.2) {};
    	\node[circle,fill=black,draw,inner sep=0pt,minimum size=4pt] (q) at (2.8,3.4) {};
    	\node[circle,fill=black,draw,inner sep=0pt,minimum size=4pt] (r) at (3.2,4.2) {};
    	\node[circle,fill=black,draw,inner sep=0pt,minimum size=4pt] (s) at (3.6,3.4) {};
    	\node[circle,fill=black,draw,inner sep=0pt,minimum size=4pt] (t) at (4,4.2) {};

    	\node[circle,fill=black,draw,inner sep=0pt,minimum size=4pt] (v) at (5,4.2) {};
    	\path (t) ++(-30:1) node[circle,fill=black,draw,inner sep=0pt,minimum size=4pt] (w) {};
    	\path (t) ++(-60:1) node[circle,fill=black,draw,inner sep=0pt,minimum size=4pt] (x) {};
    	\path (t) ++(-90:1) node[circle,fill=black,draw,inner sep=0pt,minimum size=4pt] (y) {};

    	\node[circle,fill=black,draw,inner sep=0pt,minimum size=4pt] (z) at (4.8,2.8) {};
    	\node[circle,fill=black,draw,inner sep=0pt,minimum size=4pt] (aa) at (4,2.4) {};
    	\node[circle,fill=black,draw,inner sep=0pt,minimum size=4pt] (ab) at (4.8,2) {};
    	\node[circle,fill=black,draw,inner sep=0pt,minimum size=4pt] (ac) at (4,1.6) {};
    	\node[circle,fill=black,draw,inner sep=0pt,minimum size=4pt] (ad) at (4.8,1.2) {};
    	\node[circle,fill=black,draw,inner sep=0pt,minimum size=4pt] (ae) at (4.8,0.2) {};
    	
    	\path (ad) ++(-120:1) node[circle,fill=black,draw,inner sep=0pt,minimum size=4pt] (ag) {};
    	\path (ad) ++(-150:1) node[circle,fill=black,draw,inner sep=0pt,minimum size=4pt] (ah) {};
    	\path (ad) ++(-180:1) node[circle,fill=black,draw,inner sep=0pt,minimum size=4pt] (ai) {};
    	
    	\node[circle,fill=black,draw,inner sep=0pt,minimum size=4pt] (aj) at (3.4,0.4) {};
    	\node[circle,fill=black,draw,inner sep=0pt,minimum size=4pt] (ak) at (3,1.2) {};
    	\node[circle,fill=black,draw,inner sep=0pt,minimum size=4pt] (al) at (2.6,0.4) {};
    	\node[circle,fill=black,draw,inner sep=0pt,minimum size=4pt] (am) at (2.2,1.2) {};
    	
    	\begin{scope}[every node/.style={scale=0.7}]
    	\draw[ thick] (a) -- (b) -- (c) (b) -- (d) (b) to node[midway,inner sep=1pt,outer sep=1pt,minimum size=4pt,fill=white] {$\alpha_0$} (e) to node[midway,inner sep=1pt,outer sep=1pt,minimum size=4pt,fill=white] {$\beta_0$} (f) to node[midway,inner sep=1pt,outer sep=1pt,minimum size=4pt,fill=white] {$\alpha_0$} (g) to node[midway,inner sep=1pt,outer sep=1pt,minimum size=4pt,fill=white] {$\beta_0$} (h) to node[midway,inner sep=1pt,outer sep=1pt,minimum size=4pt,fill=white] {$\alpha_0$} (i) to node[midway,inner sep=1pt,outer sep=1pt,minimum size=4pt,fill=white] {$\beta_0$} (j) to node[midway,inner sep=1pt,outer sep=1pt,minimum size=4pt,fill=white] {$\alpha_0$} (k);
    	
    	\draw[ thick] (l) -- (k) -- (m) (k) -- (n) (k) to node[midway,inner sep=1pt,outer sep=1pt,minimum size=4pt,fill=white] {$\alpha_1$} (o) to node[midway,inner sep=1pt,outer sep=1pt,minimum size=4pt,fill=white] {$\beta_1$} (p) to node[midway,inner sep=1pt,outer sep=1pt,minimum size=4pt,fill=white] {$\alpha_1$} (q) to node[midway,inner sep=1pt,outer sep=1pt,minimum size=4pt,fill=white] {$\beta_1$} (r) to node[midway,inner sep=1pt,outer sep=1pt,minimum size=4pt,fill=white] {$\alpha_1$} (s) to node[midway,inner sep=1pt,outer sep=1pt,minimum size=4pt,fill=white] {$\beta_1$} (t);
    	
    	\draw[ thick] (v) -- (t) -- (w) (t) -- (x) (t) to node[midway,inner sep=1pt,outer sep=1pt,minimum size=4pt,fill=white] {$\alpha_2$} (y) to node[midway,inner sep=1pt,outer sep=1pt,minimum size=4pt,fill=white] {$\beta_2$} (z) to node[midway,inner sep=1pt,outer sep=1pt,minimum size=4pt,fill=white] {$\alpha_2$} (aa) to node[midway,inner sep=1pt,outer sep=1pt,minimum size=4pt,fill=white] {$\beta_2$} (ab) to node[midway,inner sep=1pt,outer sep=1pt,minimum size=4pt,fill=white] {$\alpha_2$} (ac) to node[midway,inner sep=1pt,outer sep=1pt,minimum size=4pt,fill=white] {$\beta_2$} (ad);
    	
    	\draw[thick] (ae) -- (ad) -- (ag) (ad) -- (ah) (ad) to node[midway,inner sep=1pt,outer sep=1pt,minimum size=4pt,fill=white] {$\alpha_3$} (ai) to node[midway,inner sep=1pt,outer sep=1pt,minimum size=4pt,fill=white] {$\beta_3$} (aj) to node[midway,inner sep=1pt,outer sep=1pt,minimum size=4pt,fill=white] {$\alpha_3$} (ak) to node[midway,inner sep=1pt,outer sep=1pt,minimum size=4pt,fill=white] {$\beta_3$} (al);
            \draw[thick, dotted] (al) -- (am);
        \end{scope}
    	
    	\end{tikzpicture}
    	\caption{After shifting}
    \end{subfigure}
    \caption{A 4-step Vizing chain before and after shifting.}\label{fig:multi_Viz_chain}
\end{figure}

\section{Data Structures}\label{section: data_structures}

In this section, we describe how we will store our graph $G$ and partial coloring $\phi$.
Below, we discuss how our choices for the data structures affect the runtime of certain procedures.

\begin{itemize}
    \item We store $G$ as a list of vertices and edges, and include the partial coloring $\phi$ as an attribute of the graph.

    \item We store the partial coloring $\phi$ as a hash map, which maps edges to their respective colors (or $\blank$ if the edge is uncolored).
    Furthermore, the missing sets $M(\phi, \cdot)$ are also stored as hash maps, which map a vertex $x$ to a $q$-element array such that the following holds for each $\alpha \in [q]$:
    \[M(\phi, x)[\alpha] = \left\{\begin{array}{cc}
        y & \text{such that $\phi(xy) = \alpha$;} \\
        \blank & \text{if no such $y$ exists.}
    \end{array}\right.\]
    Note that as $\phi$ is a proper partial coloring, the vertex $y$ above is unique.
    In the remainder of the paper, we will use the notation $M(\phi, x)[\cdot]$ as described above in our algorithms, and the notation $M(\phi, x)$ to indicate the set of missing colors at $x$ in our proofs.
\end{itemize}

In the following algorithm, we formally describe the procedure of shifting a $\phi$-shiftable chain $C$.

\vspace{10pt}
\begin{breakablealgorithm}
\caption{Shifting a Chain}\algsize\label{alg:chain_shift}
\begin{flushleft}
\textbf{Input}: A partial coloring $\phi$ and a $\phi$-shiftable chain $C=(e_0, \ldots, e_{k-1})$. \\
\textbf{Output}: The coloring $\Shift(\phi, C)$.
\end{flushleft}
\begin{algorithmic}[1]
    \For{$i = 0, \ldots, k-2$}
        \State Let $x, y, z \in V(G)$ be such that $e_i = xy$ and $e_{i+1} = yz$.
        \State $\phi(e_i) \gets \phi(e_{i+1})$, \quad $\phi(e_{i+1}) \gets \blank$.
        \State $M(\phi, x)[\phi(e_i)] \gets y$, \quad $M(\phi, y)[\phi(e_i)] \gets x$, \quad $M(\phi, z)[\phi(e_i)] \gets \blank$.
    \EndFor
    \State \Return $\phi$.
\end{algorithmic}
\end{breakablealgorithm}
\vspace{10pt}

By our choice of data structures, each iteration of the \textsf{for} loop takes $O(1)$ time and so Algorithm~\ref{alg:chain_shift} has running time $O(\length(C))$.

Next, we show how to augment a $\phi$-happy chain $C$.

\vspace{10pt}
\begin{breakablealgorithm}
\caption{Augmenting a Chain}\algsize\label{alg:chain_aug}
\begin{flushleft}
\textbf{Input}: A partial coloring $\phi$, a $\phi$-happy chain $C=(e_0, \ldots, e_{k-1})$, and a color $\xi$ such that $\End(C)$ can be assigned $\xi$ in the coloring $\Shift(\phi, C)$. \\
\textbf{Output}: The coloring $\aug(\phi, C, \xi)$.
\end{flushleft}
\begin{algorithmic}[1]
    \State $\phi \gets \hyperref[alg:chain_shift]{\Shift}(\phi, C)$ \Comment{Algorithm~\ref{alg:chain_shift}}
    \State $\phi(\End(C)) \gets \xi$
    \State Let $x, y \in V(G)$ such that $\End(C) = xy$.
    \State $M(\phi, x)[\xi] \gets y$, \quad $M(\phi, y)[\xi] \gets x$
    \State \Return $\phi$.
\end{algorithmic}
\end{breakablealgorithm}
\vspace{10pt}

A similar argument as before shows that Algorithm~\ref{alg:chain_aug} has running time $O(\length(C))$.

\section{The Multi-Step Vizing Algorithm}\label{sec:MSVC}

In this section we will describe our \hyperref[alg:multi_viz_chain]{Multi-Step Vizing Algorithm} (MSVA for short), which constructs a $\phi$-happy multi-step Vizing chain given a partial coloring $\phi$, an uncolored edge $e$, and a vertex $x \in e$. This is a rigorous version of the sketch presented as Algorithm~\ref{inf:MSVC} in the introduction.
We split the section into two subsections. 
In the first, we provide an overview of the algorithm, and in the second, we prove the correctness of the \hyperref[alg:multi_viz_chain]{MSVA}.

\subsection{Algorithm Overview}

In this section we will provide an overview of the \hyperref[alg:multi_viz_chain]{MSVA}.
In \S\ref{section: sequential}, we will iteratively apply this procedure and bound the runtime in terms of $m$, $\epsilon$, and a parameter $\ell$ to be specified later.
For now, we may assume $\ell$ is a polynomial in $1/\eps$ of sufficiently high degree.

The first procedure we describe is the one we use to find a random color missing at a given vertex.
The algorithm takes as input a partial coloring $\phi$, a vertex $x$, and a color $\theta$ and outputs a random color from $M(\phi, x) \setminus\set{\theta}$.
As we will see in certain procedures, we have restrictions on the color being chosen at specific vertices.
The color $\theta$ will play the role of the ``forbidden color.''
We remark that in our application of Algorithm~\ref{alg:rand_col}, whenever $\theta \neq \blank$, there is an uncolored edge incident to $x$ under $\phi$.
In particular, $M(\phi, x) \setminus\set{\theta} \neq \0$ for $\eps \geq 1/\Delta$.
The algorithm repeatedly picks a color from $[q]$ uniformly at random until it picks one from $M(\phi, x) \setminus\set{\theta}$.

\begin{algorithm}[h]
\caption{Random Missing Color}\algsize\label{alg:rand_col}
\begin{flushleft}
\textbf{Input}: A partial coloring $\phi$, a vertex $x$, and a color $\theta \in [q] \cup \set{\blank}$. \\
\textbf{Output}: A uniformly random color in $M(\phi, x) \setminus \set{\theta}$.
\end{flushleft}
\begin{algorithmic}[1]
    \While{true}
        \State Choose $\eta \in [q]$ uniformly at random.
        \If{$\eta \neq \theta$ and $M(\phi, x)[\eta] = \blank$}
            \State \Return $\eta$
        \EndIf
    \EndWhile
\end{algorithmic}
\end{algorithm}

Let us show that Algorithm~\ref{alg:rand_col} returns a uniformly random color over the set of interest.

\begin{Lemma}\label{lemma:unif_rand_color}
    Let $\eta$ be the output after running Algorithm~\ref{alg:rand_col} on input $(\phi, x, \theta)$. Assuming $M(\phi,x) \setminus \set{\theta} \neq \0$,  
    for any $\alpha \in M(\phi, x) \setminus\set{\theta}$, we have
    \[\P[\eta = \alpha] = \frac{1}{|M(\phi, x) \setminus\set{\theta}|}.\]
\end{Lemma}

\begin{proof}
    It is clear that $\eta$ must belong to $M(\phi, x)\setminus \set{\theta}$ and that the algorithm outputs every color in $M(\phi, x)\setminus \set{\theta}$ with the same probability.
\end{proof}

We may also bound the probability the \textsf{while} loop in Algorithm~\ref{alg:rand_col} lasts at least $t$ iterations.

\begin{Lemma}\label{lemma:rand_color_runtime}
    Consider running Algorithm~\ref{alg:rand_col} on input $(\phi, x, \theta)$ such that if $\theta \neq \blank$ then $x$ is incident to an uncolored edge.
    Then,
    \[\P[\text{the \textsf{while} loop lasts at least $t$ iterations}] \,\leq \left(1 - \frac{\eps}{2}\right)^{t-1}.\]
\end{Lemma}

\begin{proof}
    Let $S = M(\phi, x) \setminus\set{\theta}$.
    We have
    \[\P[\text{the \textsf{while} loop lasts at least $t$ iterations}] \,=\, \left(1 - \frac{|S|}{q}\right)^{t-1}.\]
    Note that $|S| \geq \eps \Delta$ as $\theta \neq \blank$ implies $x$ is incident to an uncolored edge.
    In particular, since $\eps \leq 1$, we have
    \[|S| \geq \eps\Delta \,\geq\, \frac{\eps\,q}{2},\]
    completing the proof.
\end{proof}

The next algorithm we describe is a procedure to construct fans that will be used as a subroutine for the \hyperref[alg:multi_viz_chain]{MSVA}.
Described formally in Algorithm~\ref{alg:rand_fan}, the \hyperref[alg:rand_fan]{Random Fan Algorithm} takes as input a partial coloring $\phi$, an edge $e = xy$, a vertex $x \in e$, and a color $\beta \in M(\phi, y) \cup \set{\blank}$.
The output is a tuple $(F, \delta, j)$, which satisfies certain properties we describe and prove in the next subsection.
To construct $F$, we follow a series of iterations.
At the start of an iteration, we have $F = (xy_0, \ldots, xy_k)$ where $y_0 = y$.
We pick a random color $\eta \in M(\phi, y_k)$ with the added restriction that $\eta \neq \beta$ when $k = 0$.
If $\eta \in M(\phi, x)$, then $F$ is $\phi$-happy and we return $(F, \eta, k+1)$.
Similarly, if $\eta = \beta$, we return $(F, \eta, k+1)$.
Now, we check whether $\eta \in M(\phi, y_j)$ for some $0 \leq j < k$.
If so, we return $(F, \eta, j+1)$.
If not, we let $y_{k+1}$ be the unique neighbor of $x$ such that $\phi(xy_{k+1}) = \eta$ and continue iterating.
We remark that when $\beta = \blank$, the algorithm is similar to the \textsf{First Fan Algorithm} of \cite[Algorithm 3.1]{bernshteyn2023fast}, while for $\beta \neq \blank$, it is similar to the \textsf{Next Fan Algorithm} of \cite[Algorithm~3.2]{bernshteyn2023fast}.
As we will see, we invoke this procedure with input $\beta = \blank$ precisely when constructing the first Vizing chain in the \hyperref[alg:multi_viz_chain]{MSVA}, and $\beta \neq \blank$ for every subsequent chain.

As mentioned in \S\ref{section: informal overview}, we require the length of the fan $F$ output by Algorithm~\ref{alg:rand_fan} to be at most $\poly(1/\eps)$ in order to avoid a dependence on $\Delta$ in our runtime.
To this end, we let $\kmax = \poly(1/\eps)$ (to be defined explicitly in \S\ref{section: sequential}) and ensure the length of $F$ is at most $\kmax$.
In particular, if the process described in the previous paragraph does not terminate within $\kmax$ iterations, we try again from the beginning.

\vspace{10pt}
\begin{breakablealgorithm}
\caption{Random Fan}\algsize\label{alg:rand_fan}
\begin{flushleft}
\textbf{Input}: A partial coloring $\phi$, an edge $e = xy$, a vertex $x \in e$, and a color $\beta \in M(\phi, y) \cup \set{\blank}$. \\
\textbf{Output}: A fan $F$ such that $\Start(F) = e$ and $\Pivot(F) = x$, a color $\delta \in [q]$, and an index $j$ such that $\delta \in M(\phi, \vend(F)) \cap M(\phi, \vend(F|j))$.
\end{flushleft}
\begin{algorithmic}[1]
    \State $F \gets (xy_0)$, \quad $k \gets 0$, \quad $y_0 \gets y$, \quad $\theta \gets \beta$ \label{step:start}
    \While{$k < \kmax$}
        \State $\eta \gets \hyperref[alg:rand_col]{\mathsf{RandomColor}}(\phi, y_k, \theta)$ \Comment{Algorithm \ref{alg:rand_col}} \label{step:missing_color_choice}
        \State $\theta \gets \blank$ \label{step:theta_blank}
        \If{$M(\phi, x)[\eta] = \blank$}
            \State \Return $(F, \eta, k+1)$ \label{step:happy_fan}
        \ElsIf{$\eta = \beta$}
            \State \Return $(F, \eta, k+1)$ \label{step:same_colors}
        \EndIf
        \For{$j = 1,\ldots,k$}\label{step: loop_in_fan}
            \If{$M(\phi, y_{j-1})[\eta] = \blank$}
                \State \Return $(F, \eta, j)$ \label{step:success_fan}
            \EndIf
        \EndFor
        \State $k \gets k + 1$
        \State $y_k \gets M(\phi, x)[\eta]$
        \State $\mathsf{append}(F, xy_k)$
    \EndWhile
    \State Return to Step \ref{step:start}.\label{step:return}
\end{algorithmic}
\end{breakablealgorithm}
\vspace{10pt}

The next algorithm we describe is a procedure to construct Vizing chains (see Algorithm~\ref{alg:rand_chain}).
Given a partial coloring $\phi$, an uncolored edge $e = xy$, a vertex $x \in e$, and a pair of colors $\alpha,\,\beta \in [q] \cup \set{\blank}$ such that either $\alpha = \beta = \blank$ or $\alpha \in M(\phi, x) \setminus M(\phi, y)$ and $\beta \in M(\phi, y)$, it returns a fan $F$, a path $P$, and a color $\eta$ such that
\begin{itemize}
    \item $C = F + P$ is a Vizing chain, and
    \item if $C$ is $\phi$-happy, then $\eta$ is a valid color to assign to $\End(P)$ in the coloring $\Shift(\phi, C)$.
\end{itemize}
The chain $C$ satisfies certain properties we describe and prove in the next subsection.
The colors $\alpha$, $\beta$ represent the colors on the path of the previous Vizing chain (which are $\blank$ if we are constructing the first chain), and the coloring $\phi$ refers to the shifted coloring with respect to the multi-step Vizing chain computed so far.

We start by applying the \hyperref[alg:rand_fan]{Random Fan Algorithm} (Algorithm \ref{alg:rand_fan}) as a subprocedure. 
Let $(\tilde F, \delta, j)$ be its output.
If $\delta \in M(\phi, x)$, we let $F = \tilde F$, $P = (\End(\tilde F))$, and $\eta = \delta$.
If $\delta = \beta$, we will show that the fan $\tilde F$ is $(\phi, \alpha\beta)$-hopeful.
In that case, we will return the fan $F = \tilde F$ and the path $P \defeq P(\End(F); \Shift(\phi, F), \alpha\beta)|2\ell$ (this notation is defined in \S\ref{subsec:pathchains}).
We then define $\eta$ as follows:
\begin{itemize}
    \item if $P = (\End(F))$, we let $\eta = \alpha$,
    \item otherwise, we let $\eta \in \set{\alpha, \beta}$ be a color distinct from $\phi(\End(P))$.
\end{itemize}
If $\delta \neq \beta$, we pick a color $\gamma \in M(\phi, x) \setminus \set{\alpha}$ uniformly at random by applying Algorithm~\ref{alg:rand_col}.
We will show that either $\tilde F$ or $\tilde F'\defeq \tilde F|j$ is $(\phi, \gamma\delta)$-successful.
Let 
\[\tilde P \,\defeq\, P(\End(F); \Shift(\phi, F), \gamma\delta)|2\ell \quad \text{and} \quad \tilde P' \,\defeq\, P(\End(F'); \Shift(\phi, F'), \gamma\delta)|2\ell.\] 
If $\vend(\tilde P) \neq x$, then we let $F = \tilde F$ and $P = \tilde P$; otherwise, we let $F = \tilde F'$ and $P = \tilde P'$.
Finally, we define $\eta$ as follows:
\begin{itemize}
    \item if $P = (\End(F))$, we let $\eta = \gamma$,
    \item otherwise, we let $\eta \in \set{\gamma, \delta}$ be a color distinct from $\phi(\End(P))$.
\end{itemize}

\vspace{10pt}
\begin{breakablealgorithm}
\caption{Random Vizing Chain}\algsize\label{alg:rand_chain}
\begin{flushleft}
\textbf{Input}: A proper partial edge-coloring $\phi$, an uncolored edge $e = xy$, a vertex $x \in e$, and a pair of colors $\alpha,\,\beta \in [q] \cup \set{\blank}$ such that either $\alpha = \beta = \blank$ or $\alpha \in M(\phi, x) \setminus M(\phi, y)$ and $\beta \in M(\phi, y)$. \\
\textbf{Output}: A fan $F$ with $\Start(F) = e$ and $\Pivot(F) = x$, a path $P$ with $\Start(P) = \End(F)$ and $\vstart(P) = \Pivot(F) = x$, and a color $\eta$.
\end{flushleft}
\begin{algorithmic}[1]
    \State $(\tilde F, \delta, j) \gets \hyperref[alg:rand_fan]{\mathsf{RandomFan}}(\phi, xy, x, \beta)$ \Comment{Algorithm \ref{alg:rand_fan}} \label{step:call_to_fan}
    \medskip
    \If{$M(\phi, x)[\delta] = \blank$}
        \State \Return $\tilde F$, $(\End(\tilde F))$, $\delta$ \label{step:happy_fan}
    \EndIf
    \medskip
    \If{$\delta  = \beta$}
        \State $P \gets P(\End(\tilde F);\, \hyperref[alg:chain_shift]{\Shift}(\phi, \tilde F),\, \alpha\beta)|2\ell$, \quad $\eta \gets \alpha$
        \If{$P \neq (\End(\tilde F))$ and $\phi(\End(P)) = \alpha$}
            \State $\eta \gets \beta$
        \EndIf
        \State \Return $\tilde F$, $P$, $\eta$ \label{step:beta_hopeful}
    \EndIf
    \medskip
    \State $\gamma \gets \hyperref[alg:rand_col]{\mathsf{RandomColor}}(\phi, x, \alpha)$ \Comment{Algorithm \ref{alg:rand_col}} \label{step:choose_gamma}
    \State $\tilde F' \gets \tilde F|j$
    \State $\tilde P \gets P(\End(\tilde F); \, \hyperref[alg:chain_shift]{\Shift}(\phi, \tilde F),\, \gamma\delta)|2\ell, \quad \tilde P' \gets P(\End(\tilde F'); \, \hyperref[alg:chain_shift]{\Shift}(\phi, \tilde F'), \, \gamma\delta)|2\ell$
    \If{$\vend(\tilde P) \neq x$} \label{step:path_condition}
        \State $F \gets \tilde F$, $P \gets \tilde P$
    \Else
        \State $F \gets \tilde F'$, $P \gets \tilde P'$
    \EndIf
    \State $\eta \gets \gamma$
    \If{$P \neq (\End(F))$ and $\phi(\End(P)) = \gamma$}
        \State $\eta \gets \delta$
    \EndIf
    \State \Return $F$, $P$, $\eta$
\end{algorithmic}
\end{breakablealgorithm}
\vspace{10pt}

Note that by our choice of data structures, each of the path chains in Algorithm~\ref{alg:rand_chain} can be computed in time $O(\ell)$.

Before we describe the \hyperref[alg:multi_viz_chain]{MSVA}, we define non-intersecting multi-step Vizing chains $C$ as in \cite{bernshteyn2023fast}.
Recall that $\IE(P)$ and $\IV(P)$ are the sets of internal edges and vertices of a path chain $P$, introduced in Definition~\ref{defn:internal}.

\begin{defn}[Non-intersecting chains]\label{defn:non-int}
    A $k$-step Vizing chain $C = F_0 + P_0 + \cdots + F_{k-1} + P_{k-1}$ is \emphd{non-intersecting} if for all $0\leq i < j < k$, $V(F_i) \cap V(F_j + P_j) = \0$ and $\IE(P_i) \cap E(F_j + P_j) = \0$.
\end{defn}

In our algorithm, we build a non-intersecting multi-step Vizing chain $C$.
As discussed in \S\ref{section: informal overview}, the algorithm is nearly identical to that of \cite{bernshteyn2023fast} with the main changes appearing in the subprocedures described earlier in this section.
Before we state the algorithm rigorously, we provide an informal overview. 
To start with, we have a chain $C = (xy)$ containing just the uncolored edge.
Using Algorithm~\ref{alg:rand_chain}, we find the first Vizing chain $F+P$.
Additionally, we let $\xi$ be the color in the output.
With these defined, we begin iterating.

At the start of each iteration, we have a non-intersecting chain $C = F_0 + P_0 + \cdots + F_{k-1} +P_{k-1}$, a \emphd{candidate chain} $F+P$, and a color $\xi$ satisfying the following properties for $\psi \defeq \Shift(\phi, C)$:
\begin{enumerate}[label=\ep{\normalfont{}\texttt{Inv}\arabic*},labelindent=15pt,leftmargin=*]
    \item\label{inv:start_F_end_C} $\Start(F) = \End(C)$ and $\vstart(F) = \vend(C)$,
    \item\label{inv:non_intersecting_shiftable} $C + F + P$ is non-intersecting and $\phi$-shiftable, 
    \item\label{inv:hopeful_length} $F$ is either $\psi$-happy or $(\psi, \alpha\beta)$-hopeful, where $P$ is an $\alpha\beta$-path; furthermore, if $F$ is $(\psi, \alpha\beta)$-disappointed, then $\length(P) = 2\ell$, and
    \item \label{inv:eta} if $F + P$ is $\psi$-happy, then $\xi$ is a valid color for $\End(P)$ in the coloring $\Shift(\psi, F+P)$.
\end{enumerate}
We will prove that these invariants hold in the next subsection. For now, let us take them to be true. Fig.~\ref{fig:iteration_start} shows an example of $C$ (in black) and $F+P$ (in red) at the start of an iteration.
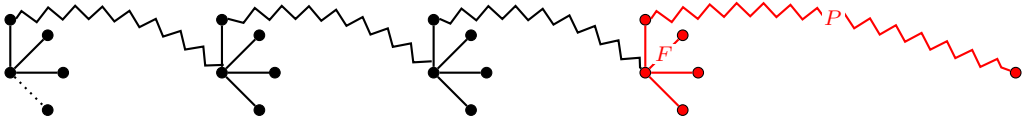
\begin{figure}[htb!]
    \centering
    \begin{tikzpicture}[xscale = 0.7,yscale=0.7]
        \node[circle,fill=black,draw,inner sep=0pt,minimum size=4pt] (a) at (0,0) {};
        	\path (a) ++(-45:1) node[circle,fill=black,draw,inner sep=0pt,minimum size=4pt] (b) {};
        	\path (a) ++(0:1) node[circle,fill=black,draw,inner sep=0pt,minimum size=4pt] (c) {};
        	\path (a) ++(45:1) node[circle,fill=black,draw,inner sep=0pt,minimum size=4pt] (d) {};
        	\path (a) ++(90:1) node[circle,fill=black,draw,inner sep=0pt,minimum size=4pt] (e) {};

            \path (c) ++(0:3) node[circle,fill=black,draw,inner sep=0pt,minimum size=4pt] (f) {};
            
        	\path (f) ++(-45:1) node[circle,fill=black,draw,inner sep=0pt,minimum size=4pt] (g) {};
        	\path (f) ++(0:1) node[circle,fill=black,draw,inner sep=0pt,minimum size=4pt] (h) {};
        	\path (f) ++(45:1) node[circle,fill=black,draw,inner sep=0pt,minimum size=4pt] (i) {};
        	\path (f) ++(90:1) node[circle,fill=black,draw,inner sep=0pt,minimum size=4pt] (j) {};
        	
        	\path (h) ++(0:3) node[circle,fill=black,draw,inner sep=0pt,minimum size=4pt] (k) {};

            \path (k) ++(-45:1) node[circle,fill=black,draw,inner sep=0pt,minimum size=4pt] (l) {};
        	\path (k) ++(0:1) node[circle,fill=black,draw,inner sep=0pt,minimum size=4pt] (m) {};
        	\path (k) ++(45:1) node[circle,fill=black,draw,inner sep=0pt,minimum size=4pt] (n) {};
        	\path (k) ++(90:1) node[circle,fill=black,draw,inner sep=0pt,minimum size=4pt] (o) {};
        	
        	\path (m) ++(0:3) node[circle,fill=red,draw,inner sep=0pt,minimum size=4pt] (p) {};
        	
        	\path (p) ++(-45:1) node[circle,fill=red,draw,inner sep=0pt,minimum size=4pt] (q) {};
        	\path (p) ++(0:1) node[circle,fill=red,draw,inner sep=0pt,minimum size=4pt] (r) {};
        	\path (p) ++(45:1) node[circle,fill=red,draw,inner sep=0pt,minimum size=4pt] (s) {};
        	\path (p) ++(90:1) node[circle,fill=red,draw,inner sep=0pt,minimum size=4pt] (t) {};
        	
        	\path (r) ++(0:6) node[circle,fill=red,draw,inner sep=0pt,minimum size=4pt] (u) {};
        	
        	\draw[thick,dotted] (a) -- (b);
        	
        	\draw[thick, decorate,decoration=zigzag] (e) to[out=10,in=135] (f) (j) to[out=10,in=135] (k) (o) to[out=10,in=135] (p);
            \draw[thick, decorate,decoration=zigzag, red](t) to[out=10,in=160] node[font=\fontsize{8}{8},midway,inner sep=1pt,outer sep=1pt,minimum size=4pt,fill=white] {$P$} (u);
        	
        	\draw[thick] (a) -- (c) (a) -- (d) (a) -- (e) (f) -- (g) (f) -- (h) (f) -- (i) (f) -- (j) (k) -- (l) (k) -- (m) (k) -- (n) (k) -- (o);
            \draw[thick, red] (p) -- (q) (p) -- (r) (p) to node[font=\fontsize{8}{8},midway,inner sep=1pt,outer sep=1pt,minimum size=4pt,fill=white] {$F$} (s) (p) -- (t);
    	
    \end{tikzpicture}
    \caption{The chain $C$ and candidate chain $F+P$ at the start of an iteration.}
    \label{fig:iteration_start}
\end{figure}
We first check whether $\length(P) < 2\ell$, in which case $F$ is $(\psi, \alpha\beta)$-successful.
If so, we have found a $\phi$-happy multi-step Vizing chain and we return $C+F+P$ and the color $\xi$.
If not, we let $F_k = F$ and consider a random initial segment $P_k$ of $P$ of length between $\ell$ and $2\ell - 1$. 
Using Algorithm \ref{alg:rand_chain}, we construct a chain $\tilde F + \tilde P$ with $\Start(\tilde F) = \End(P_k)$ and $\vstart(\tilde{F}) = \vend(P_k)$.
Additionally, we let $\eta$ be the color in the output.

At this point, we have two cases to consider. First, suppose $C + F_k + P_k + \tilde F + \tilde P$ is non-intersecting.
We then continue on, updating the chain to be $C+F_k+P_k$, the candidate chain to be $\tilde F + \tilde P$, and the color $\xi$ to be $\eta$.
We call such an update a \textsf{forward} iteration.
Fig.~\ref{fig:non_intersecting_iteration} shows an example of such an update with $\tilde F + \tilde P$ shown in blue.

\begin{figure}[htb!]
    \centering
        \begin{tikzpicture}[xscale = 0.7,yscale=0.7]
            \node[circle,fill=black,draw,inner sep=0pt,minimum size=4pt] (a) at (0,0) {};
        	\path (a) ++(-45:1) node[circle,fill=black,draw,inner sep=0pt,minimum size=4pt] (b) {};
        	\path (a) ++(0:1) node[circle,fill=black,draw,inner sep=0pt,minimum size=4pt] (c) {};
        	\path (a) ++(45:1) node[circle,fill=black,draw,inner sep=0pt,minimum size=4pt] (d) {};
        	\path (a) ++(90:1) node[circle,fill=black,draw,inner sep=0pt,minimum size=4pt] (e) {};

            \path (c) ++(0:3) node[circle,fill=black,draw,inner sep=0pt,minimum size=4pt] (f) {};
            
        	\path (f) ++(-45:1) node[circle,fill=black,draw,inner sep=0pt,minimum size=4pt] (g) {};
        	\path (f) ++(0:1) node[circle,fill=black,draw,inner sep=0pt,minimum size=4pt] (h) {};
        	\path (f) ++(45:1) node[circle,fill=black,draw,inner sep=0pt,minimum size=4pt] (i) {};
        	\path (f) ++(90:1) node[circle,fill=black,draw,inner sep=0pt,minimum size=4pt] (j) {};
        	
        	\path (h) ++(0:3) node[circle,fill=black,draw,inner sep=0pt,minimum size=4pt] (k) {};

            \path (k) ++(-45:1) node[circle,fill=black,draw,inner sep=0pt,minimum size=4pt] (l) {};
        	\path (k) ++(0:1) node[circle,fill=black,draw,inner sep=0pt,minimum size=4pt] (m) {};
        	\path (k) ++(45:1) node[circle,fill=black,draw,inner sep=0pt,minimum size=4pt] (n) {};
        	\path (k) ++(90:1) node[circle,fill=black,draw,inner sep=0pt,minimum size=4pt] (o) {};
        	
        	\path (m) ++(0:3) node[circle,fill=black,draw,inner sep=0pt,minimum size=4pt] (p) {};
        	
        	\path (p) ++(-45:1) node[circle,fill=black,draw,inner sep=0pt,minimum size=4pt] (q) {};
        	\path (p) ++(0:1) node[circle,fill=black,draw,inner sep=0pt,minimum size=4pt] (r) {};
        	\path (p) ++(45:1) node[circle,fill=black,draw,inner sep=0pt,minimum size=4pt] (s) {};
        	\path (p) ++(90:1) node[circle,fill=black,draw,inner sep=0pt,minimum size=4pt] (t) {};
        	
        	\path (r) ++(0:3) node[circle,fill=red,draw,inner sep=0pt,minimum size=4pt] (u) {};
        	
        	\path (u) ++(-45:1) node[circle,fill=red,draw,inner sep=0pt,minimum size=4pt] (v) {};
        	\path (u) ++(0:1) node[circle,fill=red,draw,inner sep=0pt,minimum size=4pt] (w) {};
        	\path (u) ++(45:1) node[circle,fill=red,draw,inner sep=0pt,minimum size=4pt] (x) {};
        	\path (u) ++(90:1) node[circle,fill=red,draw,inner sep=0pt,minimum size=4pt] (y) {};
        	
        	\path (w) ++(0:6) node[circle,fill=red,draw,inner sep=0pt,minimum size=4pt] (z) {};
        	
        	\draw[thick,dotted] (a) -- (b);
        	
        	\draw[thick, decorate,decoration=zigzag] (e) to[out=10,in=135] (f) (j) to[out=10,in=135] (k) (o) to[out=10,in=135] (p) (t) to[out=10,in=135] node[font=\fontsize{8}{8},midway,inner sep=1pt,outer sep=1pt,minimum size=4pt,fill=white] {$P_k$} (u);
            \draw[thick, decorate,decoration=zigzag, red](y) to[out=10,in=160] node[font=\fontsize{8}{8},midway,inner sep=1pt,outer sep=1pt,minimum size=4pt,fill=white] {$\tilde P$} (z);
        	
        	\draw[thick] (a) -- (c) (a) -- (d) (a) -- (e) (f) -- (g) (f) -- (h) (f) -- (i) (f) -- (j) (k) -- (l) (k) -- (m) (k) -- (n) (k) -- (o) (p) -- (q) (p) -- (r) (p) to node[font=\fontsize{8}{8},midway,inner sep=1pt,outer sep=1pt,minimum size=4pt,fill=white] {$F_k$} (s) (p) -- (t);
            \draw[thick, red] (u) -- (v) (u) -- (w) (u) to node[font=\fontsize{8}{8},midway,inner sep=1pt,outer sep=1pt,minimum size=4pt,fill=white] {$\tilde F$} (x) (u) -- (y);

        \begin{scope}[yshift=5.5cm]
            \node[circle,fill=black,draw,inner sep=0pt,minimum size=4pt] (a) at (0,0) {};
        	\path (a) ++(-45:1) node[circle,fill=black,draw,inner sep=0pt,minimum size=4pt] (b) {};
        	\path (a) ++(0:1) node[circle,fill=black,draw,inner sep=0pt,minimum size=4pt] (c) {};
        	\path (a) ++(45:1) node[circle,fill=black,draw,inner sep=0pt,minimum size=4pt] (d) {};
        	\path (a) ++(90:1) node[circle,fill=black,draw,inner sep=0pt,minimum size=4pt] (e) {};

            \path (c) ++(0:3) node[circle,fill=black,draw,inner sep=0pt,minimum size=4pt] (f) {};
            
        	\path (f) ++(-45:1) node[circle,fill=black,draw,inner sep=0pt,minimum size=4pt] (g) {};
        	\path (f) ++(0:1) node[circle,fill=black,draw,inner sep=0pt,minimum size=4pt] (h) {};
        	\path (f) ++(45:1) node[circle,fill=black,draw,inner sep=0pt,minimum size=4pt] (i) {};
        	\path (f) ++(90:1) node[circle,fill=black,draw,inner sep=0pt,minimum size=4pt] (j) {};
        	
        	\path (h) ++(0:3) node[circle,fill=black,draw,inner sep=0pt,minimum size=4pt] (k) {};

            \path (k) ++(-45:1) node[circle,fill=black,draw,inner sep=0pt,minimum size=4pt] (l) {};
        	\path (k) ++(0:1) node[circle,fill=black,draw,inner sep=0pt,minimum size=4pt] (m) {};
        	\path (k) ++(45:1) node[circle,fill=black,draw,inner sep=0pt,minimum size=4pt] (n) {};
        	\path (k) ++(90:1) node[circle,fill=black,draw,inner sep=0pt,minimum size=4pt] (o) {};
        	
        	\path (m) ++(0:3) node[circle,fill=red,draw,inner sep=0pt,minimum size=4pt] (p) {};
        	
        	\path (p) ++(-45:1) node[circle,fill=red,draw,inner sep=0pt,minimum size=4pt] (q) {};
        	\path (p) ++(0:1) node[circle,fill=red,draw,inner sep=0pt,minimum size=4pt] (r) {};
        	\path (p) ++(45:1) node[circle,fill=red,draw,inner sep=0pt,minimum size=4pt] (s) {};
        	\path (p) ++(90:1) node[circle,fill=red,draw,inner sep=0pt,minimum size=4pt] (t) {};
        	
        	\path (r) ++(0:3) node[circle,fill=blue,draw,inner sep=0pt,minimum size=4pt] (u) {};
        	
        	\path (u) ++(-45:1) node[circle,fill=blue,draw,inner sep=0pt,minimum size=4pt] (v) {};
        	\path (u) ++(0:1) node[circle,fill=blue,draw,inner sep=0pt,minimum size=4pt] (w) {};
        	\path (u) ++(45:1) node[circle,fill=blue,draw,inner sep=0pt,minimum size=4pt] (x) {};
        	\path (u) ++(90:1) node[circle,fill=blue,draw,inner sep=0pt,minimum size=4pt] (y) {};
        	
        	\path (w) ++(0:6) node[circle,fill=blue,draw,inner sep=0pt,minimum size=4pt] (z) {};
        	
        	\draw[thick,dotted] (a) -- (b);
        	
        	\draw[thick, decorate,decoration=zigzag] (e) to[out=10,in=135] (f) (j) to[out=10,in=135] (k) (o) to[out=10,in=135] (p);
            \draw[thick, decorate,decoration=zigzag, red](t) to[out=10,in=135] node[font=\fontsize{8}{8},midway,inner sep=1pt,outer sep=1pt,minimum size=4pt,fill=white] {$P_k$} (u);
            \draw[thick, decorate,decoration=zigzag, blue](y) to[out=10,in=160] node[font=\fontsize{8}{8},midway,inner sep=1pt,outer sep=1pt,minimum size=4pt,fill=white] {$\tilde P$} (z);
        	
        	\draw[thick] (a) -- (c) (a) -- (d) (a) -- (e) (f) -- (g) (f) -- (h) (f) -- (i) (f) -- (j) (k) -- (l) (k) -- (m) (k) -- (n) (k) -- (o);
            \draw[thick, red] (p) -- (q) (p) -- (r) (p) to node[font=\fontsize{8}{8},midway,inner sep=1pt,outer sep=1pt,minimum size=4pt,fill=white] {$F_k$} (s) (p) -- (t);
            \draw[thick, blue] (u) -- (v) (u) -- (w) (u) to node[font=\fontsize{8}{8},midway,inner sep=1pt,outer sep=1pt,minimum size=4pt,fill=white] {$\tilde F$} (x) (u) -- (y);
        \end{scope}

        \begin{scope}[yshift=3cm]
            \draw[-{Stealth[length=3mm,width=2mm]},very thick,decoration = {snake,pre length=3pt,post length=7pt,},decorate] (11.5,1) -- (11.5,-1);
        \end{scope}
        	
        \end{tikzpicture}
    \caption{Example of a \textsf{forward} iteration.}
    \label{fig:non_intersecting_iteration}
\end{figure}
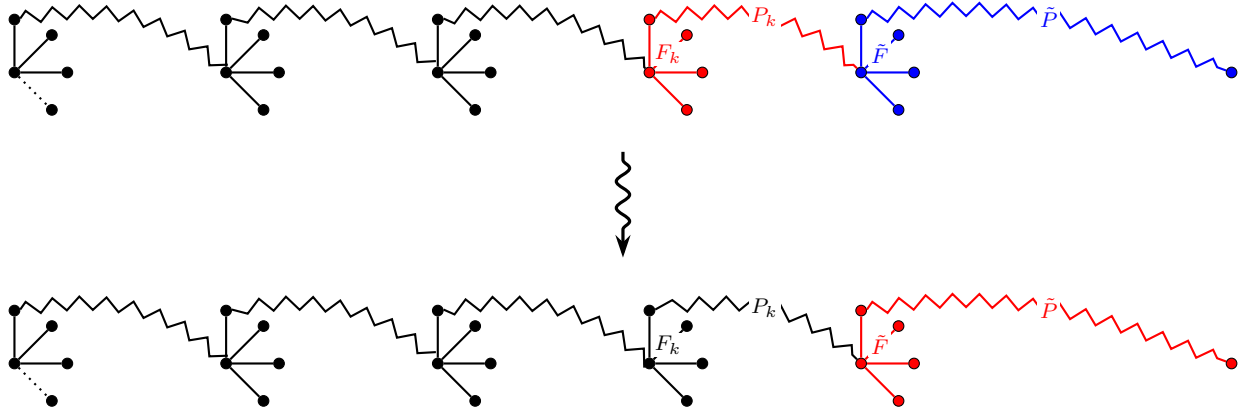

Now suppose $\tilde F + \tilde P$ intersects $C + F_k + P_k$. The edges and vertices of $\tilde{F} + \tilde{P}$ are naturally ordered, and we let $0 \leq j \leq k$ be the index such that the first intersection point of $\tilde{F} + \tilde{P}$ with $C + F_k + P_k$ occurred at $F_j+P_j$. Then we update $C$ to $C' \defeq F_0 + P_0 + \cdots + F_{j-1} + P_{j-1}$ and $F+P$ to $F_j + P'$, where $P'$ is the path of length $2\ell$ from which $P_j$ was obtained as an initial segment.
Finally, assuming $P'$ is a $\gamma\delta$-path, we update $\xi$ to be the unique color in $\set{\gamma, \delta} \setminus \set{\Shift(\phi, C')(\End(P'))}$.
For $r = k-j$, we call such an update an \textsf{$r$-backward} iteration.
Fig.~\ref{fig:intersecting_iteration} shows an example of such an update.

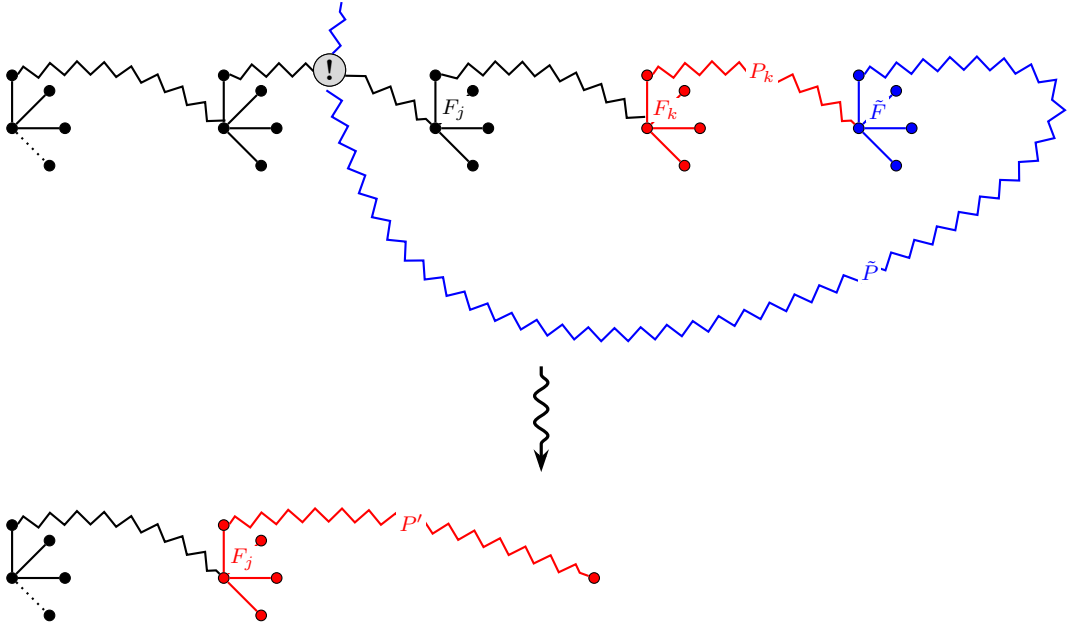
\begin{figure}[htb!]
    \centering
        \begin{tikzpicture}[xscale = 0.7, yscale=0.7]
            \clip (-0.5,-9.5) rectangle (20,2.5);
    
            \node[circle,fill=black,draw,inner sep=0pt,minimum size=4pt] (a) at (0,0) {};
        	\path (a) ++(-45:1) node[circle,fill=black,draw,inner sep=0pt,minimum size=4pt] (b) {};
        	\path (a) ++(0:1) node[circle,fill=black,draw,inner sep=0pt,minimum size=4pt] (c) {};
        	\path (a) ++(45:1) node[circle,fill=black,draw,inner sep=0pt,minimum size=4pt] (d) {};
        	\path (a) ++(90:1) node[circle,fill=black,draw,inner sep=0pt,minimum size=4pt] (e) {};

            \path (c) ++(0:3) node[circle,fill=black,draw,inner sep=0pt,minimum size=4pt] (f) {};
            
        	\path (f) ++(-45:1) node[circle,fill=black,draw,inner sep=0pt,minimum size=4pt] (g) {};
        	\path (f) ++(0:1) node[circle,fill=black,draw,inner sep=0pt,minimum size=4pt] (h) {};
        	\path (f) ++(45:1) node[circle,fill=black,draw,inner sep=0pt,minimum size=4pt] (i) {};
        	\path (f) ++(90:1) node[circle,fill=black,draw,inner sep=0pt,minimum size=4pt] (j) {};
        	
        	\path (h) ++(0:3) node[circle,fill=black,draw,inner sep=0pt,minimum size=4pt] (k) {};

            \path (k) ++(-45:1) node[circle,fill=black,draw,inner sep=0pt,minimum size=4pt] (l) {};
        	\path (k) ++(0:1) node[circle,fill=black,draw,inner sep=0pt,minimum size=4pt] (m) {};
        	\path (k) ++(45:1) node[circle,fill=black,draw,inner sep=0pt,minimum size=4pt] (n) {};
        	\path (k) ++(90:1) node[circle,fill=black,draw,inner sep=0pt,minimum size=4pt] (o) {};
        	
        	\path (m) ++(0:3) node[circle,fill=red,draw,inner sep=0pt,minimum size=4pt] (p) {};
        	
        	\path (p) ++(-45:1) node[circle,fill=red,draw,inner sep=0pt,minimum size=4pt] (q) {};
        	\path (p) ++(0:1) node[circle,fill=red,draw,inner sep=0pt,minimum size=4pt] (r) {};
        	\path (p) ++(45:1) node[circle,fill=red,draw,inner sep=0pt,minimum size=4pt] (s) {};
        	\path (p) ++(90:1) node[circle,fill=red,draw,inner sep=0pt,minimum size=4pt] (t) {};
        	
        	\path (r) ++(0:3) node[circle,fill=blue,draw,inner sep=0pt,minimum size=4pt] (u) {};
        	
        	\path (u) ++(-45:1) node[circle,fill=blue,draw,inner sep=0pt,minimum size=4pt] (v) {};sep=0pt,minimum size=4pt] (q) {};
        	\path (u) ++(0:1) node[circle,fill=blue,draw,inner sep=0pt,minimum size=4pt] (w) {};
        	\path (u) ++(45:1) node[circle,fill=blue,draw,inner sep=0pt,minimum size=4pt] (x) {};
        	\path (u) ++(90:1) node[circle,fill=blue,draw,inner sep=0pt,minimum size=4pt] (y) {};
        	
        	\node[circle,fill=gray!30,draw,inner sep=1.3pt] (z) at (6, 1.1) {\textbf{!}};
        	\path (z) ++(80:1.5) node (aa) {};

        	\draw[thick,dotted] (a) -- (b);
        	
        	\draw[thick, decorate,decoration=zigzag] (e) to[out=10,in=135] (f) (o) to[out=10,in=135] (p) (j) to[out=10,in=180] (z) to[out=-20,in=135] (k);
            \draw[thick,decorate,decoration=zigzag, red] (t) to[out=10,in=135] node[font=\fontsize{8}{8},midway,inner sep=1pt,outer sep=1pt,minimum size=3pt,fill=white] {$P_k$} (u); 
            
            \draw[thick,decorate,decoration=zigzag, blue] (y) to[out=10,in=-80, looseness=3] node[font=\fontsize{8}{8},midway,inner sep=1pt,outer sep=1pt,minimum size=3pt,fill=white] {$\tilde P$} (z) -- (aa);
        	
        	\draw[thick] (a) -- (c) (a) -- (d) (a) -- (e) (f) -- (g) (f) -- (h) (f) -- (i) (f) -- (j) (k) -- (l) (k) -- (m) (k) to node[font=\fontsize{8}{8},midway,inner sep=1pt,outer sep=1pt,minimum size=4pt,fill=white] {$F_j$} (n) (k) -- (o);
            \draw[thick, red] (p) -- (q) (p) -- (r) (p) to node[font=\fontsize{8}{8},midway,inner sep=1pt,outer sep=1pt,minimum size=4pt,fill=white] {$F_k$} (s) (p) -- (t);
            \draw[thick, blue] (u) -- (v) (u) -- (w) (u) to node[font=\fontsize{8}{8},midway,inner sep=1pt,outer sep=1pt,minimum size=4pt,fill=white] {$\tilde F$} (x) (u) -- (y);

            \node[circle,fill=gray!30,draw,inner sep=1.3pt] at (6, 1.1) {\textbf{!}};
        
        \begin{scope}[yshift=-5.5cm]
            \draw[-{Stealth[length=3mm,width=2mm]},very thick,decoration = {snake,pre length=3pt,post length=7pt,},decorate] (10,1) -- (10,-1);
        \end{scope}
        
        \begin{scope}[yshift=-8.5cm]
            \node[circle,fill=black,draw,inner sep=0pt,minimum size=4pt] (a) at (0,0) {};
        	\path (a) ++(-45:1) node[circle,fill=black,draw,inner sep=0pt,minimum size=4pt] (b) {};
        	\path (a) ++(0:1) node[circle,fill=black,draw,inner sep=0pt,minimum size=4pt] (c) {};
        	\path (a) ++(45:1) node[circle,fill=black,draw,inner sep=0pt,minimum size=4pt] (d) {};
        	\path (a) ++(90:1) node[circle,fill=black,draw,inner sep=0pt,minimum size=4pt] (e) {};

            \path (c) ++(0:3) node[circle,fill=red,draw,inner sep=0pt,minimum size=4pt] (f) {};
            
        	\path (f) ++(-45:1) node[circle,fill=red,draw,inner sep=0pt,minimum size=4pt] (g) {};
        	\path (f) ++(0:1) node[circle,fill=red,draw,inner sep=0pt,minimum size=4pt] (h) {};
        	\path (f) ++(45:1) node[circle,fill=red,draw,inner sep=0pt,minimum size=4pt] (i) {};
        	\path (f) ++(90:1) node[circle,fill=red,draw,inner sep=0pt,minimum size=4pt] (j) {};
        	
        	\path (h) ++(0:6) node[circle,fill=red,draw,inner sep=0pt,minimum size=4pt] (k) {};

        \draw[thick, dotted] (a) -- (b);
         \draw[thick] (a) -- (c) (a) -- (d) (a) -- (e);
         \draw[thick, red] (f) -- (g) (f) -- (h) (f) to node[font=\fontsize{8}{8},midway,inner sep=1pt,outer sep=1pt,minimum size=4pt,fill=white] {$F_j$} (i) (f) -- (j);

         \draw[thick, decorate,decoration=zigzag] (e) to[out=10,in=135] (f);
         \draw[thick, decorate,decoration=zigzag, red] (j) to[out=10,in=160] node[font=\fontsize{8}{8},midway,inner sep=1pt,outer sep=1pt,minimum size=4pt,fill=white] {$P'$} (k);
        \end{scope}
        \end{tikzpicture}
    \caption{Example of a \textsf{$2$-backward} iteration.}
    \label{fig:intersecting_iteration}
\end{figure}

In order to track intersections we define a hash map $\visited$ with key set $V\cup E$. We let $\visited(v) = 1$ if and only if $v$ is a vertex on a fan in the current chain $C$ and $\visited(e) = 1$ if and only if $e$ is an internal edge of a path in the current chain $C$.
The formal statement of the \hyperref[alg:multi_viz_chain]{Multi-Step Vizing Algorithm} is given in Algorithm~\ref{alg:multi_viz_chain}.

\vspace{10pt}
\begin{breakablealgorithm}
\caption{Multi-Step Vizing Algorithm (MSVA)}\algsize\label{alg:multi_viz_chain}
\begin{flushleft}
\textbf{Input}: A proper partial coloring $\phi$, an uncolored edge $xy$, and a vertex $x \in e$. \\
\textbf{Output}: A $\phi$-happy multi-step Vizing chain $C$ with $\Start(C) = xy$ and a color $\xi$.
\end{flushleft}
\begin{algorithmic}[1]
    \State $\visited(e) \gets 0, \quad \visited(v) \gets 0$ \quad \textbf{for each} $e \in E$, $v \in V$
    \State $(F,\,P,\,\xi) \gets \hyperref[alg:rand_chain]{\mathsf{RandomChain}}(\phi, xy, x, \blank, \blank)$ \label{step:first_chain} \Comment{Algorithm \ref{alg:rand_chain}}
    \State $C\gets (xy), \quad \psi \gets \phi, \quad k \gets 0$
    \medskip
    \While{true}\label{line:basic_loop_fr}
        \If{$\length(P) < 2\ell$}
            \State \Return $C+F+P$, $\xi$ \label{step:success} \Comment{Success}
        \EndIf
        \medskip
        \State\label{step:random_choice} Let $\ell' \in [\ell,2\ell-1]$ be an integer chosen uniformly at random.
        \State $F_k \gets F,\quad P_k\gets P|\ell'$ \label{step:Pk} \Comment{Randomly shorten the path}
        \State Let $\alpha$, $\beta$ be such that $P_k$ is an $\alpha\beta$-path where $\psi(\End(P_k)) = \beta$.
        \State $\psi \gets \hyperref[alg:chain_shift]{\Shift}(\psi, F_k+P_k)$ 
        \State $\visited(v) \gets 1$ \textbf{for each} $v \in V(F_k)$
        \State $\visited(e) \gets 1$ \textbf{for each} $e \in \IE(P_k)$
        \State $uv \gets \End(P_k), \quad v \gets \vend(P_k)$
        \State $(\tilde F,\, \tilde P,\, \eta) \gets \hyperref[alg:rand_chain]{\mathsf{RandomChain}}(\psi, uv, u, \alpha, \beta)$ \label{step:alpha_beta_order} \Comment{Algorithm \ref{alg:rand_chain}}
        \medskip
        \If{$\visited(v) = 1$ or $\visited(e) = 1$ for some $v\in V(\tilde F + \tilde P)$, $e \in E(\tilde F + \tilde P)$}
            \State Let $0 \leq j \leq k$ be such that the first intersection occurs at $F_j + P_j$.\label{step:choosej}
            \State $\psi \gets \hyperref[alg:chain_shift]{\Shift}(\psi, (F_j + P_j + \cdots + F_k + P_k)^*)$ \label{step:psi}
            \State $\visited(v) \gets 0$ \textbf{for each} $v \in V(F_j) \cup \cdots \cup V(F_k)$
            \State $\visited(e) \gets 0$ \textbf{for each} $e \in \IE(P_j) \cup \cdots \cup \IE(P_k)$\label{step:visited}
            \State $C\gets F_0 + P_0 + \cdots + F_{j-1} + P_{j-1}, \quad k \gets j$ \label{step:truncate_chain} \Comment{Return to step $j$}
            \State $F\gets F_j, \quad P \gets P'$ \quad where $P_j$ is an initial segment of $P'$ as described earlier.
            \State Let $\gamma$, $\delta$ be such that $P'$ is a $\gamma\delta$-path where $\psi(\End(P')) = \delta$.
            \State $\xi \gets \gamma$
        \medskip
        \ElsIf{$2 \leq \length(\tilde P) < 2\ell$ \textbf{and} $\vend(\tilde P) = \Pivot(\tilde F)$}
            \State\label{step:fail_chain} \Return \textbf{\textsf{FAIL}} \Comment{Failure}
        \medskip
        \Else
            \State $ C \gets C + F_k + P_k, 
        \quad F \gets \tilde F, \quad P \gets \tilde P, \quad \xi \gets \eta, \quad k \gets k + 1$ \label{step:append} \Comment{Append}
        \EndIf
    \EndWhile
\end{algorithmic}
\end{breakablealgorithm}
\vspace{10pt}

A few remarks are in order. The \textsf{while} loop starting on line \ref{line:basic_loop_fr} of Algorithm~\ref{alg:multi_viz_chain} is what we called the ``basic \textsf{while} loop'' in the introduction. 
Note that in steps \ref{step:psi}--\ref{step:visited} of the algorithm, we can update $\visited$ while simultaneously preforming the $\Shift$ operation as in Algorithm~\ref{alg:chain_shift}.
By construction, $\length(P_k) \geq \ell > 2$ for all $k$. This ensures that at least one edge of each color $\alpha$, $\beta$ is on the $\alpha\beta$-path $P_k$ and also guarantees that $V(\End(F_j)) \cap V(\Start(F_{j+1})) = \0$ for all $j$.
In Step \ref{step:truncate_chain} of the algorithm, we truncate the current chain at the \emph{first} vertex $v$ or edge $e$ on $\tilde F + \tilde P$ such that $\visited(v) = 1$ or $\visited(e) = 1$.
These observations will be important for the proofs in the sequel.

\subsection{Proof of Correctness}\label{subsec: poc}

In this subsection, we prove the correctness of the \hyperref[alg:multi_viz_chain]{MSVA} as well as some auxiliary results on the chain it outputs.
These results will be important for the analysis in later sections.
First, let us consider the output of the \hyperref[alg:rand_fan]{Random Fan Algorithm}.

\begin{Lemma}\label{lemma:fan_lemma}
    Let $\phi$ be a proper partial coloring, let $xy$ be an uncolored edge, and let $\beta \in [q]$ be such that $\beta \in M(\phi, y) \cup \set{\blank}$. 
    Let $(F, \delta, j)$ be the output of Algorithm \ref{alg:rand_fan} on input $(\phi, xy, x, \beta)$, let $F' \defeq F|j$, and let $\alpha \in M(\phi, x) \setminus M(\phi, y)$ be arbitrary.
    Then no edge in $F$ is colored $\alpha$ or $\beta$ and at least one of the following statements holds:
    \begin{itemize}
        \item $F$ is $\phi$-happy, or
        \item $\delta = \beta$ and the fan $F$ is $(\phi, \alpha\beta)$-hopeful, or
        \item either $F$ or $F'$ is $(\phi, \gamma\delta)$-successful for any $\gamma \in M(\phi, x) \setminus \{\alpha\}$.
    \end{itemize}
\end{Lemma}

\begin{proof}
    For $\beta = \blank$, the proof is identical to \cite[Lemma 4.8]{bernshteyn2023fast}, \textit{mutatis mutandis}.
    For $\beta \neq \blank$, the proof is identical to \cite[Lemma 4.9]{bernshteyn2023fast}, \textit{mutatis mutandis}.
\end{proof}

Next, let us consider the output of Algorithm~\ref{alg:rand_chain}.

\begin{Lemma}\label{lemma:chain_lemma}
    Let $\phi$ be a proper partial coloring, let $xy$ be an uncolored edge, and let $\alpha,\,\beta \in [q]$ be such that $\alpha \in \left(M(\phi, x) \setminus M(\phi, y)\right) \cup \set{\blank}$ and $\beta \in M(\phi, y) \cup \set{\blank}$. 
    Let $(\tilde F, \,\tilde P,\, \eta)$ be the output of Algorithm \ref{alg:rand_chain} on input $(\phi, xy, x, \alpha, \beta)$, where $\tilde P$ is a $\gamma\delta$-path. Then no edge in $\tilde F$ is colored $\alpha$, $\beta$, $\gamma$, or $\delta$, and
    \begin{itemize}
        \item either $\tilde F$ is $\phi$-happy and $\tilde P = (\End(\tilde F))$, or
        \item $\tilde F$ is $(\phi, \gamma\delta)$-hopeful and $\{\gamma, \delta\} = \{\alpha, \beta\}$, or
        \item $\tilde F$ is $(\phi, \gamma\delta)$-hopeful, $\{\gamma, \delta\} \cap \{\alpha, \beta\} = \0$, and $\length(\tilde P) = 2\ell$, or
        \item $\tilde F$ is $(\phi, \gamma\delta)$-successful and $\{\gamma, \delta\} \cap \{\alpha, \beta\} = \0$.
    \end{itemize}
    Furthermore, if $\tilde F$ is $\phi$-happy or $(\phi, \gamma\delta)$-successful, then $\eta \in M(\psi, u) \cap M(\psi, v)$, where $\psi \defeq \Shift(\phi, \tilde F + \tilde P)$ and $uv = \End(\tilde P)$.
\end{Lemma}

The proof of the bullet points is identical to \cite[Lemmas 5.2 and 5.3]{bernshteyn2023fast}, \textit{mutatis mutandis}.
Furthermore, the proof of the last claim regarding the color $\eta$ follows from standard arguments regarding Vizing chains and so we omit it here.

It follows that Algorithm \ref{alg:multi_viz_chain} outputs a $\phi$-happy multi-step Vizing chain as long as 
\begin{enumerate}[label=\ep{\normalfont{}\texttt{Happy}\arabic*},labelindent=15pt,leftmargin=*]
    \item\label{item:valid_input} the input to Algorithm \ref{alg:rand_chain} at Step \ref{step:alpha_beta_order} is valid,

    \item\label{item:never_fail} we never reach Step \ref{step:fail_chain}, and

    \item\label{item:invariants} the invariants \ref{inv:start_F_end_C}--\ref{inv:eta} hold for each iteration of the \textsf{while} loop.
\end{enumerate}
The proof that conditions \ref{item:valid_input} and \ref{item:never_fail} hold follows identically to \cite[Lemmas 5.4 and 5.5]{bernshteyn2023fast}.
Similarly, the proofs of items \ref{inv:start_F_end_C}--\ref{inv:hopeful_length} follow identically to \cite[Lemma 5.6]{bernshteyn2023fast}.
It remains to prove \ref{inv:eta}.

\begin{Lemma}\label{lemma:inv_eta}
    Consider running Algorithm \ref{alg:multi_viz_chain} on input $(\phi, xy, x, \ell)$.
    Let $C$, $F + P$, and $\xi$ be the multi-step chain, the candidate chain, and the color at the beginning of some iteration of the {\upshape\textsf{while}} loop, respectively.
    Then $\xi$ satisfies the invariant \ref{inv:eta}.
\end{Lemma}

\begin{proof}
    The only case to consider is when $F + P$ is $\psi$-happy for $\psi \defeq \Shift(\phi, C)$.
    It must be the case that $F+P$ was constructed during the previous iteration (or else we would have terminated earlier).
    The claim follows by Lemma~\ref{lemma:chain_lemma}.
\end{proof}

We conclude this subsection with two lemmas from \cite{bernshteyn2023fast}.
The first lemma describes an implication of Lemma~\ref{lemma:chain_lemma} on a certain kind of intersection.

\begin{Lemma}[{\cite[Lemma~5.8]{bernshteyn2023fast}}]\label{lemma:intersection_prev}
    Suppose we have a {\upshape\textsf{$0$-backward}} iteration of the {\upshape\textsf{while}} loop in Algorithm~\ref{alg:multi_viz_chain}.
    Then the first intersection must occur at a vertex in $V(F_k)$.
\end{Lemma}

The next lemma describes some properties of non-intersecting chains that will be useful in the proofs presented in the subsequent sections.

\begin{Lemma}[{\cite[Lemma~5.7]{bernshteyn2023fast}}]\label{lemma:non-intersecting_degrees}
    Let $\phi$ be a proper partial coloring and let $e = xy$ be an uncolored edge. 
    Consider running Algorithm~\ref{alg:multi_viz_chain} with input $(\phi, e, x)$.
    Let $C = F_0+P_0+\cdots+F_{k-1}+P_{k-1}$ be the multi-step Vizing chain at the beginning of an iteration of the {\upshape\textsf{while}} loop and let $F_k + P_k$ be the chain formed at Step \ref{step:Pk} such that, for each $j$, $P_j$ is an $\alpha_j\beta_j$-path in the coloring $\Shift(\phi, F_0+P_0+\cdots +F_{j-1}+P_{j-1})$. Then:
    \begin{enumerate}[label=\ep{\normalfont{}\texttt{Chain}\arabic*},labelindent=15pt,leftmargin=*]
        \item\label{item:degree_end} $\deg(\vend(F_j); \phi, \alpha_j\beta_j) = 1$  for each $0 \leq j \leq k$, and
        
        \item\label{item:related_phi} for each $0 \leq j \leq k$, all edges of $P_j$ except $\Start(P_j)$ are colored $\alpha_j$ or $\beta_j$ under $\phi$. 
    \end{enumerate}
\end{Lemma}

\section{Analysis of the Random Fan Algorithm}\label{subsec: fan analysis}

Consider running Algorithm~\ref{alg:rand_fan} on input $(\phi, xy, x, \beta)$.
We say the algorithm \emphd{succeeds} during a run of the \textsf{while} loop if it successfully returns a fan within $\kmax$ iterations.
The main result of this section is the following:

\begin{prop}\label{lemma:fan_number_of_tries}
    Consider running Algorithm~\ref{alg:rand_fan} with input $(\phi, xy, x, \beta)$.
    Let $T$ be the number of times we reach Step~\ref{step:return} before a successful run of the {\upshape{\textsf{while}}} loop.
    For $\kmax \geq \frac{8(1+\eps)}{\eps(2- \eps)}$, we have
    \[\P[T \geq t] \,\leq\, \exp\left(-\eps^2\,t\,\kmax/100\right).\]
\end{prop}

Consider a single run of the \textsf{while} loop.
In order to prove the above proposition, we define the following parameters for the $i$-th iteration:
\begin{align*}
    \eta_i &\defeq \text{the color chosen at Step~\ref{step:missing_color_choice}}, \\
    y_0 &\defeq y, \\
    y_i &\defeq \text{the unique vertex in $N_G(x)$ such that $\phi(xy_i) = \eta_{i}$}, \\
    B_1 &\defeq M(\phi, y_0) \setminus \left(\set{\beta} \cup M(\phi, x)\right),\\
    B_i &\defeq M(\phi, y_{i-1}) \setminus \left(\set{\beta} \cup M(\phi, x) \cup \left(\bigcup\limits_{j = 0}^{i-2}M(\phi, y_j)\right)\right), \quad i \geq 2, \\
    G_1 &= M(\phi, y_0) \setminus \left(\set{\beta} \cup B_1\right), \\
    G_i &\defeq M(\phi, y_{i-1}) \setminus B_i, \quad i \geq 2.
\end{align*}
We think of $G_i$ as the \emphd{good} colors and $B_i$ as the \emphd{bad} colors missing at $y_{i-1}$ during the $i$-th iteration.
In particular, $\eta_{i}$ is good if it is in $G_i$ as then the algorithm would succeed during the $i$-th iteration.
At the start of the $i$-th iteration, we say $y_{i-1}$ is \textsf{Happy} if $|G_i| \geq \gamma\Delta$ for $0 < \gamma < \eps$ to be chosen later.
To assist with our proofs, we define the following random variables for $i \geq 1$:
\begin{align*}
    Y_i &\defeq \bbone\set{y_{i-1} \text{ is \textsf{Happy}}}, \\
    X_i &\defeq \bbone\set{y_{i-1} \text{ is \textsf{Happy} and } \eta_{i} \in G_{i}}, \\
    Z_i &\defeq 1 - Y_i, \\
    W_i &\defeq X_i + \frac{\gamma Z_i}{1+\eps}.
\end{align*}
Suppose we run the \textsf{while} loop for $t$ iterations.
Let us consider $W_i$ for $1 \leq i \leq t$.

\begin{Lemma}\label{lemma:Wi}
    $\E[W_i\mid W_1, \ldots, W_{i - 1}] \geq \frac{\gamma}{1 + \eps}$.
\end{Lemma}

\begin{proof}
    Note that $W_i \in \set{0, 1, \gamma/(1+\eps)}$.
    Furthermore, we have the following:
    \[\P[X_i = 1\mid Y_i = 1,\, y_{i-1}] \geq \frac{\gamma\Delta}{(1+\eps)\Delta} = \frac{\gamma}{1 + \eps},\]
    by Lemma~\ref{lemma:unif_rand_color} and since $|M(\phi, y_{i-1})| \leq (1+\eps)\Delta$.
    It follows that
    \[\E[W_i\mid Y_i,\, y_{i-1}] \geq Y_i \frac{\gamma}{1 + \eps} + \frac{\gamma (1 - Y_i)}{1+\eps} = \frac{\gamma}{1 + \eps}.\]
    As the above is independent of $Y_i$ and $y_{i-1}$, and since $W_i$ is independent of $W_1, \ldots, W_{i-1}$ given $Y_i$ and $y_{i-1}$, the proof is complete.
\end{proof}

Next, we will prove an upper bound on $\sum_{i = 1}^{t}Z_i$.

\begin{Lemma}\label{lemma:sum_Zi}
    $\sum_{i = 1}^{t}Z_i \leq \dfrac{1+\eps}{\eps - \gamma}$.
\end{Lemma}

\begin{proof}
    Note the following for $i \geq 2$:
    \[B_1 = M(\phi, y_0) \setminus \left(\set{\beta} \cup M(\phi, x)\right), \qquad B_i = M(\phi, y_{i-1}) \setminus \left(\set{\beta} \cup M(\phi, x) \cup \left(\bigcup\limits_{j = 0}^{i-2}M(\phi, y_j)\right)\right).\]
    In particular, since $B_j \subseteq M(\phi, y_{j-1})$, it follows that the sets $B_i$ are pairwise disjoint for $1 \leq i \leq t$. 
    Therefore, we may conclude
    \begin{align}\label{eq:ub_Bi}
        \sum_{i = 1}^{t}|B_i| \leq (1+\eps)\Delta.
    \end{align}
    Note that if $Z_i = 1$, by the definition of \textsf{Happy} we have
    \[|B_i| > |M(\phi, y_{i-1})| - \gamma \Delta \geq (\eps - \gamma)\Delta.\]
    Since $Z_i \in \set{0, 1}$ we may conclude
    \begin{align}\label{eq:lb_Bi}
        \sum_{i = 1}^{t}|B_i| \,\geq\, \sum_{i = 1}^{t}Z_i|B_i| \,\geq\, (\eps - \gamma)\Delta \sum_{i = 1}^{t}Z_i.
    \end{align}
    Putting together \eqref{eq:ub_Bi} and \eqref{eq:lb_Bi} completes the proof.
\end{proof}

Note that if $X_i = 1$, then the algorithm succeeds during the $i$-th iteration.
Therefore, it is enough to show that $\sum_{i = 1}^{\kmax}X_i > 0$ with high probability.
In light of Lemma~\ref{lemma:sum_Zi}, we have
\[\sum_{i = 1}^{\kmax}W_i > \frac{\gamma}{\eps - \gamma} \implies \sum_{i = 1}^{\kmax}X_i > 0.\]
It is now sufficient to consider the following:
\[\P\left[\sum_{i = 1}^{\kmax}W_i \leq \frac{\gamma}{\eps - \gamma}\right].\]
We shall employ the following concentration inequality due to Kuszmaul and Qi \cite{azuma}, which is a special case of their version of multiplicative Azuma's inequality for supermartingales:

\begin{theo}[{Kuszmaul--Qi \cite[Corollary 6]{azuma}}]\label{theo:azuma_supermartingale}
    Let $c > 0$ and let $\tilde X_1$, \ldots, $\tilde X_n$ be 
    random variables taking values in $[0,c]$.
    Suppose that $\E[\tilde X_i\mid \tilde X_1, \ldots, \tilde X_{i-1}] \leq a_i$ for all $i$.
    Let $\mu \defeq \sum_{i = 1}^na_i$. Then, for any $\delta > 0$,
    \[\P\left[\sum_{i = 1}^n\tilde X_i \geq (1+\delta)\mu\right] \,\leq\, \exp\left(-\frac{\delta^2\mu}{(2+\delta)c}\right).\]
\end{theo}

We are now ready to prove Proposition~\ref{lemma:fan_number_of_tries}.

\begin{proof}[Proof of Proposition~\ref{lemma:fan_number_of_tries}]
    Consider the random variables $\tilde W_i \defeq 1 - W_i$ for $1 \leq i \leq \kmax$.
    As a result of Lemma~\ref{lemma:Wi}, we have
    \[\E[\tilde W_i\mid \tilde W_1, \ldots, \tilde W_{i-1}] \leq 1 - \frac{\gamma}{1 + \eps}.\]
    We may apply Theorem~\ref{theo:azuma_supermartingale} with 
    \[\tilde X_i = \tilde W_i, \quad c = 1, \quad \mu = \kmax\left(1 - \frac{\gamma}{1 + \eps}\right),\]
    to get
    \[\P\left[\sum_{i = 1}^{\kmax}\tilde W_i \geq (1+\delta)\mu\right] \,\leq\, \exp\left(-\frac{\delta^2\kmax(1 + \eps - \gamma)}{(2+\delta)(1+\eps)}\right).\]
    Note the following:
    \[\sum_{i = 1}^{\kmax}\tilde W_i \geq (1+\delta)\mu \iff \sum_{i = 1}^{\kmax}W_i \,\leq\, \kmax - (1+\delta)\mu \,=\, \kmax\left(\frac{\gamma(1 + \delta)}{1 + \eps} - \delta\right).\]
    Let $\gamma = \eps/2$ and $\delta = \eps/4$.
    We have
    \[\frac{\gamma(1 + \delta)}{1 + \eps} - \delta = \frac{\eps(2-\eps)}{8(1 + \eps)} \quad \text{and} \quad \frac{\gamma}{\eps - \gamma} = 1.\]
    We may now conclude the following for $\kmax \geq \frac{8(1+\eps)}{\eps(2- \eps)}$:
    \begin{align*}
        \P\left[\text{we reach Step~\ref{step:return}}\right] &\leq \P\left[\sum_{i = 1}^{\kmax}X_i = 0\right] \\
        &\leq \P\left[\sum_{i = 1}^{\kmax}W_i \leq 1\right] \\
        &\leq \P\left[\sum_{i = 1}^{\kmax}\tilde W_i \geq (1+\delta)\mu\right] \\
        &\leq \exp\left(-\frac{\eps^2\kmax(1 + \eps/2)}{16(2+\eps/4)(1+\eps)}\right) \\
        &\leq \exp\left(-\eps^2\kmax/100\right).
    \end{align*}
    The claim now follows as each run of the \textsf{while} loop is independent.
\end{proof}

\section{Analysis of the MSVA}\label{subsec: rmsva analysis}

In this section, we will bound the probability the \textsf{while} loop in Algorithm~\ref{alg:multi_viz_chain} terminates within $t$ iterations. 
We will assume the partial coloring $\phi$ is fixed for this section.
Additionally, we fix an ordering of the vertices of $G$.
Our main result is the following:

\begin{prop}\label{theo:num_iter_while}
    Let $\phi$ be a proper partial $(1+\eps)\Delta$-edge-coloring and let $e = xy$ be an uncolored edge. For any $t > 0$, $\kmax \geq 16/\eps$, and $\ell \geq 6400\kmax^4$, the {\upshape{\textsf{while}}} loop in Algorithm \ref{alg:multi_viz_chain} with input $(\phi, xy, x)$ terminates within $t$ iterations with probability at least $1 - \frac{8m}{\eps}\left(\frac{6400\kmax^4}{\ell}\right)^{t/2}$.
\end{prop}

Before we prove the above result, we show how it implies Theorem~\ref{theo:augment_small}, i.e., the existence of augmenting subgraphs with $O((\log n)/\epsilon^4)$ edges.

\begin{proof}[Proof of Theorem~\ref{theo:augment_small}]
    Let $\phi$ be a proper partial $q$-edge-coloring and let $e = xy$ be an uncolored edge.
    Set $\kmax = 20/\eps$ and $\ell = 10000/\kmax^4$ and run Algorithm~\ref{alg:multi_viz_chain} on the input $(\phi, e, x)$.
    With probability at least $1 - 1/\poly(n)$, the algorithm terminates within $t = \Theta(\log n)$ iterations.
    Consider such an outcome and let $C = F_0 + P_0 + \cdots + F_{k-1} + P_{k-1}$ be the chain returned.
    By construction, we have $\length(C) = O((\ell + \kmax)\,k)$ and $k \leq t$.
    It follows that for any uncolored edge $e$, there exists a $\phi$-happy multi-step Vizing chain $C$ such that $\length(C) = O((\log n)/\eps^4)$.
    Letting $H$ be the subgraph induced by $E(C)$ completes the proof.
\end{proof}

Consider running Algorithm~\ref{alg:multi_viz_chain} with input $(\phi, f, z)$ where $f \in E$ and $z \in f$.
We will employ a version of the entropy compression argument from \cite{bernshteyn2023fast, dhawan2024edge}.
To this end, we first define an \emphd{input sequence} in order to encode the execution of Algorithm~\ref{alg:multi_viz_chain}.
Let $I = (f, z, \col_0, \ell_1, \col_1, \ldots, \ell_t, \col_t)$ be an input sequence of length $t$.
Here, $\col_i$ contains the random colors chosen during the $i$-th iteration of the \textsf{while} loop (for $i = 0$, this corresponds to Step~\ref{step:first_chain}), and $\ell_i$ is the random choice made at Step~\ref{step:Pk}.
We remark that we only consider the random colors chosen during the successful execution of the \textsf{while} loop in Algorithm~\ref{alg:rand_fan}.
In particular, $\col_i$ contains at most $\kmax + 1$ entries.
Clearly, the execution of the first $t$ iterations of the \textsf{while} loop is uniquely determined by the input sequence $I$.
We let $\mathcal{I}^{(t)}$ be the set of all input sequences for which Algorithm~\ref{alg:multi_viz_chain} does not terminate within the first $t$ iterations.

\begin{defn}[Records and termini]\label{def:record}
    Let $I = (f, z, \col_0, \ell_1, \col_1, \ldots, \ell_t, \col_t) \in \mathcal{I}^{(t)}$.
    Consider running Algorithm~\ref{alg:multi_viz_chain} for the first $t$ iterations with this input sequence.
    We define the \emphd{record} of $I$ to be the tuple $(D(I), S(I))$, where $D(I) = (d_1, \ldots, d_t) \in \Z^t$ and $S(I) = (s_1, \ldots, s_t) \in \set{0, 1}^t$.
    For each $i$, we define the relevant parameters as follows:
    \[
        d_i \,\defeq\, \begin{cases}
            1 &\text{if we reach Step \ref{step:append}};\\
            j-k &\text{if we reach Step \ref{step:truncate_chain}}.
        \end{cases}, \qquad \text{and} \qquad s_i \,\defeq\, \begin{cases}
            1 &\text{if we reach Step \ref{step:beta_hopeful}};\\
            0 &\text{otherwise}.
        \end{cases}
    \]
    Note that $d_i \leq 0$ if we reach Step~\ref{step:truncate_chain}.
    Let $C$ be the multi-step Vizing chain at the end of the $t$-th iteration.
    The \emphd{terminus} of $I$ is the pair $\tau(I) = (\End(C), \vend(C))$ (in the case that $C = (f)$, we define $\vend(C)$ such that $\vend(C) \neq z$).
\end{defn}

Let $\mathcal{D}^{(t)}$ be the set of all tuples $D \in \Z^t$ such that $D = D(I)$ for some $I \in \mathcal{I}^{(t)}$, and define $\mathcal{S}^{(t)}$ similarly.
Given $D \in \mathcal{D}^{(t)}$, $S \in \mathcal{S}^{(t)}$, a pair $(uv, u)$ such that $uv \in E$, $\y \in [\Delta]$, and a sequence of colors $\boldcol = (\col_0, \ldots, \col_t)$, we let $\mathcal{I}^{(t)}(D, S, \y, uv, u, \boldcol)$ be the set of all input sequences $I = (f, z, \col_0, \ell_1, \col_1, \ldots, \ell_t, \col_t) \in \mathcal{I}^{(t)}$ satisfying the following:
\begin{itemize}
    \item the record $(D(I), S(I))$ is $(D, S)$, 
    \item the terminus $\tau(I)$ is $(uv, u)$, and
    \item for $f = wz$, the vertex $w$ is the $\y$-th neighbor of $z$ (here we use that the vertices are ordered).
\end{itemize}
The following functions will assist with our proof:
\begin{align*}
    \val(z) \,&\defeq\, \left\{\begin{array}{cc}
        2\kmax & \text{if } z = 1 \\
        28\kmax^3 & \text{if } z = 0 \\
        50\ell\kmax^3 & \text{if } z < 0
    \end{array}\right., &z \in \Z,\ z\leq 1, \\
    \wt(D) \,&\defeq\, \prod_{i = 1}^t\val(d_i), &D = (d_1, \ldots, d_t) \in \mathcal{D}^{(t)}.
\end{align*}

We begin with the following lemma.

\begin{Lemma}\label{lemma:chain_formation}
    Consider running Algorithm~\ref{alg:multi_viz_chain} on input sequence $I =(f, z, \col_0, \ell_1, \col_1, \ldots, \ell_t, \col_t) \in \mathcal{I}^{(t)}$.
    Let $C = F_0 + P_0 + \cdots + F_{k-1} + P_{k-1}$ and $F_k + P_k$ be the chain and candidate chain at the end of the $t$-th iteration, respectively.
    For each $0 \leq j \leq k$, given $D(I)$, we may determine $ 0 \leq t_j \leq t$ such that $F_j + Pj$ is an initial segment of the chain constructed during the $t_j$-th iteration.
\end{Lemma}

\begin{proof}
    Let $D(I) = (d_1, \ldots, d_t)$.
    Note that by definition of $D(I)$, the chain at the end of the $t'$-th iteration is a $k'$-step Vizing chain, where $k' = \sum_{i = 1}^{t'}d_i$.
    
    Clearly, $t_0 = 0$.
    For $1\leq j \leq k$, we define $t_j$ as follows:
    \[t_j \defeq \max\left\{l\,:\,1 \leq l \leq t,\, d_l = 1,\, \sum_{r = 1}^{l} d_r = j\right\}.\]
    In particular, the $t_j$-th iteration is the latest \textsf{forward} iteration such that the chain at the end of the iteration is a $j$-step Vizing chain.
    We claim that $F_j + P_j$ is an initial segment of the candidate chain at the end of this iteration.
    Let $\tilde F_i + \tilde P_i$ be the chain constructed during the $i$-th iteration for each $i \in [t]$.
    Clearly, there is some \textsf{forward} iteration $l$ where $F_j = \tilde F_l$ and $P_j \subseteq \tilde P_l$.
    Furthermore, at the end of this iteration, the current chain will be a $j$-step Vizing chain.
    If $t_j$ is not this iteration, there is some \textsf{backward} iteration $l' > t_j$ such that the chain after the $l'$-th iteration has length less than $j$. 
    This contradicts the definition of $t_j$ and completes the proof.
\end{proof}

Now we apply the above result to determine the colors on the paths constructed at each iteration.

\begin{Lemma}\label{lemma:all_colors}
    Consider running Algorithm~\ref{alg:multi_viz_chain} on input sequence $I =(f, z, \col_0, \ell_1, \col_1, \ldots, \ell_t, \col_t) \in \mathcal{I}^{(t)}$.
    Let $\tilde P_i$ be the path formed during the $i$-th call to Algorithm~\ref{alg:rand_chain}.
    The colors on the paths $\tilde P_i$ can be determined by the entries $D(I),\,S(I)$, and $\col_0, \ldots, \col_t$.
\end{Lemma}

\begin{proof}
    We will prove this by induction on $i$.
    For $i = 0$, the colors on the path $\tilde P_0$ are the last two entries in $\col_0$.
    Suppose we have determined the colors on the paths $\tilde P_j$ for all $0 \leq j < i$.
    Consider the $i$-th iteration of the \textsf{while} loop.
    Let $C,\,F+P$ be the chain and the candidate chain at the start of this iteration and let $\tilde F + \tilde P$ be the chain returned at Step~\ref{step:alpha_beta_order}.
    If $\same_i = 0$, then the colors on the path $\tilde P$ are once again the last two entries in $\col_i$.
    Suppose $\same_i = 1$.
    Then, $\tilde P$ and $P$ have the same colors.
    Note that $P = \tilde P_j$ for some $j < i$.
    It is therefore enough to determine $j$ as we may then apply the induction hypothesis.
    The proof now follows by Lemma~\ref{lemma:chain_formation}.
\end{proof}

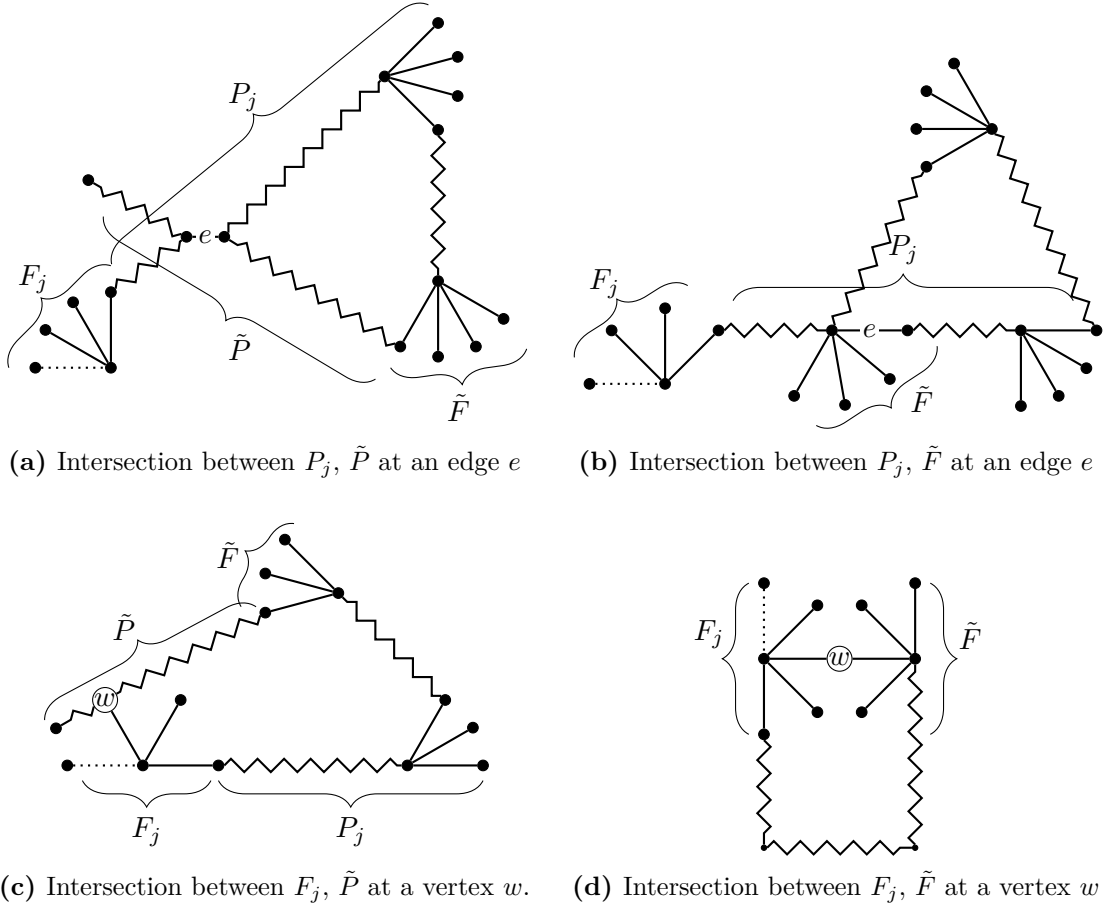
\begin{figure}[htb!]
    \begin{subfigure}[t]{0.45\textwidth}
        \centering
    	\begin{tikzpicture}
    	    \node[circle,fill=black,draw,inner sep=0pt,minimum size=4pt] (a) at (0,0) {};
    		\node[circle,fill=black,draw,inner sep=0pt,minimum size=4pt] (b) at (1,0) {};
    		\path (b) ++(150:1) node[circle,fill=black,draw,inner sep=0pt,minimum size=4pt] (c) {};
    		\path (b) ++(120:1) node[circle,fill=black,draw,inner sep=0pt,minimum size=4pt] (d) {};
    		\path (b) ++(90:1) node[circle,fill=black,draw,inner sep=0pt,minimum size=4pt] (e) {};
    		
    		\draw[decoration={brace,amplitude=10pt},decorate] (-0.35,0) -- node [midway,above,yshift=5pt,xshift=-10pt] {$F_j$} (1,1.35);
    		
    		\path (b) ++(60:2) node[circle,fill=black,draw,inner sep=0pt,minimum size=4pt] (f) {};
    		\path (f) ++(0:0.5) node[circle,fill=black,draw,inner sep=0pt,minimum size=4pt] (g) {};
    		
    		\draw[decoration={brace,amplitude=10pt},decorate] (1.01,1.37) -- node [midway,above,yshift=5pt,xshift=-10pt] {$P_j$} (5.2,4.8);
    		
    		\path (g) ++(45:3) node[circle,fill=black,draw,inner sep=0pt,minimum size=4pt] (h) {};
    		\path (h) ++(15:1) node[circle,fill=black,draw,inner sep=0pt,minimum size=4pt] (i) {};
    		\path (h) ++(45:1) node[circle,fill=black,draw,inner sep=0pt,minimum size=4pt] (j) {};
    		\path (h) ++(-15:1) node[circle,fill=black,draw,inner sep=0pt,minimum size=4pt] (k) {};
    		\path (h) ++(-45:1) node[circle,fill=black,draw,inner sep=0pt,minimum size=4pt] (l) {};
    		
    		\draw[decoration={brace,amplitude=10pt,mirror},decorate] (4.7,0) -- node [midway,below,yshift=-10pt,xshift=0pt] {$\tilde F$} (6.5,0.25);
    		
    		\path (l) ++(270:2) node[circle,fill=black,draw,inner sep=0pt,minimum size=4pt] (m) {};
    		\path (m) ++(270:1) node[circle,fill=black,draw,inner sep=0pt,minimum size=4pt] (n) {};
    		\path (m) ++(-60:1) node[circle,fill=black,draw,inner sep=0pt,minimum size=4pt] (o) {};
    		\path (m) ++(-30:1) node[circle,fill=black,draw,inner sep=0pt,minimum size=4pt] (p) {};
    		\path (m) ++(-120:1) node[circle,fill=black,draw,inner sep=0pt,minimum size=4pt] (q) {};
    		
    		\draw[decoration={brace,amplitude=10pt,mirror}, decorate] (0.9, 2) -- node [midway,below,xshift=0pt,yshift=-10pt] {$\tilde P$} (4.5,-0.15);
    		
    		\path (f) ++(150:1.5) node[circle,fill=black,draw,inner sep=0pt,minimum size=4pt] (r) {};
    		
    		\draw[thick,dotted] (a) -- (b);
    		\draw[thick] (f) to node[midway,inner sep=1pt,outer sep=1pt,minimum size=4pt,fill=white] {$e$} (g);
    		\draw[thick, snake=zigzag] (e) -- (f) -- (r) (g) -- (h) (l) -- (m) (q) -- (g);
    		\draw[ thick] (b) -- (c) (b) -- (d) (b) -- (e) (h) --  (i) (h) -- (j) (h) - -(k) (h) -- (l) (m) -- (n) (m) -- (o) (m) -- (p) (m) -- (q);
    		
    	\end{tikzpicture}
    	\caption{Intersection between $P_j,\, \tilde P$ at an edge $e$}\label{fig:Viz_path_path_intersect}
    \end{subfigure}
    \begin{subfigure}[t]{0.45\textwidth}
        \centering
    	\begin{tikzpicture}
    	    \node[circle,fill=black,draw,inner sep=0pt,minimum size=4pt] (a) at (0,0) {};
    		\node[circle,fill=black,draw,inner sep=0pt,minimum size=4pt] (b) at (1,0) {};
    		\path (b) ++(135:1) node[circle,fill=black,draw,inner sep=0pt,minimum size=4pt] (c) {};
    		\path (b) ++(90:1) node[circle,fill=black,draw,inner sep=0pt,minimum size=4pt] (d) {};
    		\path (b) ++(45:1) node[circle,fill=black,draw,inner sep=0pt,minimum size=4pt] (e) {};
    		
    		\draw[decoration={brace,amplitude=10pt},decorate] (-0.2,0.35) -- node [midway,above,yshift=5pt,xshift=-10pt] {$F_j$} (1.3, 1.3);
    		
    		\path (e) ++(0:1.5) node[circle,fill=black,draw,inner sep=0pt,minimum size=4pt] (f) {};
    		\path (f) ++(0:1) node[circle,fill=black,draw,inner sep=0pt,minimum size=4pt] (g) {};
    		\path (g) ++(0:1.5) node[circle,fill=black,draw,inner sep=0pt,minimum size=4pt] (h) {};
    		
    		\draw[decoration={brace,amplitude=10pt}, decorate] (1.9, 1.1) -- node [midway,above,xshift=0pt,yshift=10pt] {$P_j$} (6.4,1.1);
    		
    		\path (h) ++(-30:1) node[circle,fill=black,draw,inner sep=0pt,minimum size=4pt] (i) {};
    		\path (h) ++(-60:1) node[circle,fill=black,draw,inner sep=0pt,minimum size=4pt] (j) {};
    		\path (h) ++(-90:1) node[circle,fill=black,draw,inner sep=0pt,minimum size=4pt] (k) {};
    		\path (h) ++(0:1) node[circle,fill=black,draw,inner sep=0pt,minimum size=4pt] (l) {};

    		\path (f) ++(-120:1) node[circle,fill=black,draw,inner sep=0pt,minimum size=4pt] (m) {};
    		\path (f) ++(-80:1) node[circle,fill=black,draw,inner sep=0pt,minimum size=4pt] (n) {};
    		\path (f) ++(-40:1) node[circle,fill=black,draw,inner sep=0pt,minimum size=4pt] (o) {};

    		\draw[decoration={brace,amplitude=10pt,mirror},decorate] (3,-0.55) -- node [midway,yshift=-5pt,xshift=17pt] {$\tilde F$} (4.6, 0.5);
    		
    		\path (f) ++(60:2.5) node[circle,fill=black,draw,inner sep=0pt,minimum size=4pt] (p) {};
    		\path (p) ++(30:1) node[circle,fill=black,draw,inner sep=0pt,minimum size=4pt] (q) {};
    		\path (q) ++(180:1) node[circle,fill=black,draw,inner sep=0pt,minimum size=4pt] (r) {};
    		\path (q) ++(150:1) node[circle,fill=black,draw,inner sep=0pt,minimum size=4pt] (s) {};
    		\path (q) ++(120:1) node[circle,fill=black,draw,inner sep=0pt,minimum size=4pt] (t) {};

    		\draw[thick,dotted] (a) -- (b);
    		\draw[thick] (f) to node[midway,inner sep=1pt,outer sep=1pt,minimum size=4pt,fill=white] {$e$} (g);
    		\draw[thick, snake=zigzag] (e) -- (f) (g) -- (h) (f) -- (p) (q) -- (l);
    		\draw[ thick] (b) -- (c) (b) -- (d) (b) -- (e) (h) -- (i) (h) -- (j) (h) -- (k) (h) -- (l) (f) -- (m) (f) -- (n) (f) -- (o) (p) -- (q) -- (r) (q) -- (s) (q) -- (t);

    	\end{tikzpicture}
    	\caption{Intersection between $P_j,\,\tilde F$ at an edge $e$}\label{fig:Viz_path_fan_intersect}
    \end{subfigure}

    \vspace{15pt}
    
    \begin{subfigure}[b]{0.45\textwidth}
        \centering
    	\begin{tikzpicture}
    	    \node[circle,fill=black,draw,inner sep=0pt,minimum size=4pt] (a) at (0,0) {};
    		\node[circle,fill=black,draw,inner sep=0pt,minimum size=4pt] (b) at (1,0) {};
    		\path (b) ++(120:1) node[circle,draw,inner sep=0pt,minimum size=4pt] (c) {$w$};
    		\path (b) ++(60:1) node[circle,fill=black,draw,inner sep=0pt,minimum size=4pt] (d) {};
    		\path (b) ++(0:1) node[circle,fill=black,draw,inner sep=0pt,minimum size=4pt] (e) {};
    		
    		\draw[decoration={brace,amplitude=10pt,mirror}, decorate] (0.2, -0.2) -- node [midway,below,xshift=0pt,yshift=-10pt] {$F_j$} (1.9,-0.2);
    		
    		\path (e) ++(0:2.5) node[circle,fill=black,draw,inner sep=0pt,minimum size=4pt] (f) {};
    		\path (f) ++(0:1) node[circle,fill=black,draw,inner sep=0pt,minimum size=4pt] (g) {};
    		\path (f) ++(30:1) node[circle,fill=black,draw,inner sep=0pt,minimum size=4pt] (h) {};
    		\path (f) ++(60:1) node[circle,fill=black,draw,inner sep=0pt,minimum size=4pt] (i) {};
    		
    		\draw[decoration={brace,amplitude=10pt,mirror}, decorate] (2, -0.2) -- node [midway,below,xshift=0pt,yshift=-10pt] {$P_j$} (5.5,-0.2);
    		
    		\path (i) ++(135:2) node[circle,fill=black,draw,inner sep=0pt,minimum size=4pt] (j) {};
    		\path (j) ++(135:1) node[circle,fill=black,draw,inner sep=0pt,minimum size=4pt] (k) {};
    		\path (j) ++(165:1) node[circle,fill=black,draw,inner sep=0pt,minimum size=4pt] (l) {};
    		\path (j) ++(195:1) node[circle,fill=black,draw,inner sep=0pt,minimum size=4pt] (m) {};
    		
    		\draw[decoration={brace,amplitude=10pt}, decorate] (2.3, 2.1) -- node [midway,above,xshift=-15pt,yshift=-5pt] {$\tilde F$} (3,3.2);
    		
    		\path (c) ++(-150:0.75) node[circle,fill=black,draw,inner sep=0pt,minimum size=4pt] (n) {};
    		
    		\draw[decoration={brace,amplitude=10pt},decorate] (-0.3,0.6) -- node [midway,above,yshift=5pt,xshift=-10pt] {$\tilde P$} (2.5,2.1);

    		\draw[thick,dotted] (a) to (b);
    		\draw[thick, snake=zigzag] (e) -- (f) (i) -- (j) (m) -- (c) (c) -- (n);
    		\draw[thick] (b) -- (c) (b) -- (d) (b) -- (e) (f) -- (g) (f) -- (h) (f) -- (i) (j) -- (k) (j) -- (l) (j) -- (m);
    		
    	\end{tikzpicture}
    	
    	\caption{Intersection between $F_j,\, \tilde P$ at a vertex $w$.}\label{fig:Viz_fan_path_intersect}
    \end{subfigure}
    \begin{subfigure}[b]{0.45\textwidth}
        \centering
    	\begin{tikzpicture}
    	    \node[circle,fill=black,draw,inner sep=0pt,minimum size=4pt] (a) at (0,0) {};
    		\node[circle,fill=black,draw,inner sep=0pt,minimum size=4pt] (b) at (0,-1) {};
    		\path (b) ++(45:1) node[circle,fill=black,draw,inner sep=0pt,minimum size=4pt] (c) {};
    		\path (b) ++(0:1) node[circle,draw,inner sep=0pt,minimum size=4pt] (d) {$w$};
    		\path (b) ++(-45:1) node[circle,fill=black,draw,inner sep=0pt,minimum size=4pt] (e) {};
    		\path (b) ++(-90:1) node[circle,fill=black,draw,inner sep=0pt,minimum size=4pt] (f) {};
    		
    		\draw[decoration={brace,amplitude=10pt},decorate] (-0.2,-2) -- node [midway,above,yshift=0pt,xshift=-15pt] {$F_j$} (-0.2, 0);
    		
    		\path (f) ++(-90:1.5) node[circle,fill=black,draw,inner sep=0pt,minimum size=2pt] (g) {};
    		\path (g) ++(0:2) node[circle,fill=black,draw,inner sep=0pt,minimum size=2pt] (h) {};
    		\path (h) ++(90:2.5) node[circle,fill=black,draw,inner sep=0pt,minimum size=4pt] (i) {};
    		
    		\draw[decoration={brace,amplitude=10pt},decorate] (2.2,0) -- node [midway,above,yshift=0pt,xshift=15pt] {$\tilde F$} (2.2, -2);
    		
    		\path (i) ++(90:1) node[circle,fill=black,draw,inner sep=0pt,minimum size=4pt] (j) {};
    		\path (i) ++(135:1) node[circle,fill=black,draw,inner sep=0pt,minimum size=4pt] (k) {};
    		\path (i) ++(-135:1) node[circle,fill=black,draw,inner sep=0pt,minimum size=4pt] (l) {};
    		
    		\draw[thick,dotted] (a) to (b);
    		\draw[thick, snake=zigzag] (f) -- (g) -- (h) -- (i);
    		\draw[thick] (b) -- (c) (b) -- (d) (b) -- (e) (b) -- (f) (i) -- (j) (i) -- (k) (i) -- (l) (i) -- (d);
    		
    	\end{tikzpicture}
    	
    	\caption{Intersection between $F_j,\, \tilde F$ at a vertex $w$}\label{fig:Viz_fan_fan_intersect}
    \end{subfigure}
    \caption{Intersecting Vizing Chains at Step \ref{step:truncate_chain}.}
    \label{fig:Viz_intersect}
\end{figure}

The next two lemmas constitute the heart of our argument.
First, we prove an upper bound on $|\mathcal{I}^{(t)}(D, S, \y, uv, u, \boldcol)|$.

\begin{Lemma}\label{lemma:size_bounded_by_wt}
    Let $D \in \mathcal{D}^{(t)}$, $S \in \mathcal{S}^{(t)}$, $(uv, u)$ be such that $uv \in E$, $\y \in [\Delta]$, and let $\boldcol = (\col_0, \ldots, \col_t)$ be a sequence of colors.
    For $\ell \geq \kmax \geq 4$, we have $|\mathcal{I}^{(t)}(D, S, \y, uv, u, \boldcol)| \leq \wt(D)$.
\end{Lemma}

\begin{proof}\stepcounter{ForClaims}\renewcommand{\theForClaims}{\ref{lemma:size_bounded_by_wt}}
    We will prove this by induction on $t$.
    Let us first consider the case where $t = 0$.
    Then
    \[\mathcal{I}^{(t)}(D, S, \y, uv, u, \boldcol) = \set{(uv, v, \col_0)}.\]
    In particular, $|\mathcal{I}^{(t)}(D, S, \y, uv, u, \boldcol)| = 1 = \wt(())$ as the empty product is $1$.

    Now consider $t \geq 1$. 
    Consider any $I = (f, z, \col_0, \ell_1, \col_1, \ldots, \ell_t, \col_t) \in \mathcal{I}^{(t)}(D, S, \y, uv, u, \boldcol)$ and run Algorithm~\ref{alg:multi_viz_chain} with this input for $t$ iterations of the \textsf{while} loop.
    Let $C = F_0 + P_0 + \cdots + F_{k-1} + P_{k-1}$ and $F + P$ be the chain and candidate chain at the start of the $t$-th iteration, respectively.
    On Step~\ref{step:Pk} we compute $P_k \defeq P|\ell_t$ and then construct $\tilde F + \tilde P$ on Step~\ref{step:alpha_beta_order}.
    Note that for
    \begin{align*}
        D' &\defeq (d_1, \ldots, d_{t-1}), \quad S' \defeq (s_1, \ldots, s_{t-1}), \\ 
        v' &\defeq \Pivot(F), \quad u' \defeq \vstart(F), \quad \boldcol' \defeq (\col_0, \ldots, \col_{t-1}),
    \end{align*}
    we have $(f, z, \col_0, \ell_1, \col_1, \ldots, \ell_{t-1}, \col_{t-1}) \in \mathcal{I}^{(t)}(D', S', \y, u'v', u', \boldcol')$.
    Since 
    \[|\mathcal{I}^{(t)}(D', S', \y, u'v', u', \boldcol')| \leq \wt(D')\]
    by the inductive hypothesis, it is enough to show that there are at most $\val(d_t)$ choices for the tuple $(u', v', \ell_t)$.
    The following claim will assist with our proof.

    \begin{claim}\label{claim:one_step}
        Given the tuple $(\Start(\tilde F), \Pivot(\tilde F))$, there are at most $2\kmax$ choices for $(u', v', \ell_t)$.
    \end{claim}

    \begin{claimproof}
        Note that for some $t' < t$ the path $P$ was formed during the $t'$-th iteration.
        By Lemma~\ref{lemma:chain_formation}, we may determine $t'$.
        Furthermore, by Lemma~\ref{lemma:all_colors}, we may determine the colors $\alpha$ and $\beta$ such that $P$ is an $\alpha\beta$-path.
        By Lemma~\ref{lemma:non-intersecting_degrees} \ref{item:degree_end} and \ref{item:related_phi}, the vertex $\vend(F)$ is the endpoint of the $\alpha\beta$-path passing through $\Start(\tilde F)$ closer to $\Pivot(\tilde F)$.
        Once we have $\vend(F)$, we note that by the non-intersecting property, $\phi(\End(F))$ is one of the first $\kmax$ entries in $\col_{t'}$.
        Therefore, there are at most $\kmax$ choices for $v'$.
        Given $\End(F)$ and the $\alpha\beta$-path from $\vend(F)$ to $\Start(\tilde F)$, we may determine $\ell_t$.
        In order to determine $u'$, there are two cases to consider.
        First, $k = 1$.
        In this case, $u'v' = f$ and $v' = z$.
        Therefore, $\y$ determines $u'$.
        Next, $k \geq 2$.
        Applying a similar argument as earlier, we may determine the colors $\gamma, \delta$ such that $P_{k-1}$ is a $\gamma\delta$-path.
        As a result of Lemma~\ref{lemma:non-intersecting_degrees} \ref{item:related_phi}, we must have $\phi(u'v') \in \set{\gamma, \delta}$.
        Therefore, there are at most $2$ choices of $u'$, as desired.
    \end{claimproof}

    \vspace{10pt}
    As a result of the above claim, it is now enough to show that there are at most $\val(d_t)/(2\kmax)$ choices for the tuple $(\Start(\tilde F), \Pivot(\tilde F))$.
    We will split into cases depending on the value $d_t$:
    \begin{enumerate}[label=\ep{\textbf{Case \arabic*}}, wide] 
        \item\label{case:d_i_1} $d_t = 1$.
        Here, we have $\Start(\tilde F) = uv$ and $\Pivot(\tilde F) = v$.

        \item\label{case:d_i_neg} $d_t < 0$.
        Here, we have $v = \Pivot(F_j)$ and $u = \vstart(F_j)$ where $j = k + d_t$ and the first intersection between $\tilde F + \tilde P$ and $C + F + P$ was with $F_j + P_j$.
        Once again, by Lemma~\ref{lemma:chain_formation}, we may determine $t' < t$ such that $F_j + P'$ was constructed during the $t'$-th iteration (where $P_j \subseteq P'$).
        Furthermore, by Lemma~\ref{lemma:all_colors}, we may determine the colors $\alpha$ and $\beta$ such that $P'$ is an $\alpha\beta$-path.

        Recall that we have $V(\tilde F + \tilde P) \cap V(F_j) \neq \0$ or $E(\tilde F + \tilde P) \cap \IE(P_j) \neq \0$ (see Fig.~\ref{fig:Viz_intersect} for examples).
        We will consider each case separately.

        \begin{enumerate}[label=\ep{\textit{Subcase 2\alph*}}, wide]
            \item\label{subcase_F} $V(\tilde F + \tilde P) \cap V(F_j) \neq \0$.
            Let $w$ be the common vertex at the first intersection.
            There are at most $\kmax + 1 \leq 2\kmax$ choices for $w$.

            Let us first consider the case that $w \in V(\tilde F)$.
            Suppose $w \in \set{\Pivot(\tilde F), \vstart(\tilde F)}$.
            As $\Start(\tilde F) \in E(P_k)$ and we may determine the colors on $P_k$ by Lemma~\ref{lemma:all_colors}, there are at most $2$ choices for the vertex in $\set{\Pivot(\tilde F), \vstart(\tilde F)} \setminus \set{w}$ and therefore, at most $2$ choices for the tuple $(\Start(\tilde F), \Pivot(\tilde F))$ in each of these cases.
            Now suppose $w \notin \set{\Pivot(\tilde F), \vstart(\tilde F)}$.
            Then it must be the case that $\phi(w\Pivot(\tilde F)) = \Shift(\phi, C + F_k + P_k)(w\Pivot(\tilde F))$.
            If not, we would have truncated earlier.
            Therefore, there are at most $\kmax$ choices for $\Pivot(\tilde F)$ (as $\phi(w\Pivot(\tilde F))$ is one of the first $\kmax$ entries of $\col_t$), and at most $2$ choices for $\vstart(\tilde F)$ given $\Pivot(\tilde F)$ (as we know the colors on $P_k$).
            It follows that there are at most $4 + 2\kmax \leq 3\kmax$ choices for $(\Start(\tilde F), \Pivot(\tilde F))$ in this case.

            Now, suppose $w \in V(\tilde P)$.
            As we truncate at the first intersection, we may conclude that $w$ and $\vend(\tilde F)$ are $(\phi, \gamma\delta)$-related, where $\tilde P$ is a $\gamma\delta$-path (the colors $\gamma$ and $\delta$ can be determined as a result of Lemma~\ref{lemma:all_colors}).
            In particular, there are at most $2$ choices for $\vend(\tilde F)$.
            Following an identical argument as that of Claim~\ref{claim:one_step}, there are at most $2\kmax$ choices for $(\Start(\tilde F), \Pivot(\tilde F))$ from here.
            Therefore, there are at most $4\kmax$ choices for $(\Start(\tilde F), \Pivot(\tilde F))$ in this case.

            Putting the above together, we conclude there are at most $2\kmax(3\kmax + 4\kmax) = 14\kmax^2$ choices for $(\Start(\tilde F), \Pivot(\tilde F))$.

            \item\label{subcase_P} $E(\tilde F + \tilde P) \cap \IE(P_j) \neq \0$.
            By the non-intersecting property, $\phi(\End(F_j))$ is one of the first $\kmax$ entries in $\col_{t'}$ and so there are at most $\kmax$ choices for $\vend(F_j)$.
            Similarly, by Lemma~\ref{lemma:non-intersecting_degrees} \ref{item:related_phi} $P_j$ is an initial segment of $G(\End(F_j); \phi, \alpha\beta)$ of length at most $2\ell$.
            Let $e = ww'$ be the common edge at the first intersection.
            There are at most $2\ell$ choices for $e$.

            Let us first consider the case that $e \in E(\tilde F)$.
            Clearly, $\Pivot(\tilde F) \in \set{w, w'}$ and so there are at most $2$ choices for $\Pivot(\tilde F)$.
            Note that $e \neq \Start(\tilde F)$ as this would imply $P_k$ intersects $P_j$ contradicting the fact that $C + F + P$ is non-intersecting.
            Furthermore, as we are truncating at the first intersection, we may conclude that $\phi(\Start(\tilde F)) = \Shift(\phi, C)(\Start(\tilde F))$.
            Therefore, there are at most $2$ choices for $\Start(\tilde F)$ as we know $\phi(\Start(\tilde F))$ is one of the colors on the path $P_k$, which may be determined by Lemma~\ref{lemma:all_colors}.
            In this case, there are at most $4$ choices for $(\Start(\tilde F), \Pivot(\tilde F))$.

            Now, suppose $e \in E(\tilde P)$.
            Either $w$ and $\vend(\tilde F)$ or $w'$ and $\vend(\tilde F)$ are $(\phi, \gamma\delta)$-related where $\tilde P$ is a $\gamma\delta$-path.
            If not, we would have truncated earlier.
            So there are at most $4$ choices for $\vend(\tilde F)$.
            From here, as in previous cases, we may conclude there are at most $8\kmax$ choices for $(\Start(\tilde F), \Pivot(\tilde F))$ in this case.

            Putting the above together, we conclude there are at most $2\ell\kmax(4 + 8\kmax) \leq 20\ell\kmax^2$ choices for $(\Start(\tilde F), \Pivot(\tilde F))$.

        \end{enumerate}

        Putting both subcases together, we conclude that there are at most $20\ell\kmax^2 + 14\kmax^2 \leq 25\ell\kmax^2$ choices for $(\Start(\tilde F), \Pivot(\tilde F))$ in this case (where we use the fact that $\ell \geq 3$).

        \item\label{case:d_i_0} $d_i = 0$.
        As a result of Lemma~\ref{lemma:intersection_prev}, we need only consider \ref{subcase_F} above and so there are at most $14\kmax^2$ choices for $(\Start(\tilde F), \Pivot(\tilde F))$.
        
    \end{enumerate}
    
    This covers all cases and completes the proof.
\end{proof}

To assist with the remainder of the proof, we define the following sets:
\[\mathcal{D}_k^{(t)} \defeq \left\{D \in \mathcal{D}^{(t)}\,:\, \sum_{i = 0}^td_i = k\right\}.\]
Given $D \in \mathcal{D}^{(t)}$, $S \in \mathcal{S}^{(t)}$, a pair $(uv, u)$ such that $uv \in E$, and $\y \in [\Delta]$, we define $\mathcal{I}^{(t)}(D, S, \y, uv, u)$ to be the union of $\mathcal{I}^{(t)}(D, S, \y, uv, u, \boldcol)$ over all possible options for $\boldcol$.
For an input sequence $I = (f, z, \col_0, \ell_1, \col_1, \ldots, \ell_t, \col_t) \in \mathcal{I}^{(t)}$, we let $\P[I]$ denote the probability of $I$ occurring during the first $t$ iterations of the \textsf{while} loop of Algorithm~\ref{alg:multi_viz_chain} with input $(\phi, f, z)$.
The next lemma bounds the probability of input sequences in $\mathcal{I}^{(t)}(D, S, \y, uv, u)$.

\begin{Lemma}\label{lemma:prob_bound_wt}
    Let $D \in \mathcal{D}_k^{(t)}$, $S \in \mathcal{S}^{(t)}$, $(uv, u)$ be such that $uv \in E$, and $\y \in [\Delta]$.
    Then,
    \[\sum_{I \in \mathcal{I}^{(t)}(D, S, \y, uv, u)}\P[I] \,\leq\, \left\{\begin{array}{cc}
        \dfrac{\wt(D)}{\ell^t} & k = 0, \vspace{5pt} \\
        \dfrac{2\,\wt(D)}{\ell^t\,\eps\Delta} & k > 0.
    \end{array}\right.\]
\end{Lemma}

\begin{proof}
    We let $\mathcal{C}$ denote the set of all sequences $\boldcol$ such that $\mathcal{I}^{(t)}(D, S, \y, uv, u, \boldcol) \neq \0$.
    We say the sequences $\boldcol \in \mathcal{C}$ are \emphd{valid}.
    Consider the following:
    \begin{align*}
        \sum_{I \in \mathcal{I}^{(t)}(D, S, \y, uv, u)}\P[I] &= \sum_{I \in \mathcal{I}^{(t)}(D, S, \y, uv, u)}\sum_{\boldcol\in \mathcal{C}}\P[\boldcol]\,\P[I\mid\boldcol] \\
        &= \sum_{\boldcol\in \mathcal{C}}\P[\boldcol]\sum_{I \in \mathcal{I}^{(t)}(D, S, \y, uv, u)}\P[I\mid \boldcol],
    \end{align*}
    where $\P[\boldcol]$ is the probability of observing the sequence of colors in $\boldcol$.
    By Lemma~\ref{lemma:size_bounded_by_wt} and since each $\ell_i$ is chosen independently and uniformly at random from $\ell$ options, it follows that
    \begin{align}\label{eqn:prob_wt_col}
        \sum_{I \in \mathcal{I}^{(t)}(D, S, \y, uv, u)}\P[I] \,\leq\, \sum_{\boldcol\in \mathcal{C}}\P[\boldcol]\sum_{I \in \mathcal{I}^{(t)}(D, S, \y, uv, u, \boldcol)}\frac{1}{\ell^t} \,\leq\, \frac{\wt(D)}{\ell^t}\sum_{\boldcol\in \mathcal{C}}\P[\boldcol].
    \end{align}
    Note that $\sum_{\boldcol\in \mathcal{C}}\P[\boldcol] \leq 1$ trivially.
    With this bound, the proof for $k = 0$ is complete.
    It turns out that for $k > 0$, we can provide a better bound on $\sum_{\boldcol\in \mathcal{C}}\P[\boldcol] \leq 1$, which will be important for the proof of Proposition~\ref{theo:num_iter_while}.

    Suppose $k > 0$.
    Note that for any $I \in \mathcal{I}^{(t)}(D, S, \y, uv, u)$, we have $\tau(I) = (uv, u)$ and the chain after running Algorithm~\ref{alg:multi_viz_chain} with input sequence $I$ is a $k$-step Vizing chain.
    In particular, one of the colors on the $k$-th chain is $\alpha = \phi(uv)$.
    By Lemma~\ref{lemma:chain_formation}, we may determine the iterations $t_0, \ldots, t_{k-1}$ at which each of the Vizing chains within the $k$-step chain would have been constructed.
    Let $k^\star \defeq \max\set{0 \leq i < k\,:\, s_{t_i} = 0}$.
    For $\boldcol = (\col_0, \ldots, \col_t)$ to be valid, it must be the case that $\alpha$ appears as one of the last two entries in $\col_{t_{k^\star}}$.
    Let $A$ be the event that $\alpha$ appears as the last entry in $\col_{t_{k^\star}}$ and let $B$ the event that $\alpha$ appears as the second to last entry in $\col_{t_{k^\star}}$.
    Then we have
    \begin{align}\label{eqn:prob_A_B}
        \sum_{\boldcol}\P[\boldcol] \leq \P[A] + \P[B].
    \end{align}
    Let us consider $\P[A]$ (the analysis for $\P[B]$ is identical).
    Note that $\alpha$ is chosen from the missing set of some random vertex $z$ under some random coloring $\psi$.
    We have
    \begin{align*}
        \P[A] &= \sum_{z', \psi'}\P[z = z', \psi = \psi']\P[A\mid z = z', \psi = \psi'] \\
        &= \sum_{z', \psi'}\frac{\P[z = z', \psi = \psi']}{|M(\psi', z')|} \\
        &\leq \sum_{z', \psi'}\frac{\P[z = z', \psi = \psi']}{\eps\Delta} = \frac{1}{\eps\Delta},
    \end{align*}
    where the last step follows since $|M(\phi, x)| \geq \eps\Delta$ for any $x \in V$ and any partial coloring $\phi$.
    The desired bound follows by \eqref{eqn:prob_wt_col}, \eqref{eqn:prob_A_B}, and the computation above.
\end{proof}

Let us now bound $\wt(D)$.

\begin{Lemma}\label{lemma:wt_bound}
    Let $D \in \mathcal{D}_k^{(t)}$. 
    Then, $\wt(D) \leq \left(100\ell\kmax^4\right)^{t/2}\left(25\ell\kmax^2\right)^{-k/2}$, for $\ell \geq 8\kmax^2$.
\end{Lemma}

\begin{proof}
    Let $D = (d_1, \ldots, d_t) \in \mathcal{D}_k^{(t)}$.
    Define
    \begin{align*}
        I\defeq \{i\,:\,d_i = 1\}, \quad J\defeq \{i\,:\,d_i = 0\}, \quad K\defeq [t] \setminus (I\cup J).
    \end{align*}
    Using the definition of $\wt(D)$ and $\val(z)$, we write
    \begin{align*}
        \wt(D) \,&=\, (2\kmax)^{|I|}\,(28\kmax^3)^{|J|}\, (50\ell\kmax^3)^{|K|} \\
        &=\, (2\kmax)^t\,(14\kmax^2)^{|J|}\, (25\ell\kmax^2)^{|K|},
    \end{align*}
    where we use the fact that $|I| + |J| + |K| = t$.
    Note the following:
    \[k \,=\, |I| - \sum_{k \in K}|d_k| \,\leq\, |I| - |K| \,=\, t - |J| - 2|K|.\]
    It follows that $|K| \leq (t - k - |J|)/2$. With this in hand, we have, for $\ell \geq 8\kmax^2$,
    \begin{align*}
        \wt(D) \,&\leq\, (2\kmax)^t\,(14\kmax^2)^{|J|}\, (25\ell\kmax^2)^{(t - k - |J|)/2} \\
        &\leq\, \left(100\ell\kmax^4\right)^{t/2}\left(25\ell\kmax^2\right)^{-k/2}\left(\frac{8\kmax^2}{\ell}\right)^{|J|/2} \,\leq\, \left(100\ell\kmax^4\right)^{t/2}\left(25\ell\kmax^2\right)^{-k/2}. \qedhere
    \end{align*}
\end{proof}

We will use the following result of \cite[Lemma~6.5]{bernshteyn2023fast}:

\begin{Lemma}\label{lemma:Dst}
    $|\mathcal{D}_k^{(t)}|\leq 4^t$.
\end{Lemma}

We are now ready to prove Proposition~\ref{theo:num_iter_while}.

\begin{proof}[Proof of Proposition~\ref{theo:num_iter_while}]
    First, we note the folllowing:
    \begin{align*}
        \P[\text{Algorithm~\ref{alg:multi_viz_chain} does not terminate after $t$ iterations}] &= \sum_{I \in \mathcal{I}^{(t)}}\P[I].
    \end{align*}
    By Lemmas~\ref{lemma:prob_bound_wt}, \ref{lemma:wt_bound}, and \ref{lemma:Dst}, we have
    \begin{align*}
        \sum_{I \in \mathcal{I}^{(t)}}\P[I] &= \sum_{e\in E, u \in e}\sum_{D \in \mathcal{D}^{(t)}}\sum_{S \in \mathcal{S}^{(t)}}\sum_{\y = 1}^\Delta\sum_{I \in \mathcal{I}^{(t)}(D, S, \y, uv, u)}\P[I] \\
        &= \sum_{e\in E, u \in e}\sum_{k = 0}^t\sum_{D \in \mathcal{D}_k^{(t)}}\sum_{S \in \mathcal{S}^{(t)}}\sum_{\y = 1}^\Delta\sum_{I \in \mathcal{I}^{(t)}(D, S, \y, uv, u)}\P[I] \\
        [\text{by Lemma \ref{lemma:prob_bound_wt}}]\qquad &\leq \sum_{e\in E, u \in e}\sum_{k = 0}^t\sum_{D \in \mathcal{D}_k^{(t)}}\sum_{S \in \mathcal{S}^{(t)}}\sum_{\y = 1}^\Delta\frac{\wt(D)}{\ell^t}\left(\frac{2}{\eps\Delta}\right)^{\bbone\set{k > 0}} \\
        [\text{by Lemma \ref{lemma:wt_bound}}]\qquad &\leq \sum_{e\in E, u \in e}\sum_{k = 0}^t\sum_{\y = 1}^\Delta\sum_{D \in \mathcal{D}_k^{(t)}}\sum_{S \in \mathcal{S}^{(t)}}\left(\frac{100\kmax^4}{\ell}\right)^{t/2}\left(25\ell\kmax^2\right)^{-k/2}\left(\frac{2}{\eps\Delta}\right)^{\bbone\set{k > 0}} \\
        [\text{by Lemma \ref{lemma:Dst}}]\qquad &\leq \sum_{e\in E, u \in e}\sum_{k = 0}^t\sum_{S \in \mathcal{S}^{(t)}}\left(\frac{1600\kmax^4}{\ell}\right)^{t/2}\sum_{\y = 1}^\Delta\left(25\ell\kmax^2\right)^{-k/2}\left(\frac{2}{\eps\Delta}\right)^{\bbone\set{k > 0}} \\
        [|\mathcal{S}^{(t)}| \leq 2^t]\qquad &\leq \left(\frac{6400\kmax^4}{\ell}\right)^{t/2}\sum_{e\in E, u \in e}\sum_{k = 0}^t\sum_{\y = 1}^\Delta\left(25\ell\kmax^2\right)^{-k/2}\left(\frac{2}{\eps\Delta}\right)^{\bbone\set{k > 0}}.
    \end{align*}
    Note that when $k = 0$, we have $\y = l$ where $u$ is the $l$-th neighbor of the vertex $v$ such that $e = uv$.
    With this in hand, we may simplify the above as follows:
    \begin{align}
        \sum_{I \in \mathcal{I}^{(t)}}\P[I] &\leq \left(\frac{6400\kmax^4}{\ell}\right)^{t/2}\sum_{e\in E, u \in e}\left(1 + \sum_{k = 1}^t\sum_{\y = 1}^\Delta\frac{2}{\eps\Delta\left(25\ell\kmax^2\right)^{k/2}}\right) \nonumber\\ 
        &= \left(\frac{6400\kmax^4}{\ell}\right)^{t/2}\sum_{e\in E, u \in e}\left(1 + \sum_{k = 1}^t\frac{2}{\eps\left(25\ell\kmax^2\right)^{k/2}}\right) \nonumber\\
        &\leq \left(\frac{6400\kmax^4}{\ell}\right)^{t/2}\sum_{e\in E, u \in e}\sum_{k = 0}^t\frac{2}{\eps\left(25\ell\kmax^2\right)^{k/2}} \nonumber\\
        &\leq \frac{8m}{\eps}\left(\frac{6400\kmax^4}{\ell}\right)^{t/2}, \label{eq: bound_on_P[I]}
    \end{align}
    for $25\ell\kmax^2 \geq 4$.
\end{proof}

We remark that the condition $\kmax \geq 16/\eps$ in the statement of Proposition~\ref{theo:num_iter_while} is not necessary for any of the computations in this section.
However, as a result of Proposition~\ref{lemma:fan_number_of_tries}, it is sufficient to ensure each call to Algorithm~\ref{alg:rand_fan} does in fact succeed (and so the input sequences are well defined).

We conclude this section with a lemma regarding the runtime of Algorithm~\ref{alg:multi_viz_chain}.

\begin{Lemma}\label{lemma:runtime_RMSVA}
    Consider running Algorithm~\ref{alg:multi_viz_chain} on input $(\phi, xy, x)$, and suppose the algorithm succeeds during the $(t+1)$-th iteration of the \textsf{while} loop.
    For $0 \leq i \leq t$, let $R_i$ denote the runtime of the $i$-th call to Algorithm~\ref{alg:rand_chain}.
    Then, the runtime of Algorithm~\ref{alg:multi_viz_chain} is $O(\ell\,t + \sum_{i = 0}^tR_i)$.
\end{Lemma}

\begin{proof}
    Let $I \in \mathcal{I}^{(t)}$ be an input sequence and let $D = (d_1, \ldots, d_t) \defeq D(I)$. We bound the runtime of the $i$-th iteration in terms of $d_i$. 
    Note that while updating the coloring $\psi$, we must update the sets of missing colors as well. 
    By the way we store our missing sets (see \S\ref{section: data_structures}), each update to $M(\cdot)$ takes $O(1)$ time. 
    Let us now consider two cases depending on whether $d_i$ is positive.
        
    \begin{enumerate}[label=\ep{\textbf{Case \arabic*}}, wide] 
        \item $d_i = 1$. Here, we test for success, shorten the path $P$ to $P_k$, update $\psi$ and the hash map $\visited$, call Algorithm \ref{alg:rand_chain}, and check that the resulting chain is non-intersecting. All of these steps can be performed in time $O(R_i + \ell)$.
        
        \item $d_i \leq 0$. Here we again conduct all the steps mentioned in the previous case, which takes $O(R_i + \ell)$ time.
        Once we locate the intersection, we then update $\psi$ and the hash map $\visited$, which takes $O((|d_i| + 1)\ell)$ time (we are adding $1$ to account for the case $d_i = 0$).
    \end{enumerate}
    
    Since $\sum_{i=1}^t d_i \geq 0$, we have
    \[\sum_{i,\, d_i \leq 0}|d_i| \leq \sum_{i,\, d_i = 1}d_i \,\leq\, t.\]
    It follows that the running time of $t$ iterations of the \textsf{while} loop is $O(\ell\,t + \sum_{i = 0}^tR_i)$, as desired.
\end{proof}

\section{Proof of Theorem~\ref{theo:main_theo}}\label{section: sequential}

In this section we prove Theorem~\ref{theo:main_theo}.
The algorithm takes as input a graph $G$ of maximum degree $\Delta$ along with parameters $\eps > 0$ and $\ell \in \N$, and outputs a proper $(1+\eps)\Delta$-edge-coloring of $G$.
At each iteration, the algorithm picks an uncolored edge uniformly at random and colors it by finding an augmenting multi-step Vizing chain using Algorithm~\ref{alg:multi_viz_chain}. See Algorithm~\ref{alg:seq} for the details.

\begin{algorithm}[h]\algsize
\caption{Sequential Coloring with Multi-Step Vizing Chains}\label{alg:seq}
\begin{flushleft}
\textbf{Input}: A graph $G = (V, E)$ of maximum degree $\Delta$ and parameters $\eps > 0$ and $\ell \in \N$. \\
\textbf{Output}: A proper $(1+ \eps)\Delta$-edge-coloring $\phi$ of $G$.
\end{flushleft}
\begin{algorithmic}[1]
    \State $U \gets E$, $\phi(e) \gets \blank$ for each $e \in U$.
    \While{$U \neq \0$}
        \State Pick an edge $e \in U$ and a vertex $x \in e$ uniformly at random.
        \State $(C, \alpha) \gets \hyperref[alg:multi_viz_chain]{\mathsf{RMSVA}}(\phi, e, x, \ell)$ \Comment{Algorithm \ref{alg:multi_viz_chain}}
        \State $\phi \gets \hyperref[alg:chain_aug]{\aug}(\phi, C, \alpha)$ \Comment{Algorithm~\ref{alg:chain_aug}}
        \State $U \gets U \setminus \set{e}$
    \EndWhile
    \State \Return $\phi$
\end{algorithmic}
\end{algorithm}

The correctness of Algorithm~\ref{alg:seq} follows from the results of \S\ref{subsec: poc}.
To assist with the analysis, we define the following variables:
\begin{align*}
    \phi_i &\defeq \text{ the coloring at the start of iteration } i, \\
    T_i &\defeq \text{ the number of iterations of the \textsf{while} loop during  the $i$-th call to Algorithm \ref{alg:multi_viz_chain}},  \\
    C_i &\defeq \text{ the chain produced by Algorithm \ref{alg:multi_viz_chain} at the $i$-th iteration}, \\
    U_i &\defeq \set{e\in E\,:\, \phi_i(e) = \blank}.
\end{align*}
Furthermore, we set $\kmax = \Theta(1/\eps)$ and $\ell = \Theta(1/\kmax^4)$, where the implicit constant factors are assumed to be sufficiently large for the computations that follow to hold.
Let $T = \sum_{i = 1}^mT_i$.
Note that there are $T$ total calls to Algorithm~\ref{alg:rand_chain} (as there is no call to Algorithm~\ref{alg:rand_chain} in the final iteration of the \textsf{while} loop in Algorithm~\ref{alg:multi_viz_chain} and there is one call before beginning the \textsf{while} loop).
Moreover, since the chain $C_i$ can be augmented in time $O(\length(C_i)) = O(\ell\,T_i)$, as a result of Lemma~\ref{lemma:runtime_RMSVA} we may conclude that the runtime of Algorithm~\ref{alg:multi_viz_chain} is at most
\begin{align}\label{eq:runtime}
    O\left(\ell\,T + \sum_{i = 1}^{T}R_i\right),
\end{align}
where $R_i$ is the runtime of the $i$-th call to Algorithm~\ref{alg:rand_chain}.
It remains to show that \eqref{eq:runtime} is $O(m\,\log(1/\eps)/\eps^4)$ with high probability.
To this end, we define a number of random variables.

For the $i$-th call to Algorithm~\ref{alg:rand_chain}, let $S_i$ be the number of times we reach Step~\ref{step:return} during the corresponding call to Algorithm~\ref{alg:rand_fan}, and let $K_i$ be the number of calls to Algorithm~\ref{alg:rand_col}.
Note that 
\begin{align}\label{eqn: num_rand_col}
    K_i \,\leq\, \kmax(S_i + 1) + 1.
\end{align}
For $i \in [m]$ and $j \in [K_i]$, we let $Y_{i,j}$ denote the runtime of the corresponding call to Algorithm~\ref{alg:rand_col}.
Since Step~\ref{step: loop_in_fan} of Algorithm~\ref{alg:rand_fan} takes $O(\kmax)$ time and we reach this step at most $K_i$ times, 
\[R_i = O\left(\ell + \kmax\,K_i + \sum_{j = 1}^{K_i}Y_{i,j}\right).\]
Plugging this into \eqref{eq:runtime}, we conclude that the runtime of Algorithm~\ref{alg:seq} is
\begin{align}\label{eqn: runtime}
    O\left(\ell\,T + \ell\,m + \kmax\sum_{i = 1}^TK_i + \sum_{i = 1}^T\sum_{j = 1}^{K_i}Y_{i,j}\right) \,=\, O\left(\ell\,T + \ell\,m + \kmax K + \sum_{i = 1}^KY_{i}\right),
\end{align}
where $K \defeq \sum_{j = 1}^TK_j$ and $Y_i$ is the runtime of the corresponding call to Algorithm~\ref{alg:rand_col}.

In what follows, we will bound the random 
the variables $T$, $K$, and $Y \defeq \sum_{i = 1}^KY_{i}$. 
Let us first consider the random variable $T$.
We begin with the following lemma.

\begin{Lemma}\label{lemma:t_i_prob_bound}
    For all $t_1$, \ldots, $t_i$, we have $\P[T_i > t_i\mid T_1 > t_1, \ldots, T_{i-1} > t_{i-1}] \leq \frac{4m}{\eps\,|U_i|}\,\left(\frac{6400\kmax^{4}}{\ell}\right)^{t_i/2}$.
\end{Lemma}

\begin{proof}
    To analyze the number of iterations of the \textsf{while} loop in the $i$-th call to Algorithm \ref{alg:multi_viz_chain}, we use the definition of input sequences from \S\ref{subsec: rmsva analysis}. Note that an input sequence $I = (f, z, \col_0, \ell_1, \col_1, \ldots, \ell_t, \col_t)$ contains not only the random choices $\ell_i$ and $\col_i$, but also the starting edge $f$ and the vertex $z \in f$. 
    Since at every iteration of Algorithm~\ref{alg:seq}, the starting edge and vertex are chosen randomly, we have
    \[\P[T_i > t_i\mid \phi_i] \,\leq\, \frac{1}{2|U_i|}\sum_{I \in \mathcal{I}^{(t_i)}}\P[I].\]
    By the inequality \eqref{eq: bound_on_P[I]} from the proof of Proposition~\ref{theo:num_iter_while}, we have
    \[\sum_{I \in \mathcal{I}^{(t_i)}}\P[I] \,\leq\,
    \frac{8m}{\eps}\left(\frac{6400\kmax^4}{\ell}\right)^{t_i/2}.\]
    Furthermore, $T_i$ is independent of $T_1$, \ldots, $T_{i-1}$ given $\phi_i$. Therefore, 
    \begin{align*}
        \P[T_i > t_i\mid T_1 > t_1, \ldots T_{i-1} > t_{i-1}] \,\leq\, \frac{1}{2|U_i|}\sum_{I \in \mathcal{I}^{(t_i)}}\P[I] \,\leq\, \frac{4m}{\eps\,|U_i|}\,\left(\frac{6400\kmax^{4}}{\ell}\right)^{t_i/2}, 
    \end{align*}
    as desired.
\end{proof}

For the rest of the proof, we fix the following parameters:
\[t = \Theta(m\,\log(1/\eps)), \quad s = \Theta(t/\eps), \quad k = \kmax(s + m) + m,\]
where the implicit constants are assumed to be sufficiently large.

\begin{Lemma}\label{lemma:T_bound}
    $\P[T > t] = \eps^{\Omega(m)}$.
\end{Lemma}

\begin{proof}
    Note that if $T > t$, then there exist non-negative integers $t_1$, \ldots, $t_m$ such that $\sum_it_i = t$ and $T_i > t_i$ for each $i$.
    Using the union bound we obtain
    \begin{align*}
        \P[T > t] \,&\leq\, \sum_{\substack{t_1, \ldots, t_m \\ \sum_it_i = t}}\P[T_1 > t_1, \ldots, T_m > t_m] \\
        &=\, \sum_{\substack{t_1, \ldots, t_m \\ \sum_it_i = t}}\,\prod_{i = 1}^m\P[T_i > t_i\mid T_1 > t_1, \ldots, T_{i-1} > t_{i-1}] \\
        [\text{by Lemma \ref{lemma:t_i_prob_bound}}]\qquad &\leq\, \sum_{\substack{t_1, \ldots, t_m \\ \sum_it_i = t}}\,\prod_{i = 1}^m\frac{4m}{\eps\,|U_i|}\,\left(\frac{6400\kmax^{4}}{\ell}\right)^{t_i/2} \\
        &=\, \sum_{\substack{t_1, \ldots, t_m \\ \sum_it_i = t}}\frac{1}{m!}\,\left(\frac{4m}{\eps}\right)^m\,\left(\frac{6400\kmax^{4}}{\ell}\right)^{t/2} \\
        [m! \geq (m/e)^m] \qquad &\leq\, \binom{t+m-1}{m-1} \,\left(\frac{4e}{\eps}\right)^m\,\left(\frac{6400\kmax^{4}}{\ell}\right)^{t/2} \\
        [{\textstyle {a \choose b} \leq (ea/b)^b}] \qquad &\leq\, \left(e\left(1 + \frac{t}{m-1}\right)\right)^{m-1}\,\left(\frac{4e}{\eps}\right)^m\,\left(\frac{6400\kmax^{4}}{\ell}\right)^{t/2} \\
        [(1 + a)\leq e^a] \qquad &\leq\, \left(\frac{4e^2}{\eps}\right)^m\,\left(\frac{6400e^2\kmax^{4}}{\ell}\right)^{t/2}.
    \end{align*}
    The desired bound now follows, assuming the hidden constants in the definitions of $t,\,\kmax,\, \ell$ are sufficiently large.
\end{proof}

The following proposition will assist with the remainder of the proof:

\begin{prop}\label{prop: geometric}
    Let $X_1, \ldots, X_n$ be i.i.d geometric random variables with probability of success $p$.
    Then for any $\mu \geq n/p$,
    \[\P\left[\sum_{i = 1}^nX_i \,\geq\, \frac{6\mu}{p(1-p)}\right] \leq \exp\left(-\frac{\mu}{1-p}\right).\]
\end{prop}

This follows as a corollary to the following result of Assadi:

\begin{prop}[{\cite[Proposition 2.2]{Assadi}}]\label{prop: assadi}
    Let $\mu_1, \ldots, \mu_n$ be non-negative numbers and $X_1, \ldots, X_n$ be independent non-negative random variables such that for every $i \in [n]$:
    \[\E[X_i] \leq \mu_i \quad \text{and} \quad \P[X_i \geq x] \leq \eta\cdot e^{-\kappa x}\mu_i,\]
    for some $\eta$, $\kappa > 0$. Define $X \defeq \sum_{i = 1}^nX_i$ and let $\mu \geq \sum_{i = 1}^n \mu_i$. Then,
    \[\P\left[X \geq \frac{6\eta\,\mu}{\kappa^2}\right] \,\leq\, \exp\left(-\frac{\eta\,\mu}{\kappa}\right).\]
\end{prop}

\begin{proof}[Proof of Proposition~\ref{prop: geometric}]
    Let $\mu_i = 1/p$, $\eta = p/(1-p)$, and $\kappa = p$. 
    Note that
    \[\P[X_i \geq x] = (1 - p)^{x-1} = \frac{p}{1-p}\,(1 - p)^x\,\frac{1}{p} \leq \eta\, e^{-\kappa x}\,\mu_i.\]
    Additionally, $\E[X_i] = \mu_i$.
    Therefore, we may apply Proposition~\ref{prop: assadi} to complete the proof.
\end{proof}

Let us now consider the random variable $S = \sum_{i = 1}^TS_i$.

\begin{Lemma}\label{lemma:total_fail_rand_fan}
    $\P[S > s] \,=\, \eps^{\Omega(m)}$.
\end{Lemma}

\begin{proof}
    To assist with our proof, we define the following for $i \in [t]$:
    \[\tilde S_i = \left\{\begin{array}{cc}
        S_i & i \leq T; \\
        0 & \text{otherwise.}
    \end{array}\right., \qquad \tilde S = \sum_{j = 1}^t\tilde S_j.\]
    In particular, if $T \leq t$, then $\tilde{S} = S$. Note the following as a result of Lemma~\ref{lemma:T_bound}:
    \begin{align*}
        \P[S > s] &= \P[S > s,\, T \leq t] + \P[S > s,\, T > t] \\
        &\leq \P[S > s,\, T \leq t] + \P[T > t] \\
        &\leq \P[\tilde S > s] + \eps^{\Omega(m)}.
    \end{align*}
    It is now sufficient to show that $\P\left[\tilde S > s\right] = \eps^{\Omega(m)}$.
    
    To this end, let $X_1$, \ldots, $X_t$ be i.i.d. geometric random variables with probability of success
    \[p \defeq 1 - \exp\left(-\eps^2\kmax/100\right).\]
    Take any $1 \leq i \leq t$. By definition, for any $s_i \geq 0$, we have
    \[\P\left[\tilde S_i \geq s_i \mid T < i\right] = \bbone\set{s_i = 0} \leq (1 - p)^{s_i}.\]
    Additionally, by Proposition~\ref{lemma:fan_number_of_tries}, we have that
    \[\P\left[\tilde S_i \geq s_i \mid T \geq i, \, \psi, e, x\right] \leq (1 - p)^{s_i},\]
    where $(\psi, e, x)$ is the input to Algorithm~\ref{alg:rand_fan} at the relevant step.
    As the expression on the right-hand side of the above inequalities is independent of $T$, $\psi$, $e$, $x$, we conclude that the random variable $X_i - 1$ stochastically dominates $\tilde S_i$.
    Furthermore, note that
    \[\mu = \frac{(s-t)\,p\,(1-p)}{6} \geq t/p.\]
    In particular, we may apply Proposition~\ref{prop: geometric} to get
    \[\P\left[\tilde S > s\right] \,\leq\, \P[X > s - t] \,\leq\, \exp\left(-\frac{(s-t)\,p}{6}\right) \,=\, \exp\left(-\Omega(\eps\,s)\right),\]
    where the last step follows by the definition of $s$ and since $p = \Theta(\eps)$ by the definition of $\kmax$.
\end{proof}

As a result of \eqref{eqn: num_rand_col}, we obtain the following as a corollary to the above lemma.

\begin{Lemma}\label{lemma:total_num_rand_col}
    $\P[K \geq k] =  \eps^{\Omega(m)}$.
\end{Lemma}

Finally, let us consider the random variable $Y$.

\begin{Lemma}\label{lemma:total_iters_rand_col}
    $\P\left[Y \geq \frac{50k}{\eps^2(2-\eps)}\right] = \eps^{\Omega(m)}$.
\end{Lemma}

\begin{proof}
    To assist with our proof, we define the following for $i \in [k]$:
    \[\tilde Y_i = \left\{\begin{array}{cc}
        Y_i & i \leq K; \\
        0 & \text{otherwise.}
    \end{array}\right., \qquad \tilde Y = \sum_{j = 1}^k\tilde Y_j, \qquad y = \frac{50k}{\eps^2(2-\eps)}.\]
    As in the proof of Lemma~\ref{lemma:total_fail_rand_fan}, we note that $\tilde{Y} = Y$ when $K \leq k$. Hence, by Lemma~\ref{lemma:total_num_rand_col}:
    \begin{align*}
        \P\left[Y \geq y\right] &= \P\left[Y \geq y,\, K \leq k\right] + \P\left[Y \geq y,\, K > k\right] \\
        &\leq \P\left[Y \geq y,\, K \leq k\right] + \P\left[K > k\right] \\
        &\leq \P\left[\tilde Y \geq y\right] + \eps^{\Omega(m)}.
    \end{align*}
    It is now sufficient to show that $\P\left[\tilde Y > y\right] = \eps^{\Omega(m)}$.

    To this end, let $X_1$, \ldots, $X_t$ be i.i.d. geometric random variables with probability of success $p \defeq \eps/2$.
    Take any $1 \leq i \leq k$ and let $y_i \geq 1$. Then 
    \[\P\left[\tilde Y_i \geq y_i \mid K < i\right] = 0 \leq (1 - p)^{y_i - 1}.\]
    Additionally, by Lemma~\ref{lemma:rand_color_runtime}, we note that
    \[\P\left[\tilde Y_i \geq y_i \mid K \geq i, \, \psi, x, \theta\right] \leq (1 - p)^{y_i - 1},\]
    where $(\psi, x, \theta)$ is the input to Algorithm~\ref{alg:rand_col} at the relevant step.
    As the expression on the right-hand side of the above inequalities is independent of $K$, $\psi$, $x$, $\theta$, we conclude that the random variable $X_i$ stochastically dominates $\tilde Y_i$.
    Furthermore, note that
    \[\mu \,=\, \frac{y\,p\,(1-p)}{6} \,=\, \frac{25k}{12\eps} \,\geq\, k/p.\]
    In particular, we may apply Proposition~\ref{prop: geometric} to get
    \[\P\left[\tilde Y > y\right] \,\leq\, \P\left[X > y\right] \,\leq\, \exp\left(-\frac{25k}{12\eps}\right) \,=\, \exp\left(-\Omega(t/\eps^3)\right),\]
    where the last step follows by the definition of $k$ and $\kmax$.
\end{proof}

Note that if $T \leq t$, $K \leq k$, and $Y \leq \frac{50k}{\eps^2(2-\eps)}$, the runtime of Algorithm~\ref{alg:seq} is $O(m\,\log(1/\eps)/\eps^4)$ as a result of \eqref{eqn: runtime}.
By Lemmas~\ref{lemma:T_bound}, \ref{lemma:total_fail_rand_fan}, and \ref{lemma:total_num_rand_col}, we have
\[\P\left[T > t \text{ or } K > k \text{ or } Y > \frac{50k}{\eps^2(2-\eps)}\right] = \eps^{\Omega(m)},\]
completing the proof of Theorem~\ref{theo:main_theo}.

\printbibliography

\end{document}